\documentclass[11pt, letterpaper]{article}
\usepackage{titlesec, booktabs}
\usepackage[T1]{fontenc}
\usepackage[utf8]{inputenc}
\usepackage{lmodern}
\usepackage{enumitem}
\usepackage{amsthm}
\usepackage{url}
\usepackage{comment}
\usepackage{authblk}
\usepackage{adjustbox}
\usepackage[dvipsnames]{xcolor}

\usepackage[hypertexnames=false]{hyperref}

\usepackage{mathtools}

\newtheorem{theorem}{Theorem}[section]
\newtheorem{lemma}{Lemma}[section]
\newtheorem{corollary}{Corollary}[section]
\newtheorem{proposition}{Proposition}[section]

\newtheorem{assumption}{Assumption}[section]

\usepackage{color}
\usepackage{epsfig}
\usepackage[section]{placeins}
\usepackage{graphicx}
\usepackage{amsmath}
\usepackage{amssymb}
\usepackage{amsfonts}
\usepackage{mathrsfs}
\usepackage{multirow}
\usepackage{array}
\usepackage{subfigure} 
\usepackage{chngcntr}
\usepackage{natbib} 
\usepackage{apalike}

\newsavebox{\theorembox}
\newsavebox{\lemmabox}
\newsavebox{\claimbox}
\newsavebox{\factbox}
\newsavebox{\corollarybox}
\newsavebox{\examplebox}
\newsavebox{\remarkbox}
\newsavebox{\assbox}
\newsavebox{\propositionbox}
\newsavebox{\problembox}
\newsavebox{\defbox}

\savebox{\theorembox}{\noindent\bf Theorem}
\savebox{\lemmabox}{\noindent\bf Lemma}
\savebox{\factbox}{\noindent\bf Fact}
\savebox{\corollarybox}{\noindent\bf Corollary}
\savebox{\examplebox}{\noindent\bf Example}
\savebox{\remarkbox}{\noindent\bf Remark}
\savebox{\assbox}{\noindent\bf Assumption}
\savebox{\propositionbox}{\noindent\bf Proposition}
\savebox{\problembox}{\noindent\bf Problem}
\savebox{\defbox}{\noindent\bf Definition}

\def\blackslug{\hbox{\hskip 1pt \vrule width 4pt height 8pt depth 1.5pt
\hskip 1pt}}

 \newtheorem{@remark}{\bf Remark}[section]
 \newenvironment{remark}{\begin{@remark}\rm}{\end{@remark}}

\newcommand{\argmin}{\mathop{\rm argmin}}

\def\blot{\quad {$\vcenter{\vbox{\hrule height.4pt
             \hbox{\vrule width.4pt height.9ex \kern.9ex \vrule
width.4pt}
             \hrule height.4pt}}$}}

\providecommand{\keywords}[1]
{
{  \small	
  {\textit{Keywords---}} #1}
}
\usepackage{cleveref}

\usepackage{setspace}
\numberwithin{equation}{section}

\begin{document}

\title{Latent community paths in VAR-type models via \\ dynamic directed spectral co-clustering}
\author[1]{Younghoon Kim}
\author[2]{Changryong Baek\footnote{Corresponding author. Department of Statistics, Sungkyunkwan University, 25-2, Sungkyunkwan-ro, Jongno-gu, Seoul, Korea 03063, crbaek@skku.edu}}
\affil[1]{Department of Mathematics, University of Alabama}
\affil[2]{Department of Statistics, Sungkyunkwan University}

\date{October 28, 2026}

\maketitle

\baselineskip17pt

\begin{abstract}
    This paper proposes a dynamic network framework for uncovering latent community trajectories in high-dimensional vector autoregressive (VAR)-type models. 
    The framework accommodates both periodic VAR (PVAR) models, which capture cyclical seasonal evolution, and generalized vector heterogeneous autoregressive (VHAR) models, which characterize structural transitions across ordered dependence horizons. 
    By embedding a degree-corrected stochastic co-blockmodel into the transition matrices of these VAR-type models, the proposed framework separates node-level sending and receiving roles and summarizes directional and temporal dependence in an interpretable low-dimensional representation. 
    The proposed method integrates transition matrix estimation, directed spectral co-clustering, and eigenvector smoothing to track how directional communities split, merge, or persist over time. 
    Non-asymptotic misclassification bounds are established for both PVAR- and generalized VHAR-based networks, and Monte Carlo experiments support the theoretical results. 
    Application to U.S. nonfarm payroll data using the PVAR-based framework reveals a recurrent business-centered core alongside more mobile, seasonally sensitive sectors. 
    Application to global stock volatility data using the VHAR-based framework identifies a compact U.S.-centered long-horizon block, a Europe-dominated developed-market core, and a more dynamic short-horizon reallocation of peripheral and bridge markets.
\end{abstract}


\newpage
\section{Introduction}
\label{se:intro}

Vector autoregressive (VAR) models provide a natural framework for describing dynamic dependence among multiple time series, where individual transition coefficients represent predictive influence in the sense of \citet{granger1969investigating}. As the dimension of the system increases, however, the number of such coefficients grows rapidly, making their direct interpretation increasingly difficult. This issue is particularly pronounced in high-dimensional regimes, where regularization methods are commonly employed to achieve feasible estimation by imposing sparsity \citep[e.g.,][]{basu2015regularized,brownlees2018realized,barigozzi2023fnets}. Although sparsity reduces model complexity, it still represents dependence through individual coefficients. Beyond coefficient-wise representation, identifying groups of variables with similar patterns of dynamic dependence can provide a more structured view of the system.

One approach is to incorporate network structure directly into time series models. Network autoregressive \cite[NAR;][]{zhu2017network} models, for example, condition on an observed or prespecified network and estimate own-lag, neighbor-lag, covariate, and related network effects. Several extensions of the NAR model have been considered, including network-based vector autoregressive and moving average models, within- and between-group heterogeneity, and time-varying connectivity \citep[see, e.g.,][]{yin2023general,knight2019generalised,krampe2019time,zhu2020grouped,zhu2023simultaneous,mantziou2023gnar}. When the network or group structure is unknown, however, the group memberships and model parameters are often estimated jointly through iterative procedures. Instead, in our framework, the directed weighted network and its memberships are latent and are recovered from the observed multivariate time series. Our objective is therefore latent directional-structure recovery rather than coefficient inference conditional on a supplied network.

An alternative approach is to construct graphs from estimated time series models and identify groups through clustering. The stochastic blockmodel \cite[SBM;][]{holland1983stochastic} is useful in this setting because its low-rank structure facilitates community recovery using clustering algorithms. For example, \citet{gudhmundsson2021detecting} and \citet{brownlees2022community} apply spectral clustering to symmetrized graph representations obtained from VAR-type transition matrices. More recently, \citet{gudhmundsson2025detecting} considered directional groups in ``senders'' and ``receivers'' communities using a stochastic co-blockmodel \cite[ScBM;][]{rohe2016co}. A sending community consists of series whose lagged values have similar predictive effects on the rest of the system. In contrast, a receiving community consists of series that respond similarly to lagged variation elsewhere in the system. These approaches therefore group variables according to their Granger-causal dependence patterns.

However, a single network representation may be insufficient when the dependence structure varies systematically across seasons or across short-, medium-, and long-term effects. Variables may remain grouped under some temporal settings but split, merge, or move across groups under others. To capture such changes, we develop a dynamic network framework for \emph{latent community paths}: ordered trajectories of directional groups that describe how shock senders and receivers evolve across temporal regimes. These paths can reveal, for example, a persistent core of employment in nonfarm sectors together with seasonal reallocations of other sectors, or stable long-term groups of financial markets that reorganize over shorter horizons.

To implement this dynamic representation, we first extend the ScBM framework to a periodic VAR (PVAR) model \citep[see, e.g.,][]{ursu2009modelling,baek2017sparse,baek2018periodic} and a generalized vector heterogeneous autoregressive \cite[VHAR; see, e.g.,][]{corsi2009simple,baek2021sparse,lee2023volatility} model. The PVAR model represents cyclic dependence patterns across recurring seasons. In contrast, the generalized VHAR model represents transitions along short-, medium-, and long-run dependence effects, conventionally characterized by daily, weekly, and monthly aggregations. Accordingly, the extended ScBM representation is defined over ordered seasons in the PVAR model and ordered dependence horizons in the generalized VHAR model. For each season or horizon, we apply directed spectral co-clustering based on paired left and right singular vectors to identify transitions between sending and receiving groups across adjacent settings. To encourage smooth evolution of these community transitions, we further incorporate singular-vector smoothing in the spirit of \citet{liu2018global}. Together, these components provide a dynamic community-recovery framework for tracking changes in directional group structure across ordered temporal settings.

Our main theoretical contribution is to establish non-asymptotic misclassification bounds for directed community recovery in the ScBM--PVAR and generalized ScBM--VHAR models. The analysis combines estimation error bounds for regularized high-dimensional time series estimators \citep{basu2015regularized}, perturbation results for ScBMs \citep{qin2013regularized,rohe2016co}, and an extension of persistent communities by eigenvector smoothing  \cite[PisCES;][]{liu2018global} to sequences of left and right singular vectors. We also provide sufficient conditions connecting the required spectral assumptions to the underlying ScBM parameters. The theoretical results are supported by numerical experiments that assess finite-sample performance as well as robustness to model misspecification. Finally, applications based on the PVAR and generalized VHAR models illustrate how the proposed framework uncovers evolving directional group structures across seasons and dependence horizons.

The rest of the paper is organized as follows. Section~\ref{se:model} introduces the models and population networks. Section~\ref{se:algorithm} presents the estimators and the dynamic co-clustering procedure. Section~\ref{se:property} establishes the theoretical misclassification bounds. Sections~\ref{se:simulation} and \ref{se:appl} report simulation results and empirical applications, respectively. Section~\ref{se:conclusion} concludes.

\medskip
\noindent\textbf{Notation.} For a vector $v \in \mathbb{R}^d$, we denote its Euclidean norm by $\|v\| = \big(\sum_{i=1}^d v_i^2\big)^{1/2}$. For a matrix $M$, we write $\|M\|$ for the spectral norm, $\|M\|_F$ for the Frobenius norm, and $\|M\|_1$ and $\|M\|_{\infty}$ for the matrix $\ell_1$ and $\ell_{\infty}$ norms, that is, the maximum absolute column sum and maximum absolute row sum, respectively. Denote the spectral radius of $M$ by $\rho(M)$, its largest singular value by $\sigma_{\max}(M)$, its transpose by $M^\prime$, and its trace by $\mathrm{Tr}(M)$. We denote by $\mathbb{Z}$ the set of integers, by $\mathbb{N}$ the set of positive integers, and by $\mathbb{N}_0 := \mathbb{N}\cup\{0\}$ the set of nonnegative integers. Finally, big-$\mathcal{O}$ and big-$\Omega$ denote deterministic asymptotic upper and lower bounds, respectively, while $\mathcal{O}_{\mathbb{P}}$ denotes stochastic boundedness in probability. We write $a_n \asymp b_n$ if there exist constants $0 < c < C < \infty$ such that $c\, b_n \le a_n \le C\, b_n$ for all sufficiently large $n$.

\section{Model description}
\label{se:model}

We introduce three ScBM-based specifications for multivariate autoregressions. ScBM--VAR is the basic building block; ScBM--PVAR and generalized ScBM--VHAR extend it to seasonal and ordered-horizon settings.

\subsection{ScBM--VAR model}
\label{sse:model-ScBM-VAR}

For each lag $h=1,\ldots,p$, let
$\mathcal{G}_h=(\mathcal{V},\mathcal{E}_h,\mathcal{W}_h)$ be a directed weighted graph with common node set $\mathcal{V}=\{1,\ldots,q\}$, edge set $\mathcal{E}_h\subset\mathcal{V}\times\mathcal{V}$, and weights
$\mathcal{W}_h=\{[w_h]_{ij}\in(0,\infty):(i,j)\in\mathcal{E}_h\}$. Its adjacency matrix $A_h\in\mathbb{R}^{q\times q}$ satisfies
\begin{displaymath}
    [A_h]_{ij}=
    \begin{cases}
    [w_h]_{ij}, & (i,j)\in\mathcal{E}_h,\\
    0, & (i,j)\notin\mathcal{E}_h,
    \end{cases}
\end{displaymath}
where the nonzero weights obey $0<[w_h]_{ij}\leq w_{\max}$ almost surely for some finite $w_{\max}>0$.

We allow $K_y$ sending communities and $K_z$ receiving communities. Their membership indicator matrices are
\begin{displaymath}
    \Xi^y\in\{0,1\}^{q\times K_y},
    \qquad
    \Xi^z\in\{0,1\}^{q\times K_z}.
\end{displaymath}
Let $\xi_i^y\in\{1,\ldots,K_y\}$ and $\xi_j^z\in\{1,\ldots,K_z\}$ denote the corresponding sending and receiving labels. The diagonal propensity matrices $\Theta^y,\Theta^z\in\mathbb{R}^{q\times q}$ have nonnegative diagonal entries, and $B\in[0,1]^{K_y\times K_z}$ is the community link matrix. Throughout, we use the full rectangular-rank specification $\operatorname{rank}(B)=\min(K_y,K_z)$. Under the ScBM,
\begin{displaymath}
    \mathbb{P}\bigl([A_h]_{ij}\neq0\bigr)
    =
    [\Theta^y]_{ii}[\Theta^z]_{jj}[B]_{\xi_i^y\xi_j^z}.
\end{displaymath}
Conditional on the ScBM parameters, the weighted dyads are independent across ordered pairs $(i,j)$. For identification, we normalize the degree propensities within each directional block:
\begin{displaymath}
    \sum_{i:\xi_i^y=k}[\Theta^y]_{ii}=1,
    \quad k=1,\ldots,K_y,
    \qquad
    \sum_{j:\xi_j^z=\ell}[\Theta^z]_{jj}=1,
    \quad \ell=1,\ldots,K_z.
\end{displaymath}
The corresponding population adjacency matrix is
\begin{equation}\label{eq:adjacency_var}
    \mathcal{A}_h:=\mathbb{E}[A_h]
    =
    \mu_h\Theta^y\Xi^y B(\Xi^z)^\prime\Theta^z,
\end{equation}
for some scale parameter $\mu_h>0$. Define the population weighted out- and in-degree matrices and the population regularizer by
\begin{displaymath}
    \mathcal O_h:=\operatorname{diag}(\mathcal A_h\mathbf 1_q),
    \qquad
    \mathcal P_h:=\operatorname{diag}(\mathcal A_h^\prime\mathbf 1_q),
    \qquad
    \bar\tau_h:=q^{-1}\mathbf 1_q^\prime\mathcal A_h\mathbf 1_q.
\end{displaymath}
We assume that the directional memberships and block organization are common across lags, while the realized edge weights may be lag-specific.

Denote $\{Y_t\}_{t=1}^T$ an observed mean-zero time series with $Y_t=(Y_{1t},\ldots,Y_{qt})^\prime\in\mathbb{R}^q$. We consider the VAR$(p)$ model 
\begin{equation}\label{eq:VAR}
    Y_t
    =
    \sum_{h=1}^p\Phi_hY_{t-h}+\varepsilon_t,
    \qquad
    \varepsilon_t\sim\mathrm{WN}(0,\Sigma_\varepsilon),
    \qquad
    t=1,\ldots,T,
\end{equation}
where $\Phi_h\in\mathbb{R}^{q\times q}$, $\varepsilon_t\in\mathbb{R}^q$, and $\Sigma_\varepsilon\in\mathbb{R}^{q\times q}$. Here a directed edge $i\to j$ at lag $h$ represents predictive influence from $Y_{i,t-h}$ to $Y_{jt}$; equivalently, the corresponding VAR coefficient is $[\Phi_h]_{ji}$.

The graph and its weighted adjacency matrix are latent rather than separately observable. Specifically, the random network \eqref{eq:adjacency_var} is embedded in each
transition matrix through
\begin{equation}\label{eq:trans_VAR}
    \Phi_h
    =
    \phi_h\bigl(P_h^{\tau_{A,h}}\bigr)^{-1/2}
    A_h^\prime
    \bigl(O_h^{\tau_{A,h}}\bigr)^{-1/2},
    \qquad h=1,\ldots,p,
\end{equation}
where $\phi_h\in\mathbb{R}$ is a scale parameter. Let
\begin{displaymath}
    O_h:=\operatorname{diag}(A_h\mathbf 1_q),
    \qquad
    P_h:=\operatorname{diag}(A_h^\prime\mathbf 1_q),
    \qquad
    \tau_{A,h}:=q^{-1}\mathbf 1_q^\prime A_h\mathbf 1_q,
\end{displaymath}
where $O_h$ and $P_h$ are the realized weighted out- and in-degree matrices. Define
\begin{displaymath}
    O_h^{\tau_{A,h}}:=O_h+\tau_{A,h}I_q,
    \qquad
    P_h^{\tau_{A,h}}:=P_h+\tau_{A,h}I_q.
\end{displaymath}
Thus
\begin{equation}\label{eq:laplace_ScBM-VAR}
    L_h:=(O_h^{\tau_{A,h}})^{-1/2}A_h(P_h^{\tau_{A,h}})^{-1/2},
    \qquad
    \Phi_h=\phi_h L_h^\prime.
\end{equation}
The regularizer improves spectral recovery under degree heterogeneity \citep[e.g.,][]{chaudhuri2012spectral,qin2013regularized}. Note that equation \eqref{eq:trans_VAR} is a latent structural representation; estimation of $\Phi_h$ uses the observed time series and does not require $A_h$ to be observed. The representation \eqref{eq:laplace_ScBM-VAR} is particularly useful as it connects the random network and time series model. This row--column pairing is the standard normalization for a directed graph: $O_h$ normalizes senders on the left of $A_h$, and $P_h$ normalizes receivers on the right. Thus, an edge $i\to j$ is represented by the coefficient $[\Phi_h]_{ji}$. The scalars $\{\phi_h\}_{h=1}^p$ are chosen so that
\begin{displaymath}
    \det\bigl(I_q-\Phi_1z-\cdots-\Phi_pz^p\bigr)\neq0,
    \qquad |z|\leq1.
\end{displaymath}

\subsection{ScBM--PVAR model}
\label{sse:model-ScBM-PVAR}

For season $m=1,\ldots,s$ and lag $h=1,\ldots,p_m$, let
$\mathcal{G}_{h,m}=(\mathcal{V},\mathcal{E}_{h,m},\mathcal{W}_{h,m})$ have adjacency matrix $A_{h,m}\in\mathbb{R}^{q\times q}$. Its nonzero weights satisfy $0<[w_{h,m}]_{ij}\leq w_{\max,m}$ almost surely for some finite $w_{\max,m}>0$, uniformly over $h$. Within season $m$, the lag-specific graphs share the sending and receiving
membership matrices
\begin{displaymath}
    \Xi_m^y\in\{0,1\}^{q\times K_{y_m}},
    \qquad
    \Xi_m^z\in\{0,1\}^{q\times K_{z_m}}.
\end{displaymath}

Let $\xi_{m,i}^y\in\{1,\ldots,K_{y_m}\}$ and $\xi_{m,j}^z\in\{1,\ldots,K_{z_m}\}$ denote the corresponding node labels. The
lag-specific graphs also share diagonal propensity matrices
$\Theta_m^y,\Theta_m^z\in\mathbb{R}^{q\times q}$ and a community link matrix $B_m\in[0,1]^{K_{y_m}\times K_{z_m}}$, with $\operatorname{rank}(B_m)=\min(K_{y_m},K_{z_m})$. Their population adjacency matrices take the form
\begin{displaymath}
    \mathcal{A}_{h,m}
    :=
    \mathbb{E}[A_{h,m}]
    =
    \mu_{h,m}\Theta_m^y\Xi_m^y B_m(\Xi_m^z)^\prime\Theta_m^z,
\end{displaymath}
where $\mu_{h,m}>0$. Define
\begin{displaymath}
    \mathcal O_{h,m}:=\operatorname{diag}(\mathcal A_{h,m}\mathbf 1_q),
    \qquad
    \mathcal P_{h,m}:=\operatorname{diag}(\mathcal A_{h,m}^\prime\mathbf 1_q),
    \qquad
    \bar\tau_{h,m}:=q^{-1}\mathbf 1_q^\prime\mathcal A_{h,m}\mathbf 1_q.
\end{displaymath}

We link adjacent seasonal partitions through
\begin{equation}\label{eq:season_cyclic}
    \Xi_{m-1}^z=\Xi_m^y,
    \quad m=2,\ldots,s,
    \qquad
    \Xi_s^z=\Xi_1^y.
\end{equation}
Accordingly, $K_{z_{m-1}}=K_{y_m}$ for $m=2,\ldots,s$ and
$K_{z_s}=K_{y_1}$. Equivalently, $\Xi_m^z=\Xi_{m+1}^y$ for $m=1,\ldots,s-1$, with the final receiving partition linked back to the first sending partition. Thus the cyclically-linked states are
\begin{equation}\label{eq:community_path_PVAR}
    \cdots
    \;\longrightarrow\;
    \Xi_{s}^z=\Xi_{1}^y
    \;\longrightarrow\;
    \Xi_{1}^z=\Xi_{2}^y
    \;\longrightarrow\;
    \cdots
    \;\longrightarrow\;
    \Xi_{s-1}^z=\Xi_{s}^y
    \;\longrightarrow\;
    \Xi_{s}^z=\Xi_{1}^y
    \;\longrightarrow\;
    \cdots.
\end{equation}
Equation~\eqref{eq:season_cyclic}, leading to \eqref{eq:community_path_PVAR}, is a path-continuity restriction used to define a coherent ordered community path; it is not implied by a generic PVAR specification or by VAR stability. The current state produced in season $m$ subsequently becomes a lagged input to season $m+1$, which motivates carrying the receiving organization at season $m$ into the next sending organization. The equality itself remains a parsimonious modeling assumption. It does not require the partition to remain fixed: the transition from $\Xi_m^y$ to $\Xi_m^z$ may split, merge, or recompose communities within season $m$, and a time-invariant partition is included only as a special case. The restriction is most appropriate when seasonal dependence structures display persistence and recurrent reorganization. Abrupt changes that do not admit adjacent-stage matching would require a less restrictive formulation. Sending and receiving refer to directional roles within each seasonal transition.

Let $m=m(t)\in\{1,\ldots,s\}$ denote the season containing observation $t$. For mean-zero time series $\{Y_t\}_{t=1}^T$, the PVAR model is
\begin{equation}\label{eq:PVAR}
    Y_t
    =
    \sum_{h=1}^{p_m}\Phi_{h,m}Y_{t-h}+\varepsilon_{t,m},
    \qquad
    \varepsilon_{t,m}\sim\mathrm{WN}(0,\Sigma_{\varepsilon,m}),
\end{equation}
where $\Phi_{h,m}\in\mathbb{R}^{q\times q}$, $\varepsilon_{t,m}\in\mathbb{R}^q$, and $\Sigma_{\varepsilon,m}\in\mathbb{R}^{q\times q}$. By writing $t=m+ns$ where $n\in\{0,1,\ldots,N-1\}$ indexes the seasonal cycle, with $N=T/s$ denoting the number of complete seasonal cycles, \eqref{eq:PVAR} becomes
\begin{equation}\label{eq:PVAR-cycles}
    Y_{m+ns}
    =
    \sum_{h=1}^{p_m}\Phi_{h,m}Y_{m+ns-h}+\varepsilon_{m+ns,m}.
\end{equation}
The seasonal transition matrices are
\begin{displaymath}
    \Phi_{h,m}
    =
    \phi_{h,m}
    (P_{h,m}^{\tau_{A,h,m}})^{-1/2}
    A_{h,m}^\prime
    (O_{h,m}^{\tau_{A,h,m}})^{-1/2},
    \qquad
    h=1,\ldots,p_m,\quad m=1,\ldots,s,
\end{displaymath}
where $\phi_{h,m}\in\mathbb{R}$. Let
\begin{displaymath}
    O_{h,m}:=\operatorname{diag}(A_{h,m}\mathbf 1_q),
    \qquad
    P_{h,m}:=\operatorname{diag}(A_{h,m}^\prime\mathbf 1_q),
    \qquad
    \tau_{A,h,m}:=q^{-1}\mathbf 1_q^\prime A_{h,m}\mathbf 1_q,
\end{displaymath}
and set
\begin{displaymath}
    O_{h,m}^{\tau_{A,h,m}}:=O_{h,m}+\tau_{A,h,m}I_q,
    \qquad
    P_{h,m}^{\tau_{A,h,m}}:=P_{h,m}+\tau_{A,h,m}I_q.
\end{displaymath}
Equivalently, 
\begin{equation}\label{eq:laplace_ScBM-PVAR}
    L_{h,m}:=(O_{h,m}^{\tau_{A,h,m}})^{-1/2}A_{h,m}
    (P_{h,m}^{\tau_{A,h,m}})^{-1/2},
    \qquad
    \Phi_{h,m}=\phi_{h,m}L_{h,m}^\prime.
\end{equation}
The scales $\{\phi_{h,m}\}$ are chosen so that the periodic process is stable.

For estimation and theory, assume $T=Ns$ from the representation \eqref{eq:PVAR-cycles}, so that the sample contains $N$ complete seasonal cycles, and define
\begin{displaymath}
    Y_n^*
    =
    \bigl(Y_{ns+s}^\prime,Y_{ns+s-1}^\prime,\ldots,Y_{ns+1}^\prime\bigr)^\prime
    \in\mathbb{R}^{qs}, 
    \qquad
    \varepsilon_n^*
    =
    \bigl(\varepsilon_{ns+s}^\prime,\varepsilon_{ns+s-1}^\prime,\ldots,
    \varepsilon_{ns+1}^\prime\bigr)^\prime
    \in\mathbb{R}^{qs}.
\end{displaymath}
Because the seasonal recursion links $Y_{ns+m}$ to $Y_{ns+m-1}$ within the same
cycle, the stacked process first satisfies
\begin{displaymath}
    G_0Y_n^*
    =
    \sum_{h=1}^{p^*}G_hY_{n-h}^*+\varepsilon_n^*,
\end{displaymath}
where $G_0\in\mathbb{R}^{qs\times qs}$ collects the within-cycle contemporaneous
links and is block triangular with identity diagonal blocks, hence always
invertible. Premultiplying by $G_0^{-1}$ gives the VAR representation
\begin{equation}\label{eq:PVAR_period_cycle_simple}
    Y_n^*
    =
    \sum_{h=1}^{p^*}\Psi_h^*Y_{n-h}^*+u_n^*,
    \qquad
    \Psi_h^*=G_0^{-1}G_h\in\mathbb{R}^{qs\times qs},
    \qquad
    u_n^*=G_0^{-1}\varepsilon_n^*,
\end{equation}
where $p^*=\lceil(\max_{m=1,\ldots,s}p_m)/s\rceil$ and $\Psi_h^*$ is determined by
the original seasonal coefficient matrices; see Section~2 of
\citet{ursu2009modelling}. Since $G_0$ has identity diagonal blocks and the
seasonal coefficient matrices are uniformly bounded under stability, $\|G_0\|$ and
$\|G_0^{-1}\|$ are $O(1)$, and $\{u_n^*\}$ inherits the serial uncorrelatedness of
$\{\varepsilon_n^*\}$. When $s=1$, $G_0=I_q$ and this reduces to the ScBM--VAR
model.

\subsection{Generalized ScBM--VHAR model}
\label{sse:model-ScBM-VHAR}

Let $1<b_M<b_L$ be fixed integers. The generalized VHAR model is
\begin{equation}\label{eq:VHAR}
    Y_t
    =
    \Phi_{(s)}Y_{t-1}^{(s)}
    +
    \Phi_{(m)}Y_{t-1}^{(m)}
    +
    \Phi_{(\ell)}Y_{t-1}^{(\ell)}
    +
    \varepsilon_t,
    \qquad
    \varepsilon_t\sim\mathrm{WN}(0,\Sigma_\varepsilon),
\end{equation}
where $\Phi_{(s)},\Phi_{(m)},\Phi_{(\ell)}\in\mathbb{R}^{q\times q}$ and
\begin{displaymath}
    Y_t^{(s)}=Y_t,
    \qquad
    Y_t^{(m)}
    =
    \frac{1}{b_M}\sum_{h=0}^{b_M-1}Y_{t-h},
    \qquad
    Y_t^{(\ell)}
    =
    \frac{1}{b_L}\sum_{h=0}^{b_L-1}Y_{t-h},
\end{displaymath}
with all three aggregates in $\mathbb{R}^q$. The effective sample size is $N=T-b_L$. When $(b_M,b_L)=(5,22)$, \eqref{eq:VHAR} reduces to the conventional daily--weekly--monthly specification \citep{corsi2009simple}.

Equation \eqref{eq:VHAR} is equivalent to a constrained VAR$(b_L)$. If $\Phi_1,\ldots,\Phi_{b_L}\in\mathbb{R}^{q\times q}$ denote its implied transition matrices, then
\begin{equation}\label{eq:VHAR_coefficients}
    \Phi_1
    =
    \Phi_{(s)}+\frac{\Phi_{(m)}}{b_M}+\frac{\Phi_{(\ell)}}{b_L},
    \quad
    \Phi_2=\cdots=\Phi_{b_M}
    =
    \frac{\Phi_{(m)}}{b_M}+\frac{\Phi_{(\ell)}}{b_L},
    \quad
    \Phi_{b_M+1}=\cdots=\Phi_{b_L}
    =
    \frac{\Phi_{(\ell)}}{b_L}
\end{equation}
Thus the number of free coefficients is reduced from $b_Lq^2$ to $3q^2$.

For each horizon $h\in\{s,m,\ell\}$, let $A_{(h)}\in\mathbb{R}^{q\times q}$ be the latent weighted adjacency matrix of $\mathcal{G}_{(h)}=(\mathcal{V},\mathcal{E}_{(h)},\mathcal{W}_{(h)})$. Its nonzero
weights satisfy $0<[w_{(h)}]_{ij}\leq w_{\max,(h)}$ almost surely for some finite $w_{\max,(h)}>0$. We parameterize
\begin{displaymath}
    \Phi_{(h)}
    =
    \phi_{(h)}
    \bigl(P_{(h)}^{\tau_{A,(h)}}\bigr)^{-1/2}
    A_{(h)}^\prime
    \bigl(O_{(h)}^{\tau_{A,(h)}}\bigr)^{-1/2},
    \qquad h\in\{s,m,\ell\},
\end{displaymath}
where $\phi_{(h)}\in\mathbb{R}$. Let
\begin{displaymath}
    O_{(h)}:=\operatorname{diag}(A_{(h)}\mathbf 1_q),
    \qquad
    P_{(h)}:=\operatorname{diag}(A_{(h)}^\prime\mathbf 1_q),
    \qquad
    \tau_{A,(h)}:=q^{-1}\mathbf 1_q^\prime A_{(h)}\mathbf 1_q,
\end{displaymath}
and set
\begin{displaymath}
    O_{(h)}^{\tau_{A,(h)}}:=O_{(h)}+\tau_{A,(h)}I_q,
    \qquad
    P_{(h)}^{\tau_{A,(h)}}:=P_{(h)}+\tau_{A,(h)}I_q.
\end{displaymath}
Similar to \eqref{eq:laplace_ScBM-PVAR}, write
\begin{equation}\label{eq:laplace_ScBM-VHAR}
    L_{(h)}:=(O_{(h)}^{\tau_{A,(h)}})^{-1/2}A_{(h)}
    (P_{(h)}^{\tau_{A,(h)}})^{-1/2},
    \qquad
    \Phi_{(h)}=\phi_{(h)}L_{(h)}^\prime.
\end{equation}

At horizon $h$, let
\begin{displaymath}
    \Xi_{(h)}^y\in\{0,1\}^{q\times K_{y_{(h)}}},
    \qquad
    \Xi_{(h)}^z\in\{0,1\}^{q\times K_{z_{(h)}}}
\end{displaymath}
be the sending and receiving membership matrices, and let
$\xi_{(h),i}^y\in\{1,\ldots,K_{y_{(h)}}\}$ and
$\xi_{(h),j}^z\in\{1,\ldots,K_{z_{(h)}}\}$ denote their node labels. Let $\Theta_{(h)}^y,\Theta_{(h)}^z\in\mathbb{R}^{q\times q}$ be the diagonal propensity matrices and let
$B_{(h)}\in[0,1]^{K_{y_{(h)}}\times K_{z_{(h)}}}$ be the block link matrix, with $\operatorname{rank}(B_{(h)})=\min(K_{y_{(h)}},K_{z_{(h)}})$. The population adjacency matrix is
\begin{displaymath}
    \mathcal{A}_{(h)}
    =
    \mu_{(h)}
    \Theta_{(h)}^y
    \Xi_{(h)}^y
    B_{(h)}
    (\Xi_{(h)}^z)^\prime
    \Theta_{(h)}^z,
\end{displaymath}
where $\mu_{(h)}>0$. Define
\begin{displaymath}
    \mathcal O_{(h)}:=\operatorname{diag}(\mathcal A_{(h)}\mathbf 1_q),
    \qquad
    \mathcal P_{(h)}:=\operatorname{diag}(\mathcal A_{(h)}^\prime\mathbf 1_q),
    \qquad
    \bar\tau_{(h)}:=q^{-1}\mathbf 1_q^\prime\mathcal A_{(h)}\mathbf 1_q.
\end{displaymath}

The horizons are interpreted in the natural long-to-short order.  We reserve $L/R$ for left/right singular-space orientation and use the lower-case symbols $s,m,\ell$ for short, medium, and long horizons. Under our row--column convention,
\begin{displaymath}
    L\equiv\text{left singular space}\equiv\Xi^y\equiv K_y,
    \qquad
    R\equiv\text{right singular space}\equiv\Xi^z\equiv K_z.
\end{displaymath}
The generic horizon-continuity restriction is
\begin{displaymath}
    \Xi_{(\ell)}^z=\Xi_{(m)}^y,
    \qquad
    \Xi_{(m)}^z=\Xi_{(s)}^y.
\end{displaymath}
Thus the long-to-short linked states are
\begin{equation}\label{eq:community_path_VHAR}
    \Xi_{(\ell)}^y
    \;\longrightarrow\;
    \Xi_{(\ell)}^z=\Xi_{(m)}^y
    \;\longrightarrow\;
    \Xi_{(m)}^z=\Xi_{(s)}^y
    \;\longrightarrow\;
    \Xi_{(s)}^z.
\end{equation}
The equalities link adjacent horizon-specific directional roles but do not require the two roles within a given horizon to coincide. This produces a hierarchical flow from broader horizons to shorter horizons \citep{muller1997volatilities,corsi2009simple}.

\section{Estimation, spectral co-clustering and smoothing transitions}
\label{se:algorithm}

This section presents the estimation-clustering-smoothing pipeline for the models described in Section \ref{se:model}, separating transition-matrix estimation from community recovery and smoothing. 

\subsection{Estimation of transition matrices}
\label{sse:algorithm_estimation}

As the first step, we estimate the transition matrices via ordinary least squares (OLS) as the baseline, particularly in low-dimensional settings. In higher dimensions, we use lasso estimation.

\medskip\noindent\textbf{ScBM--VAR and ScBM--PVAR.}
For the ScBM--VAR model, either OLS or lasso estimation \citep{basu2015regularized} is applied directly to the regression in \eqref{eq:VAR}. For the ScBM--PVAR model, we use the stacked representation in \eqref{eq:PVAR_period_cycle_simple}. Define $\mathbb Y_P=(Y_{p^*+1}^{*\prime},\ldots,Y_N^{*\prime})^\prime
    \in\mathbb{R}^{(N-p^*)\times qs}$,
\begin{displaymath}
    W_n^*=(Y_{n-1}^{*\prime},\ldots,Y_{n-p^*}^{*\prime})^\prime
    \in\mathbb{R}^{qsp^*},
    \qquad
    \mathbb X_P=(W_{p^*+1}^{*\prime},\ldots,W_N^{*\prime})^\prime
    \in\mathbb{R}^{(N-p^*)\times qsp^*},
\end{displaymath}
$\mathbb E_P=(u_{p^*+1}^{*\prime},\ldots,u_N^{*\prime})^\prime\in\mathbb{R}^{(N-p^*)\times qs}$, and $\alpha_P^*:=\operatorname{vec}(B_P) \in\mathbb{R}^{q^2s^2p^*},B_P=(\Psi_1^{*\prime},\ldots,\Psi_{p^*}^{*\prime})^\prime \in\mathbb{R}^{qsp^*\times qs}$. Then
\begin{displaymath}
    \mathbb Y_P=\mathbb X_PB_P+\mathbb E_P,
    \qquad
    \operatorname{vec}(\mathbb Y_P)
    =(I_{qs}\otimes\mathbb X_P)\alpha_P^*+\operatorname{vec}(\mathbb E_P).
\end{displaymath}
The lasso estimator for the ScBM--PVAR model is
\begin{equation}\label{eq:pvar_lasso}
	\hat{\alpha}_P
	=
	\argmin_{\alpha\in\mathbb{R}^{q^2s^2p^*}}
	\left\{
	\frac{1}{N}
	\left\|
	\mathrm{vec}(\mathbb{Y}_P)
	-
	(I_{qs}\otimes\mathbb X_P)\alpha
	\right\|_2^2
	+
	\lambda_{N,P}\|\alpha\|_1
	\right\},
\end{equation}
where $\lambda_{N,P}>0$ is a regularization parameter. The OLS estimator is obtained by removing the $\ell_1$ penalty in \eqref{eq:pvar_lasso}.

\medskip\noindent\textbf{Generalized ScBM--VHAR.}
For the generalized ScBM--VHAR model, define
\begin{displaymath}
    Z = (Y_{b_L+1}', \ldots, Y_T')' \in \mathbb{R}^{N\times q},
    \qquad
    B =
    \begin{pmatrix}
    \Phi_{(s)}'\\
    \Phi_{(m)}'\\
    \Phi_{(\ell)}'
    \end{pmatrix}
    \in \mathbb{R}^{3q\times q},
    \qquad
    \beta_V^* = \operatorname{vec}(B)\in\mathbb{R}^{3q^2},
\end{displaymath}
with $N=T-b_L$. Let $X \in \mathbb{R}^{N\times qb_L}$ be the unrestricted VAR$(b_L)$ lag-design matrix, and set
\begin{displaymath}
    X_e = X R_{b_M,b_L} \in \mathbb{R}^{N\times 3q},
\end{displaymath}
where $R_{b_M,b_L} \in \mathbb{R}^{qb_L\times 3q}$ is the deterministic restriction matrix implied by \eqref{eq:VHAR_coefficients}. Then
\begin{displaymath}
    Z = X_e B + E,
    \qquad E\in\mathbb{R}^{N\times q},
    \qquad
    \operatorname{vec}(Z) = (I_q \otimes X_e)\beta_V^* + \operatorname{vec}(E).
\end{displaymath}
The restricted lasso estimator is
\begin{equation} \label{eq:vhar_lasso}
    \hat\beta_V
    =
    \argmin_{\beta\in\mathbb{R}^{3q^2}}
    \left\{
    \frac{1}{N}
    \left\|
    \operatorname{vec}(Z) - (I_q\otimes X_e)\beta
    \right\|_2^2
    +
    \lambda_{N,V}\|\beta\|_1
    \right\},
\end{equation}
where $\lambda_{N,V}>0$ is a regularization parameter. The restricted OLS estimator is obtained by dropping the $\ell_1$ penalty.

\subsection{Spectral co-clustering for ScBM--VAR-type models}
\label{sse:algorithm_co}

The co-clustering step has the same backbone in all models. Given the estimated autoregressive matrices, we compute their leading left and right singular vectors, row-normalize the resulting embeddings, optionally smooth the associated projectors across ordered seasons or horizons, and apply $k$-means to the rows of the normalized embeddings. The models differ only in how the stages are ordered and in which embeddings are matched before clustering.

The retained spectral dimension is common to the two sides. At each season
or horizon, let $r_m:=\operatorname{rank}(B_m)$; the population transition
matrix has the same nonzero singular rank. We work throughout with the full
rectangular-rank specification, under which
\begin{displaymath}
    r_m=\min(K_{y_m},K_{z_m}).
\end{displaymath}
Thus the left and right spectral embeddings at a given stage have the
\emph{same} dimension $r_m$, whereas the numbers of sending and receiving
clusters remain $K_{y_m}$ and $K_{z_m}$. Throughout the technical
development, $L/R$ mean left/right singular spaces, with
$L\leftrightarrow\Xi^y$ and $R\leftrightarrow\Xi^z$; sending/receiving are
used only as directional interpretations of those two spaces. This
distinction is essential when the two directional community counts differ:
the embeddings share the dimension $r_m$, while $k$-means recovers
$K_{y_m}$ and $K_{z_m}$ row centroids within them.

\medskip\noindent\textbf{ScBM--VAR.}
Combining directed co-clustering \citep{rohe2016co} with the VAR Blockbuster
approach of \citet{gudhmundsson2021detecting}, let
\begin{displaymath}
    \widehat\Phi=\sum_{h=1}^p\widehat\Phi_h',
    \qquad
    r=\min(K_y,K_z),
\end{displaymath}
and compute the leading $r$ left and right singular-vector matrices
$\widehat X_L,\widehat X_R\in\mathbb R^{q\times r}$. We use the
row-normalization convention
\begin{equation}\label{eq:normalization}
    \mathfrak n(x)=
    \begin{cases}
    x/\|x\|, & x\neq0,\\
    0, & x=0,
    \end{cases}
    \qquad
    [\widehat X_\bullet^*]_{i\cdot}=\mathfrak n([\widehat X_\bullet]_{i\cdot}),
    \quad \bullet\in\{L,R\}.
\end{equation}
The directional memberships are then obtained by $k$-means on the normalized
rows: we apply $k$-means with $K_y$ clusters to the rows of $\widehat X_L^*$ to recover the sending communities, and with $K_z$ clusters to the rows of $\widehat X_R^*$ to recover the receiving communities. When a model restriction identifies the two memberships, $k$-means is instead
applied jointly to the concatenated rows $(\widehat X_L^*\ \widehat X_R^*)\in\mathbb R^{q\times 2r}$ with the common number of clusters, so that a single label is assigned to each series.

\begin{remark}
The spectral dimension $r=\operatorname{rank}(B)$ and the numbers of
directional communities are different objects. Singular-value diagnostics
identify the retained common subspace, while $K_y$ and $K_z$ determine the
numbers of row centroids to be recovered by $k$-means within the left and
right embeddings.
\end{remark}

\medskip\noindent\textbf{ScBM--PVAR.}
For simplicity, take $p_m=p$ for $m=1,\ldots,s$ and form
\begin{displaymath}
    \widehat\Phi_m=\sum_{h=1}^p\widehat\Phi_{m,h}',
    \qquad
    r_m=\min(K_{y_m},K_{z_m}).
\end{displaymath}
Let $\widehat X_{m,L},\widehat X_{m,R}\in\mathbb R^{q\times r_m}$ contain the
leading left and right singular vectors and define
\begin{displaymath}
    \widehat U_{m,L}=\widehat X_{m,L}\widehat X_{m,L}',
    \qquad
    \widehat U_{m,R}=\widehat X_{m,R}\widehat X_{m,R}'.
\end{displaymath}
For either orientation $\bullet\in\{L,R\}$, initialize $\bar U_{m,\bullet}^{(0)}=\widehat U_{m,\bullet}$ and iterate
\begin{align*}
    \bar U_{1,\bullet}^{(l+1)}
    &=\Pi\!\left(\widehat U_{1,\bullet}+\alpha_N\bar U_{2,\bullet}^{(l)};r_1\right),\\
    \bar U_{m,\bullet}^{(l+1)}
    &=\Pi\!\left(\alpha_N\bar U_{m-1,\bullet}^{(l)}+\widehat U_{m,\bullet}+\alpha_N\bar U_{m+1,\bullet}^{(l)};r_m\right),
    \quad m=2,\ldots,s-1,\\
    \bar U_{s,\bullet}^{(l+1)}
    &=\Pi\!\left(\alpha_N\bar U_{s-1,\bullet}^{(l)}+\widehat U_{s,\bullet};r_s\right),
\end{align*}
where $\Pi(M;r)$ is the projector onto the leading $r$-dimensional eigenspace
of $M$. After convergence, write $\bar U_{m,\bullet}=\bar X_{m,\bullet}\bar X_{m,\bullet}'$ and row-normalize $\bar X_{m,\bullet}$ using \eqref{eq:normalization}.

The PisCES \citep{liu2018global} update smooths each orientation-specific
projector sequence along the linear ordered chain $m=1,\ldots,s$; it does not
impose same-side similarity between seasons $s$ and $1$. The cyclic
restriction in \eqref{eq:season_cyclic} enters instead through cross-side
membership matching at the clustering step: we apply $k$-means jointly to the
rows of $(\bar X_{m-1,R}^*\ \bar X_{m,L}^*)$ for $m=2,\ldots,s$, and to the
rows of $(\bar X_{s,R}^*\ \bar X_{1,L}^*)$ for the terminal cyclic pair. The
required numbers of clusters are $K_{z_{m-1}}=K_{y_m}$ and $K_{z_s}=K_{y_1}$,
respectively, so that each linked pair of directional roles receives a single
common membership. These are membership-continuity restrictions, not
singular-rank restrictions. When $\alpha_N=0$, the procedure reduces to
season-wise spectral co-clustering.

\begin{remark}\label{rem:num_comm}
The directional community counts are tuning inputs selected from
singular-value diagnostics subject to the membership-continuity restriction.
Once $(K_{y_m},K_{z_m})$ is fixed, the common retained rank is
$r_m=\min(K_{y_m},K_{z_m})$. The $k$-means step is run from multiple random
initializations and the best solution in terms of within-cluster sum of
squares is retained. Numerical labels are aligned only for presentation and
do not affect the recovered partitions. Section~\ref{ap:misspecification} in
the Supplementary Material reports sensitivity to deliberate under- and
over-specification of the community path.
\end{remark}

\medskip\noindent\textbf{ScBM--VHAR.}
For the generalized ScBM--VHAR model, order the horizon matrices from long to
short as
\begin{displaymath}
    (\widehat\Phi_1,\widehat\Phi_2,\widehat\Phi_3)
    =(\widehat\Phi_{(\ell)}',\widehat\Phi_{(m)}',\widehat\Phi_{(s)}').
\end{displaymath}
At horizon $h\in\{s,m,\ell\}$, both singular spaces use
$r_{(h)}=\min(K_{y_{(h)}},K_{z_{(h)}})$, and the two projector chains are
smoothed by the same rank-$r_{(h)}$ PisCES update. In the $L/R$ convention,
the generic long-to-short membership restriction is
\begin{equation}\label{eq:vhar-path-constraints}
    \Xi^z_{(\ell)}=\Xi^y_{(m)},
    \qquad
    \Xi^z_{(m)}=\Xi^y_{(s)}.
\end{equation}
Accordingly, after smoothing we apply $k$-means separately to the rows of
$\bar X_{(\ell),L}^*$ and $\bar X_{(s),R}^*$ as the two terminal states, and
apply joint $k$-means to $(\bar X_{(\ell),R}^*\ \bar X_{(m),L}^*)$ and
$(\bar X_{(m),R}^*\ \bar X_{(s),L}^*)$ for the two linked interior states.
Thus the generic four-state path is
$(\ell,L)\to(\ell,R)/(m,L)\to(m,R)/(s,L)\to(s,R)$. The three-state empirical
aggregation used in Section~\ref{sse:appl_vhar} imposes one additional
long-horizon equality and is stated there explicitly.

\subsection{Uniform projector cross-validation for the PisCES smoothing parameter}
\label{sse:algorithm_cv}

We select the smoothing parameter by repeated held-out reconstruction of the signed transition matrices. For repeat $a=1,\ldots,M$, draw $\mathcal H_{m}^{(a)}$ uniformly from the $q(q-1)$ off-diagonal entries of the oriented stage matrix $\widehat\Phi_m$ and set those entries to zero, producing $\widehat\Phi_{m}^{(\mathrm{tr},a)}$. The same held-out sets are used for all candidate values of $\alpha_N$ within a repeat.

From $\widehat\Phi_{m}^{(\mathrm{tr},a)}$, compute the rank-$r_m$ left and right projectors $\widehat U_{m,L}^{(a)}$ and $\widehat U_{m,R}^{(a)}$.  For each candidate $\alpha_N$, smooth the two projector sequences separately and denote the resulting projectors by $\bar U_{m,L}^{(\alpha_N,a)}$ and $\bar U_{m,R}^{(\alpha_N,a)}$.  The held-out signed entries are reconstructed by
\begin{displaymath}
    \widetilde\Phi_{m}^{(\alpha_N,a)}
    =
    \bar U_{m,L}^{(\alpha_N,a)}
    \widehat\Phi_{m}^{(\mathrm{tr},a)}
    \bar U_{m,R}^{(\alpha_N,a)}.
\end{displaymath}
We minimize
\begin{equation}\label{eq:cv}
    \mathrm{CV}(\alpha_N)
    =
    \frac{
    \displaystyle\sum_{a=1}^{M}\sum_{m=1}^{s}
    \sum_{(i,j)\in\mathcal H_{m}^{(a)}}
    \left\{[\widehat\Phi_m]_{ij}-[\widetilde\Phi_{m}^{(\alpha_N,a)}]_{ij}\right\}^{2}
    }{
    \displaystyle\sum_{a=1}^{M}\sum_{m=1}^{s}
    \sum_{(i,j)\in\mathcal H_{m}^{(a)}}[\widehat\Phi_m]_{ij}^{2}
    }.
\end{equation}
The denominator in \eqref{eq:cv} is common to all candidates and the squared errors are aggregated before taking the ratio. The criterion therefore preserves coefficient signs and evaluates directly whether smoothing improves held-out reconstruction of directional dependence. Because held-out entries are zeroed rather than inverse-probability rescaled, the resulting CV values are used only to rank candidate smoothing parameters within a design and are not interpreted as absolute out-of-sample risk estimates.

In all simulations and applications, we use masking fraction $\varpi=0.20$, $M=10$, and a grid containing zero and 20 logarithmically spaced positive values ending at $0.99\alpha_{\max}$, where $\alpha_{\max}=1/(4\sqrt2+2)$ as in \citet[Theorem~1]{liu2018global}.  For the contraction coefficient $\rho_{\mathrm P}(\alpha)=2\sqrt{2} \alpha/(1-2\alpha)$, this admissible range gives $\rho_{\mathrm P}(\alpha)<1/2$ and hence a uniform contraction margin.  The bounds below hold uniformly over the grid, so they remain valid after substituting the data-measurable selector $\widehat\alpha_N$. The population smoothing bias is controlled by $\alpha_N\mathcal V_{\bullet,N}$, $\bullet\in\{P,V\}$, where $\mathcal V_{\bullet,N}$ is the total variation of the adjacent population projector chain for the ScBM-PVAR model $(P)$ and the generalized ScBM-VHAR model $(V)$. Thus the present asymptotic argument yields vanishing projector bias whenever $\alpha_N\mathcal V_{\bullet,N}\to0$; the condition is automatic for a static projector chain. Cross-validation is a finite-sample rule and is not designed to impose that asymptotic sequence.

\section{Theoretical properties}
\label{se:property}

This section develops the theory for the ScBM--PVAR and generalized ScBM--VHAR procedures. The key idea is modularity: once an operator-norm bound for $\widehat{\Phi}-\Phi$ is available, the downstream co-clustering theory follows by substitution. All proofs are deferred to Supplementary Material~\ref{ap:appendix}.

We work with a sequence of models indexed by the effective sample size $N$. For the ScBM--PVAR model \eqref{eq:PVAR-cycles}, $N=T/s$ is the number of complete seasonal cycles; for the generalized ScBM--VHAR model \eqref{eq:VHAR}, $N=T-b_L$ is the number of usable observations. We allow $q\to\infty$ as $N\to\infty$, while $s$, $p^*$, $b_M$, and $b_L$ are fixed. The coefficient matrices, covariance matrices, and latent network parameters may therefore depend on $N$, although this dependence is suppressed to avoid excessive notation. Fixed-dimensional asymptotics are included as the special case in which $q$ does not vary.

Stability is defined for each finite-dimensional member of this sequence. Uniform stability conditions are then imposed so that the constants used in the subsequent high-dimensional bounds do not deteriorate as $q$ increases.

\subsection{Stability}
\label{sse:stability}

For the ScBM--PVAR model, the stacked representation
\eqref{eq:PVAR_period_cycle_simple} is a VAR$(p^*)$ process in $\mathbb{R}^{qs}$. We use the following member-wise characteristic-root condition and a uniform norm margin.

\begin{assumption}\label{assum:PVAR_stability_simple}
For every $N$, the stacked ScBM--PVAR representation satisfies:
\begin{itemize}
    \item[(i)] For all complex $z$ with $|z|\leq1$,
    \begin{displaymath}
        \det\!\Bigl(
        I_{qs}-\Psi_1^*z-\cdots-\Psi_{p^*}^*z^{p^*}
        \Bigr)\neq0.
    \end{displaymath}
    \item[(ii)] There exists a constant $\delta_P>0$, independent of $N$, and
    nonnegative envelope constants $a_{P,h,N}$ such that
    \begin{displaymath}
        \|\Psi_{h,N}^*\|\leq a_{P,h,N},
        \qquad h=1,\ldots,p^*,
        \qquad
        \sum_{h=1}^{p^*}a_{P,h,N}\leq1-\delta_P,
    \end{displaymath}
    uniformly over $N$.
\end{itemize}
\end{assumption}
Condition (i) is the standard finite-dimensional VAR stability condition. Condition (ii) strengthens it by placing the model sequence inside a common stability envelope. The constants $a_{P,h,N}$ bound the lag matrices in the stacked VAR and are distinct from the graph-scale parameters $\phi_{h,m}$ in Section~\ref{sse:model-ScBM-PVAR}. It is this uniform margin, rather than a fixed value of $q$, that prevents the dependence and covariance constants used below from degenerating.

\begin{lemma}\label{lem:stability}
Under Assumption~\ref{assum:PVAR_stability_simple}, each member of the ScBM--PVAR sequence is stable: the stacked process $\{Y_n^*\}$ admits a unique causal stationary solution.
\end{lemma}

The generalized ScBM--VHAR model in \eqref{eq:VHAR} is a constrained VAR$(b_L)$ with transition matrices $\{\Phi_h\}_{h=1}^{b_L}$ defined by \eqref{eq:VHAR_coefficients}. We impose the analogous conditions.

\begin{assumption}\label{assum:VHAR_stability}
For every $N$, the generalized ScBM--VHAR model with fixed integers
$1<b_M<b_L$ satisfies:
\begin{itemize}
    \item[(i)] For all complex $z$ with $|z|\leq1$,
    \begin{displaymath}
        \det\!\Bigl(
        I_q-\sum_{h=1}^{b_L}\Phi_hz^h
        \Bigr)\neq0.
    \end{displaymath}
    \item[(ii)] There exists a constant $\delta_V>0$, independent of $N$, and nonnegative envelope constants $a_{V,h,N}$ such that
    \begin{displaymath}
        \|\Phi_{h,N}\|\leq a_{V,h,N},
        \qquad h=1,\ldots,b_L,
        \qquad
        \sum_{h=1}^{b_L}a_{V,h,N}\leq1-\delta_V,
    \end{displaymath}
    uniformly over $N$.
    \end{itemize}
\end{assumption}

Condition (i) defines stability for each finite-dimensional constrained VAR system. Condition (ii) supplies a uniform norm envelope for the sequence. The constants $a_{V,h,N}$ bound the implied VAR$(b_L)$ lag matrices and are distinct from the horizon-specific graph scales $\phi_{(s)},\phi_{(m)},\phi_{(\ell)}$.

\begin{lemma}\label{cor:stability}
Under Assumption~\ref{assum:VHAR_stability}, each member of the generalized
ScBM--VHAR sequence admits a unique causal stationary solution.
\end{lemma}

\subsection{Consistency of estimators}
\label{sse:consistency_estimator}

All convergence statements in this subsection refer to the joint regime $N\to\infty$ with $q$ either fixed or diverging under the stated growth conditions. For the ScBM--PVAR model, define the season-stacked autoregressive matrix
\begin{displaymath}
    \Phi
    =
    \begin{pmatrix}
    \Phi_1\\
    \vdots\\
    \Phi_s
    \end{pmatrix}
    \in\mathbb{R}^{qs\times q}.
\end{displaymath}
Let $\widehat{\Phi}$ denote its sample estimator and let $\Phi$ denote the corresponding matrix conditional on the realized latent networks. Define $\widetilde{\Phi}$ as the deterministic population-normalized counterpart obtained by replacing, block by block, $(A,O,P,\tau_A)$ with $(\mathcal A,\mathcal O,\mathcal P,\bar\tau)$. Because the degree normalization is nonlinear, $\widetilde{\Phi}$ need not equal $\mathbb E[\Phi]$. Then
\begin{displaymath}
    \|\widehat{\Phi}-\widetilde{\Phi}\|
    \leq
    \underbrace{\|\widehat{\Phi}-\Phi\|}_{\text{time-series estimation error}}
    +
    \underbrace{\|\Phi-\widetilde{\Phi}\|}_{\text{network randomness}}.
\end{displaymath}
This subsection controls the first term; Section~\ref{sse:consistency_random} controls the second.

We first consider OLS. For the PVAR model, let
$\Sigma_Y^*:=\operatorname{Var}(Y_n^*)$ and
$\Sigma_u^*:=\operatorname{Var}(u_n^*)$ with $u_n^*=G_0^{-1}\varepsilon_n^*$ as in
\eqref{eq:PVAR_period_cycle_simple}. Let
$\alpha_{Y^*}(\ell)$ denote the strong-mixing coefficients of the normalized process
$\{(\Sigma_Y^*)^{-1/2}Y_n^*\}$.

\begin{assumption}\label{assum:pvar_simple}
For the ScBM--PVAR sequence, the following conditions hold uniformly over $N$:
\begin{itemize}
    \item[(i)] $\{u_n^*\}_{n\in\mathbb{Z}}$ is ergodic with
    $\mathbb{E}[u_n^*]=0$, $\operatorname{Var}(u_n^*)=\Sigma_u^*$, and $\operatorname{Cov}(u_n^*,u_{n-\ell}^*)=0$ for every $\ell\neq0$.
    \item[(ii)] There exist constants $C_1,C_2>0$, independent of $N$, such that
    \begin{displaymath}
        \|\Sigma_u^*\|\leq C_1,
        \qquad
        \|(\Sigma_u^*)^{-1}\|\leq C_2.
    \end{displaymath}
    \item[(iii)] There exist constants $c_1,\gamma_1>0$, independent of $N$, such that
    \begin{displaymath}
        \alpha_{Y^*}(\ell)\leq\exp(-c_1\ell^{\gamma_1}),
        \qquad \ell\geq1.
    \end{displaymath}
    \item[(iv)] There exist constants $c_2,\gamma_2>0$, independent of $N$, such that,
    for every unit vector $v$ and every $x>0$,
    \begin{displaymath}
        \mathbb{P}\!\left(
        \left|v^\prime(\Sigma_Y^*)^{-1/2}Y_n^*\right|>x
        \right)
        \leq
        \exp\!\left\{1-(x/c_2)^{\gamma_2}\right\}.
    \end{displaymath}
    \item[(v)] $N\to\infty$ and
    \begin{displaymath}
        N=\Omega\!\left((qs)^{2/\gamma-1}\right),
        \qquad
        \frac{1}{\gamma}=\frac{1}{\gamma_1}+\frac{1}{\gamma_2},
        \qquad
        \gamma<1.
    \end{displaymath}
\end{itemize}
\end{assumption}

\begin{assumption}\label{assum:vhar_simple}
Assumption~\ref{assum:pvar_simple} holds, with constants uniform over $N$, after replacing the stacked PVAR process by the generalized ScBM--VHAR companion process
\begin{displaymath}
    V_t
    =
    (Y_t^\prime,Y_{t-1}^\prime,\ldots,Y_{t-b_L+1}^\prime)^\prime
    \in\mathbb{R}^{qb_L},
    \qquad
    \Sigma_V:=\operatorname{Var}(V_t),
\end{displaymath}
and replacing the sample size by $N=T-b_L$.
\end{assumption}

Assumptions~\ref{assum:pvar_simple}(iii)--(iv) are high-level concentration conditions, not empirical specification tests. The uniform norm margins in Assumptions~\ref{assum:PVAR_stability_simple}(ii) and
\ref{assum:VHAR_stability}(ii), together with fixed $p^*$ and $b_L$, imply uniform power stability of the corresponding companion matrices: there exist $C_F<\infty$ and $\rho_F\in(0,1)$, independent of $N$, such that $\|F_N^\ell\|\le C_F\rho_F^\ell$ for all $\ell\ge0$. A weighted block-norm argument is given in the Supplementary Material. The next result therefore records primitive
innovation and covariance conditions under which the high-level assumptions hold.

\begin{proposition}\label{prop:primitive_mixing}
Consider either the stacked ScBM--PVAR companion system under
Assumption~\ref{assum:PVAR_stability_simple} or the generalized ScBM--VHAR companion system under Assumption~\ref{assum:VHAR_stability}. Suppose:
\begin{itemize}
    \item[(i)] the innovations in the observed equation are i.i.d. Gaussian, with innovation-covariance eigenvalues bounded uniformly away from zero and infinity;
    \item[(ii)] the eigenvalues of the stationary companion-state covariance matrix are bounded uniformly away from zero and infinity.
\end{itemize}
Then Assumptions~\ref{assum:pvar_simple}(i)--(iv), or their VHAR counterparts in Assumption~\ref{assum:vhar_simple}, hold with constants independent of $N$. In particular, one may take $\gamma_1=1$ and $\gamma_2=2$.
\end{proposition}

Proposition~\ref{prop:primitive_mixing} provides a tractable sufficient setting: the uniform stability margin already imposed on the model sequence, Gaussian innovations, and well-conditioned innovation and state covariances. More general non-Gaussian processes remain covered by the high-level assumptions.

\begin{lemma}\label{lem:estimation}
	Suppose Assumptions~\ref{assum:PVAR_stability_simple} and \ref{assum:pvar_simple} hold. Then, with $N=T/s$,
	\begin{displaymath}
		\|\hat{\Phi}^{\mathrm{ols}} - \Phi\| 
		= \mathcal{O}_{\mathbb{P}}\!\left(\sqrt{\frac{q}{N}}\right).
	\end{displaymath}
\end{lemma}

For the generalized ScBM--VHAR model, we use the restricted VAR($b_L$) representation induced by the horizon lengths $1<b_M<b_L$, which yields the following rate.

\begin{lemma} \label{cor:estimation_vhar}
	Suppose Assumptions~\ref{assum:VHAR_stability} and \ref{assum:vhar_simple} hold for the stable generalized ScBM--VHAR model. Then, with $s=3$ and $N=T-b_L$,
	\begin{displaymath}
		\|\hat{\Phi}^{\mathrm{ols}}-\Phi\|
		=
		\mathcal{O}_{\mathbb{P}}\!\left(\sqrt{\frac{sq}{N}}\right).
	\end{displaymath}
\end{lemma}

Next we replace OLS by lasso estimation $\hat{\Phi}^{\mathrm{lasso}}$. For the ScBM--PVAR model, this is the stacked VAR($p^*$) regression, whereas for the generalized ScBM--VHAR model, it is the restricted VAR($b_L$) regression induced by $R_{b_M,b_L}$. In both cases, the downstream co-clustering theory depends on the first stage only through an operator-norm bound for $\hat{\Phi}-\Phi$. For the ScBM--PVAR model, let $\alpha_P^*$ denote the stacked coefficient vector; since $p^*$ is fixed, we suppress the distinction between $N$ and $N-p^*$ in the rates.
\begin{assumption}\label{assum:pvar_lasso}
For the stacked ScBM--PVAR regression, the following conditions hold along the model
sequence:
\begin{itemize}
    \item[(i)] The vector $\alpha_P^*$ is $\ell_{P,N}$-sparse, that is,
    \begin{displaymath}
        |\mathrm{supp}(\alpha_P^*)|=\ell_{P,N}.
    \end{displaymath}
    \item[(ii)] The sample Gram matrix
    \begin{displaymath}
        \widehat\Gamma_P
        :=N^{-1}\{I_{qs}\otimes(\mathbb X_P^\prime\mathbb X_P)\}
    \end{displaymath}
    satisfies the restricted eigenvalue condition in equation~(4.7) of
    \citet{basu2015regularized}, with curvature $\kappa_{R,P,N}>0$ and tolerance
    $\tau_{R,P,N}\ge0$. The constants are uniform in the sense that
    \begin{displaymath}
        \inf_N\kappa_{R,P,N}>0,
        \qquad
        \ell_{P,N}\tau_{R,P,N}\le \frac{\kappa_{R,P,N}}{32}
    \end{displaymath}
    for all sufficiently large $N$.
    \item[(iii)] The score vector
    \begin{displaymath}
        \widehat\gamma_P
        :=N^{-1}(I_{qs}\otimes\mathbb X_P^\prime)\,\mathrm{vec}(\mathbb Y_P)
    \end{displaymath}
    obeys the deviation bound in equation~(4.8) of \citet{basu2015regularized},
    \begin{displaymath}
        \|\widehat\gamma_P-\widehat\Gamma_P\alpha_P^*\|_\infty
        \le \frac{\lambda_{N,P}}{4},
    \end{displaymath}
    with probability tending to one, where uniformly over the model sequence
    \begin{displaymath}
        \lambda_{N,P}\asymp
        \sqrt{\frac{2\log(qs)+\log p^*}{N}}.
    \end{displaymath}
\end{itemize}
\end{assumption}

These are precisely the deterministic inputs required by \citet[Proposition~4.1]{basu2015regularized}. Note that the analysis in \citet{basu2015regularized} is conducted under Gaussian assumptions, whereas Assumption \ref{assum:pvar_simple} can be interpreted as imposing conditions on sub-Gaussian random vectors. However, subsequent studies (e.g., \cite{wong2020lasso,basu2024high}) show that the suggested convergence rates remain the same without additional conditions. Hence, the theoretical guarantees derived under Gaussian assumptions extend. When the stacked ScBM--PVAR is Gaussian and stable, \citet[Propositions~4.2--4.3]{basu2015regularized} verify Assumption~\ref{assum:pvar_lasso}(ii)--(iii) with process dimension $p=qs$ and lag order $d=p^*$. Since $p^*$ is fixed, the rate simplifies to $\lambda_{N,P}\asymp\sqrt{\log(qs)/N}$.

\begin{lemma}\label{lem:estimation_lasso}
	Let $\hat{\alpha}_P$ be the lasso estimator in \eqref{eq:pvar_lasso} and let $\hat{\Phi}^{\mathrm{lasso}}$ denote the season-stacked autoregressive matrix obtained from $\hat{\alpha}_P$ by the same deterministic aggregation as in the
	co-clustering algorithm. Under Assumptions~\ref{assum:PVAR_stability_simple} and \ref{assum:pvar_lasso},
	\begin{displaymath}
		\|\hat{\Phi}^{\mathrm{lasso}} - \Phi\|		= \mathcal{O}_{\mathbb{P}}\!\left(\sqrt{\ell_{P,N}}\lambda_{N,P}\right)
		= \mathcal{O}_{\mathbb{P}}\!\left(\sqrt{\frac{\ell_{P,N}\log(qs)}{N}}\right).
	\end{displaymath}
\end{lemma}

For the generalized ScBM--VHAR model, let $\beta_V^*$ denote the restricted coefficient vector in \eqref{eq:vhar_lasso}.
\begin{assumption}\label{assum:vhar_lasso}
For the generalized ScBM--VHAR regression, the following conditions hold along the model sequence:
\begin{itemize}
    \item[(i)] The vector $\beta_V^*$ is $\ell_{V,N}$-sparse.
    \item[(ii)] The sample Gram matrix
    \begin{displaymath}
        \widehat\Gamma_V
        :=N^{-1}\{I_q\otimes(X_e^\prime X_e)\}
    \end{displaymath}
    satisfies the restricted eigenvalue condition in equation~(4.7) of \citet{basu2015regularized}, with curvature $\kappa_{R,V,N}>0$ and tolerance $\tau_{R,V,N}\ge0$, where
    \begin{displaymath}
        \inf_N\kappa_{R,V,N}>0,
        \qquad
        \ell_{V,N}\tau_{R,V,N}\le \frac{\kappa_{R,V,N}}{32}
    \end{displaymath}
    for all sufficiently large $N$.
    \item[(iii)] The score vector
    \begin{displaymath}
        \widehat\gamma_V
        :=N^{-1}(I_q\otimes X_e^\prime)\operatorname{vec}(Z)
    \end{displaymath}
    obeys
    \begin{displaymath}
        \|\widehat\gamma_V-\widehat\Gamma_V\beta_V^*\|_\infty
        \le \frac{\lambda_{N,V}}{4}
    \end{displaymath}
    with probability tending to one, where uniformly over the model sequence
    \[
    \lambda_{N,V}\asymp\sqrt{\frac{\log(sq)}{N}},
    \qquad s=3,
    \qquad N=T-b_L.
    \]
\end{itemize}
\end{assumption}

For the fixed-horizon VHAR design, Proposition~1 of \citet{baek2021sparse} verifies the analogue of the Gram-matrix construction in their equation (2.10) and the RE and DB in their equation (2.11) by embedding VHAR into a stable high-order VAR model. The same argument carries over to arbitrarily fixed $(b_M,b_L)$ after replacing the specific aggregation matrix in their equation (2.12) by $R_{b_M,b_L}$. Because $b_M$ and $b_L$ are fixed, this changes only constants, so \citet[Proposition~4.1]{basu2015regularized} yields the same rate order.

\begin{lemma} \label{cor:estimation_vhar_lasso}
	Let $\hat{\beta}_V$ be the restricted lasso estimator in \eqref{eq:vhar_lasso}, and let $\hat{\Phi}^{\,\mathrm{lasso}}$ denote the corresponding stacked horizon-specific coefficient matrix. Under Assumptions~\ref{assum:VHAR_stability} and \ref{assum:vhar_lasso},
	\begin{displaymath}
		\|\hat{\Phi}^{\,\mathrm{lasso}}-\Phi\|
		= \mathcal{O}_{\mathbb{P}}\!\left(\sqrt{\ell_{V,N}}\lambda_{N,V}\right)
		= \mathcal{O}_{\mathbb{P}}\!\left(\sqrt{\frac{\ell_{V,N}\log(sq)}{N}}\right),
		\qquad s=3.
	\end{displaymath}
\end{lemma}

\subsection{Effect of network randomness}
\label{sse:consistency_random}

As a continuation of Section \ref{sse:consistency_estimator}, we now bound $\|\widetilde{\Phi} - \Phi\|$, i.e., the deviation arising from random graphs under the ScBM assumptions.

\begin{assumption}\label{assum_degree_simple}
For each population weighted-adjacency block $\mathcal A$ that enters the construction of $\widetilde\Phi$, define
\begin{displaymath}
    \mathcal O:=\operatorname{diag}(\mathcal A\mathbf 1_q),
    \qquad
    \mathcal P:=\operatorname{diag}(\mathcal A^\prime\mathbf 1_q),
    \qquad
    \bar\tau:=q^{-1}\mathbf 1_q^\prime\mathcal A\mathbf 1_q,
\end{displaymath}
and let
\begin{displaymath}
    \delta:=\min\!\left\{\min_i[\mathcal O]_{ii},\min_j[\mathcal P]_{jj}\right\}.
\end{displaymath}
Uniformly over all seasonal, horizon-specific, and lag-specific blocks,
\begin{displaymath}
    \delta+\bar\tau=\Omega(sqB_{sq}),
    \qquad
    \frac{sqB_{sq}}{\log(sq)}\longrightarrow\infty,
\end{displaymath}
and the corresponding scalar transition multipliers satisfy $|\phi|\le C_\phi$ for some finite constant $C_\phi$ independent of $N$.
\end{assumption}
The realized regularizer is $\tau_A=q^{-1}\mathbf 1_q^\prime A\mathbf 1_q$, whereas $\bar\tau$ is its population counterpart. The degree condition requires regularized expected degrees of order at least $\log(sq)$, as in standard spectral recovery conditions for ScBM-type models \citep[e.g.,][]{rohe2016co,abbe2018community}.

\begin{theorem}\label{lem:population}
	Under Assumption~\ref{assum_degree_simple} and the stability conditions above,
	\begin{displaymath}
		\|\widetilde{\Phi} - \Phi\|
		= \mathcal{O}_{\mathbb{P}}\!\left(\sqrt{\frac{\log(sq)}{sq B_{sq}}}\right)
	\end{displaymath}
	for both ScBM--PVAR and generalized ScBM--VHAR models.
\end{theorem}

\subsection{Singular vectors and misclassification rates}\label{sse:misclassification}

We now translate the operator-norm bounds into perturbation and misclassification rates for the \emph{common} retained singular spaces. For the ScBM-PVAR model, define the population projector-path variation
\begin{displaymath}
    \mathcal V_{P,N}
    :=
    \sum_{m=1}^{s-1}\sum_{\bullet\in\{L,R\}}
    \|U_{m+1,\bullet}-U_{m,\bullet}\|_F,
\end{displaymath}
and, in long-to-short order of the generalized ScBM-VHAR model, define
\begin{displaymath}
    \mathcal V_{V,N}
    :=
    \sum_{\bullet\in\{L,R\}}
    \left(
    \|U_{(m),\bullet}-U_{(\ell),\bullet}\|_F
    +
    \|U_{(s),\bullet}-U_{(m),\bullet}\|_F
    \right).
\end{displaymath}
These quantities measure same-orientation population-subspace heterogeneity; they are not implied by the cross-side membership-continuity restrictions. In particular, adjacent rank changes force positive variation because $\|P-Q\|_F^2\ge|\operatorname{rank}(P)-\operatorname{rank}(Q)|$ for orthogonal projectors $P,Q$.

Throughout, $\bar{\mathcal N}_m^y$ and $\bar{\mathcal N}_m^z$ denote the sets
of series misclustered by the $k$-means step of Section~\ref{sse:algorithm_co}, relative to the population memberships up to a
permutation of labels; the same convention applies to
$\bar{\mathcal N}_{(h)}^y$ and $\bar{\mathcal N}_{(h)}^z$ in the ScBM--VHAR
case.

\begin{assumption}\label{assum:cluster_regular}
For every seasonal or horizon-specific population matrix used in co-clustering:
\begin{itemize}
    \item[(i)] the common retained rank is $r_m=\min(K_{y_m},K_{z_m})$ (or $r_{(h)}$ for VHAR), is uniformly bounded, and its singular subspace is separated from the remainder by an eigengap at least $\delta_0>0$;
    \item[(ii)] the corresponding population singular-vector matrices satisfy
    \begin{displaymath}
        \min_i\|[X_{m,\bullet}]_{i\cdot}\|\ge c_0q^{-1/2},
        \qquad
        \min_i\|[X_{(h),\bullet}]_{i\cdot}\|\ge c_0q^{-1/2},
    \end{displaymath}
    for some $c_0>0$ and $\bullet\in\{L,R\}$;
    \item[(iii)] distinct population row-normalized centroids are separated by at least $\delta_*>0$.
\end{itemize}
All constants are uniform over stages.
\end{assumption}

\begin{remark}
Assumption~\ref{assum:cluster_regular} is satisfied in balanced full-rectangular-rank co-blockmodels with effective degree weights of order $q^{-1}$, separated nonzero singular values of the reduced block matrix, and distinct row-normalized block centroids. Proposition~\ref{prop:ass48-sufficient} in the Supplementary Material gives a formal sufficient condition.
\end{remark}

\begin{theorem}\label{thm:missclassification_pvar}
Consider the ScBM--PVAR model and write
\begin{displaymath}
    \eta_{P,N}=
    \begin{cases}
    q/N, & \text{for OLS},\\[1ex]
    \ell_{P,N}\log(qs)/N, & \text{for the lasso estimator in \eqref{eq:pvar_lasso}}.
    \end{cases}
\end{displaymath}
Suppose Assumption~\ref{assum:cluster_regular} holds and $\alpha_N\in[0,\alpha_{\max})$.  Then there exist orthogonal matrices $R_{m,L},R_{m,R}$ such that
\begin{equation}\label{eq:projector_bound_PVAR}
    \sum_{m=1}^{s}\sum_{\bullet\in\{L,R\}}
    \|\bar X_{m,\bullet}^*-X_{m,\bullet}^*R_{m,\bullet}\|_F
    =
    \mathcal O_{\mathbb P}\!\left(
    \sqrt{sq\,\eta_{P,N}}
    +
    \sqrt{\frac{\log(sq)}{B_{sq}}}
    +
    \sqrt q\,\alpha_N\mathcal V_{P,N}
    \right).
\end{equation}
Moreover,
\begin{displaymath}
    \frac1{sq}\sum_{m=1}^{s}
    \left(|\bar{\mathcal N}_m^y|+|\bar{\mathcal N}_m^z|\right)
    =
    \mathcal O_{\mathbb P}\!\left(
    \eta_{P,N}
    +
    \frac{\log(sq)}{sqB_{sq}}
    +
    \alpha_N^2\mathcal V_{P,N}^2
    \right).
\end{displaymath}
\end{theorem}

\begin{theorem}\label{thm:missclassification_vhar}
Consider generalized ScBM--VHAR, let $s=3$, $N=T-b_L$, and write
\begin{displaymath}
    \eta_{V,N}=
    \begin{cases}
    sq/N, & \text{for OLS},\\[1ex]
    \ell_{V,N}\log(sq)/N, & \text{for the lasso estimator in \eqref{eq:vhar_lasso}}.
    \end{cases}
\end{displaymath}
Under Assumption~\ref{assum:cluster_regular} and $\alpha_N\in[0,\alpha_{\max})$, there exist orthogonal matrices $R_{(h),L},R_{(h),R}$ such that
\begin{equation}\label{eq:projector_bound_VHAR}
    \sum_{h\in\{\ell,m,s\}}\sum_{\bullet\in\{L,R\}}
    \|\bar X_{(h),\bullet}^*-X_{(h),\bullet}^*R_{(h),\bullet}\|_F
    =
    \mathcal O_{\mathbb P}\!\left(
    \sqrt{sq\,\eta_{V,N}}
    +
    \sqrt{\frac{\log(sq)}{B_{sq}}}
    +
    \sqrt q\,\alpha_N\mathcal V_{V,N}
    \right),
\end{equation}
and
\begin{displaymath}
    \frac1{3q}\sum_{h\in\{\ell,m,s\}}
    \left(|\bar{\mathcal N}_{(h)}^y|+|\bar{\mathcal N}_{(h)}^z|\right)
    =
    \mathcal O_{\mathbb P}\!\left(
    \eta_{V,N}
    +
    \frac{\log(sq)}{sqB_{sq}}
    +
    \alpha_N^2\mathcal V_{V,N}^2
    \right).
\end{displaymath}
\end{theorem}

\begin{remark}
The projector bounds \eqref{eq:projector_bound_PVAR} and \eqref{eq:projector_bound_VHAR} separate first-stage estimation error, graph noise, and a deterministic smoothing bias proportional to $\alpha_N\mathcal V_{\bullet,N}$, $\bullet\in\{P,V\}$. The bias vanishes when $\alpha_N=0$, and also for every admissible $\alpha_N$ when the population projector chain is static, since then $\mathcal V_{\bullet,N}=0$.
Because the argument is deterministic after conditioning on the first-stage perturbations, it controls smoothing through a worst-case stability bound and does not quantify the finite-sample variance reduction obtained by combining stage-specific estimated projectors. Establishing such a gain would require a joint analysis of projector errors across stages. The result therefore does not imply that $\alpha_N=0$ is optimal. The row-normalization inequality is valid without a small-perturbation condition; vanishing conditions are needed only for asymptotic consistency.
\end{remark}

\begin{corollary}\label{cor:misclassification_consistency}
If $\eta_{\bullet,N}\to0$, $\alpha_N\mathcal V_{\bullet,N}\to0$ for $\bullet\in\{P,V\}$, and $\log(sq)/(sqB_{sq})\to0$, then the average misclassification fraction converges to zero in probability. The density condition is equivalently $sqB_{sq}/\log(sq)\to\infty$. The same conclusion holds for a data-measurable $\widehat\alpha_N$ whenever $\widehat\alpha_N\mathcal V_{\bullet,N}=o_{\mathbb P}(1)$, $\bullet\in\{P,V\}$.
\end{corollary}

\section{Finite-sample performance}
\label{se:simulation}

We organize the simulation studies around the same sequence: specify a latent community path and signal type, estimate the transition matrices, recover the linked sending and receiving communities, and evaluate the two stages separately. Spectral norm error measures first-stage coefficient estimation. Adjusted Rand index \cite[ARI; e.g.,][]{vinh2009information} provides a label-invariant comparison of recovered partitions, whereas aligned classification accuracy gives a directly interpretable node-wise recovery measure after label matching. The main tables report spectral norm error and aligned accuracy, while the tables and figures in the Supplementary Material provide ARI-based lasso--OLS comparisons; the misspecification tables report both community measures. In every setup, the network matrices are generated once and held fixed across replications, so that innovations are the only replication-level source of randomness. We omit degree correction in the baseline simulations to isolate path recovery from degree heterogeneity. Additional experiments on deliberate model-misspecification appear in Section~\ref{ap:misspecification} in the Supplementary Material.

\subsection{Experiments on ScBM--PVAR models}
\label{subsec:pvar_sparsefixed_sim}

For each season $m=1,\ldots,4$, we first generate a binary sender-by-receiver support matrix $A_m$. If node $i$ belongs to left/sending block $k$ and node $j$ belongs to right/receiving block $r$, then for $i\neq j$, $[A_{m}]_{ij}\sim \mathrm{Bernoulli}(\pi_{kr,m})$, where
\begin{displaymath}
    \pi_{kr,m}=\min\!\left\{\frac{\kappa_{kr}}{n^{z}_{m,r}-\mathbf{1}(k=r)},\,0.95\right\},
\end{displaymath}
$n^{z}_{m,r}$ is the size of receiving block $r$ in season $m$. Hence the expected off-diagonal support per row is $O(1)$, and the matrix density is $O(q^{-1})$. Conditional on $A_m$, we construct a raw sender-by-receiver coefficient matrix $\widetilde G_m$. The own-lag diagonal is always retained,
\begin{displaymath}
    [\widetilde G_{m}]_{ii}=a_{\mathrm{self}}u_{ii},
    \qquad
    u_{ii}\sim \mathrm{Unif}(0.95,1.05),
\end{displaymath}
and for $i\neq j$ with $k=\xi_{m,i}^y$ and $r=\xi_{m,j}^z$,
\begin{displaymath}
    [\widetilde G_{m}]_{ij}=
    \begin{cases}
    	a_{\mathrm{diag}}u_{ij}, & [A_{m}]_{ij}=1,\ k=r,\\
    	a_{\mathrm{upper}}u_{ij}, & [A_{m}]_{ij}=1,\ k<r,\\
    	a_{\mathrm{lower}}u_{ij}, & [A_{m}]_{ij}=1,\ k>r,\\
    	0, & [A_{m}]_{ij}=0,
    \end{cases}
    \qquad
    u_{ij}\sim \mathrm{Unif}(0.90,1.10).
\end{displaymath}
Thus $A_m$ determines the support, while the block pair determines the coefficient magnitude. The two designs use
\begin{displaymath}
    \begin{aligned}
    \text{Type~1:}\quad
    &(a_{\mathrm{self}},a_{\mathrm{diag}},a_{\mathrm{upper}},a_{\mathrm{lower}})
    =(0.30,0.14,0.04,0.06),\\
    &(\kappa_{\mathrm{diag}},\kappa_{\mathrm{upper}},\kappa_{\mathrm{lower}})
    =(2.6,0.6,0.9),\\[1mm]
    \text{Type~2:}\quad
    &(a_{\mathrm{self}},a_{\mathrm{diag}},a_{\mathrm{upper}},a_{\mathrm{lower}})
    =(0.28,0.12,0.05,0.07),\\
    &(\kappa_{\mathrm{diag}},\kappa_{\mathrm{upper}},\kappa_{\mathrm{lower}})
    =(2.2,0.7,1.3).
    \end{aligned}
\end{displaymath}
Type~2 is intentionally more difficult because it allows relatively stronger between-block connectivity. Consistent with the model convention $M_m=\Phi_m^\prime$, the actual time-series transition is obtained by transposition, $\widetilde\Phi_m=\widetilde G_m^\prime$. After constructing $\widetilde\Phi_1,\ldots,\widetilde\Phi_4$, we apply a common rescaling satisfying $\sigma_{\max}(\Phi_4\Phi_3\Phi_2\Phi_1)=0.90$. The stability rescaling is recomputed after the transposition. This preserves the relative seasonal pattern while ensuring a stable but non-trivial signal.

\begin{figure}[t]
	\centering
	\includegraphics[width=.92\linewidth]{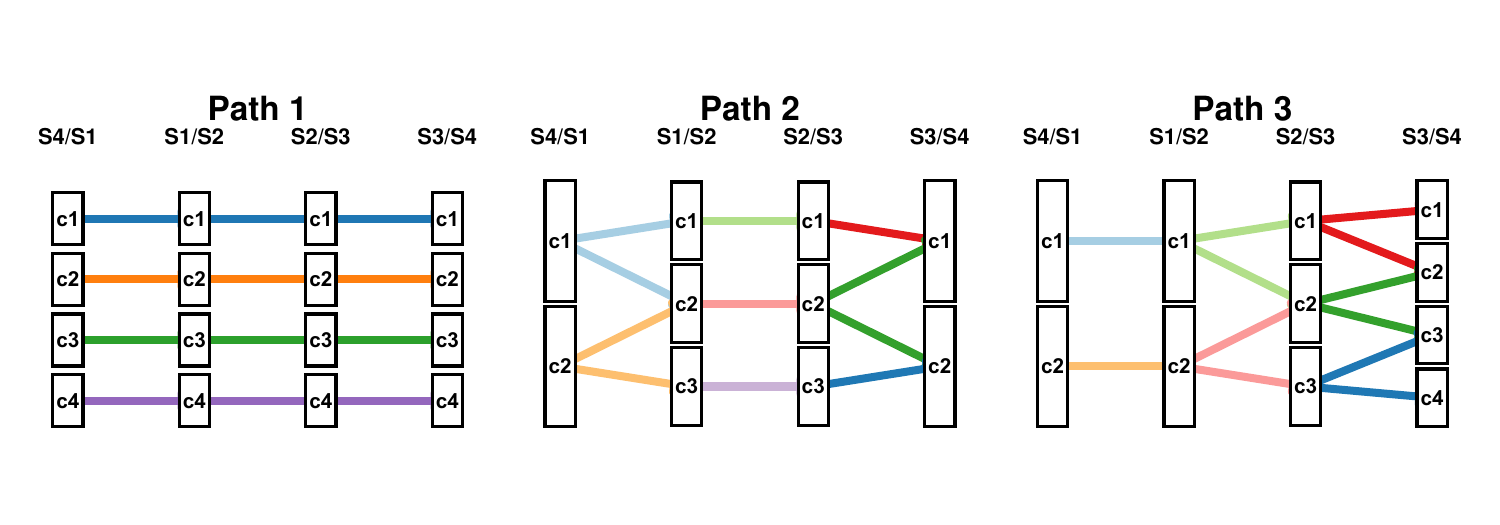}
	\vspace{-7mm}
	\caption{Community paths used in the ScBM--PVAR simulation. Path~1 is static, Path~2 follows $2\to3\to3\to2$, and Path~3 follows $2\to2\to3\to4$. The state labels report the linked seasonal states: $S4/S1\to S1/S2\to S2/S3\to S3/S4$.}
	\label{fig:pvar_sparsefixed_paths_v3}
\end{figure}

Figure~\ref{fig:pvar_sparsefixed_paths_v3} displays the three path designs, 
\begin{displaymath}
    \Xi_{4}^z=\Xi_{1}^y
    \;\longrightarrow\;
    \Xi_{1}^z=\Xi_{2}^y
    \;\longrightarrow\;
    \Xi_{2}^z=\Xi_{3}^y
    \;\longrightarrow\;
    \Xi_{3}^z=\Xi_{4}^y
    \;\longrightarrow\;
    \Xi_{4}^z=\Xi_{1}^y.
\end{displaymath}
where the transition of the partition correspond to \eqref{eq:community_path_PVAR} with $s=4$ ($\longrightarrow$ at both ends are omitted). Path~1 is static, with four communities throughout. Path~1 also serves as a static benchmark: when the block structure does not change over seasons, the proposed co-spectral clustering may be viewed as a directed extension of the symmetrized spectral clustering of \citet{gudhmundsson2021detecting}. Including this case allows us to verify that the method works well even in the non-dynamic baseline before turning to split--merge and refinement paths. Path~2 follows a split--merge pattern $2\to3\to3\to2$. Path~3 follows a nested refinement pattern $2\to2\to3\to4$. We consider $q\in\{18,36,60\}$ and $T\in\{200,500,1000,2000\}$. For each combination of dimension, path, and type, we generate one seasonal coefficient system and run $200$ replications with Gaussian innovations, diagonal innovation variance $0.5$, and burn-in length $500$. Across the main and supplementary results, performance is summarized by mean spectral norm error, aligned accuracy, and ARI, with the four season-specific clustering results averaged within each replication. The main table reports the first two measures, and the supplementary comparison figure reports ARI.

We compare OLS and lasso estimation results. For lasso estimation, we use the accelerated proximal gradient method \cite[FISTA;][]{beck2009fast} with a common penalty multiplier selected by block cross-validation \citep{baek2021sparse}. For each design, one pilot series is simulated. Let $N_{\mathrm{eff},m}$ be the usable sample size in season $m$. The baseline penalty is $\lambda_m^{\mathrm{base}}=\{\log(sq^2)/N_{\mathrm{eff},m}\}^{1/2}$ and we search over $c_\lambda\in\{0.10,0.15,0.20,\ldots,1.00\}$. For each candidate $c_\lambda$, we fit the seasonal lasso regressions with $\lambda_m=c_\lambda\lambda_m^{\mathrm{base}}$ and evaluate them by $10$-fold block cross-validation, summing the prediction errors over the four seasons. The selected $c_\lambda$ is then fixed within the same setup. The PisCES smoothing parameter $\alpha_N$ is selected by the uniform projector cross-validation procedure in Section~\ref{sse:algorithm_cv}.

\begin{table}[t!]
\centering
\scriptsize
\setlength{\tabcolsep}{5pt}
\begin{adjustbox}{max width=\textwidth}
\begin{tabular}{lllcccccccc}
\hline
\multirow{2}{*}{Dim} & \multirow{2}{*}{Path} & \multirow{2}{*}{Type} & \multicolumn{4}{c}{Spectral Norm} & \multicolumn{4}{c}{Accuracy} \\
\cline{4-11}
& & & $T=200$ & $T=500$ & $T=1000$ & $T=2000$ & $T=200$ & $T=500$ & $T=1000$ & $T=2000$ \\
\hline
\multirow{6}{*}{$q=18$} & path1 & type1 & 0.583 & 0.430 & 0.325 & 0.239 & 0.869 & 0.998 & 1.000 & 1.000 \\
 & path1 & type2 & 0.604 & 0.446 & 0.335 & 0.247 & 0.706 & 0.958 & 0.997 & 1.000 \\
 & path2 & type1 & 0.611 & 0.426 & 0.311 & 0.222 & 0.749 & 0.938 & 0.969 & 0.981 \\
 & path2 & type2 & 0.588 & 0.427 & 0.315 & 0.227 & 0.661 & 0.837 & 0.908 & 0.935 \\
 & path3 & type1 & 0.592 & 0.430 & 0.317 & 0.226 & 0.717 & 0.813 & 0.855 & 0.900 \\
 & path3 & type2 & 0.592 & 0.432 & 0.319 & 0.230 & 0.621 & 0.707 & 0.759 & 0.778 \\
\hline
\multirow{6}{*}{$q=36$} & path1 & type1 & 0.663 & 0.498 & 0.377 & 0.279 & 0.627 & 0.935 & 0.997 & 1.000 \\
 & path1 & type2 & 0.645 & 0.500 & 0.395 & 0.294 & 0.483 & 0.736 & 0.927 & 0.973 \\
 & path2 & type1 & 0.678 & 0.513 & 0.379 & 0.271 & 0.601 & 0.839 & 0.959 & 0.977 \\
 & path2 & type2 & 0.656 & 0.495 & 0.373 & 0.271 & 0.548 & 0.680 & 0.821 & 0.883 \\
 & path3 & type1 & 0.673 & 0.503 & 0.374 & 0.271 & 0.592 & 0.714 & 0.779 & 0.805 \\
 & path3 & type2 & 0.671 & 0.505 & 0.380 & 0.275 & 0.546 & 0.672 & 0.751 & 0.774 \\
\hline
\multirow{6}{*}{$q=60$} & path1 & type1 & 0.689 & 0.546 & 0.423 & 0.312 & 0.460 & 0.812 & 0.977 & 0.998 \\
 & path1 & type2 & 0.677 & 0.549 & 0.441 & 0.335 & 0.394 & 0.529 & 0.783 & 0.917 \\
 & path2 & type1 & 0.724 & 0.561 & 0.421 & 0.301 & 0.548 & 0.750 & 0.930 & 0.967 \\
 & path2 & type2 & 0.702 & 0.551 & 0.421 & 0.306 & 0.512 & 0.611 & 0.746 & 0.813 \\
 & path3 & type1 & 0.743 & 0.573 & 0.431 & 0.305 & 0.516 & 0.686 & 0.760 & 0.766 \\
 & path3 & type2 & 0.699 & 0.555 & 0.428 & 0.310 & 0.494 & 0.566 & 0.678 & 0.737 \\
\hline
\end{tabular}
\end{adjustbox}
\caption{PisCES-smoothed lasso results for the ScBM--PVAR models under common-rank spectral co-clustering and uniform projector cross-validation.}
\label{tab:pvar_pisces_lasso_v3}
\end{table}

Table~\ref{tab:pvar_pisces_lasso_v3} yields three consistent findings. First, increasing $T$ reduces spectral norm error and improves recovery in every design. Second, recovery deteriorates with dimension, Type~2 cross-block interaction, and path complexity. Even the nested-refinement Path~3 improves materially with sample size: for $(q=18,\mathrm{Type~1})$, accuracy rises from $0.717$ at $T=200$ to $0.900$ at $T=2000$. The higher-dimensional Path~3 designs remain more demanding, but the results show clear improvement with $T$ rather than a systematic plateau. Third, Table~\ref{tab:pvar_pisces_ols} and Figure~\ref{fig:pvar_ari_comparison} in the Supplementary Material show that sparse first-stage estimation is especially important in the higher-dimensional dynamic designs. Thus, path complexity continues to matter, but its effect is distinct from first-stage time-series estimation error and diminishes substantially as the sample size grows.

\subsection{Experiments on generalized ScBM--VHAR models}
\label{subsec:vhar_sparsefixed_sim}

We keep the three horizon-specific coefficient matrices fixed within each setup. If node $i$ belongs to sending block $k$ and node $j$ belongs to receiving block $r$ at horizon $h\in\{s,m,\ell\}$, we generate $[A_{(h)}]_{ij} \sim \mathrm{Bernoulli}([B_{(h)}]_{kr})$, where $B_{(h)}$ is a horizon-specific block-probability matrix. We consider the following two specifications:
\begin{displaymath}
	[B_{(h)}]_{kr}=
	\begin{cases}
		p_{\mathrm{diag},h}, & k=r,\\
		p_{\mathrm{off},h}, & k\neq r,
	\end{cases}
	\qquad
	[B_{(h)}]_{kr}=
	\begin{cases}
		p_{\mathrm{diag},h}-0.03, & k=r,\\
		p_{\mathrm{off},h}+0.03, & k<r,\\
		p_{\mathrm{off},h}+0.06, & k>r,
	\end{cases}
\end{displaymath}
These are the Type~1 and Type~2 designs, respectively, with probabilities truncated to $[0.001,0.995]$.  Type~1 has stronger within-block signal and weaker between-block interactions.  Type~2 reduces the within-block magnitudes and strengthens the cross-block magnitudes, thereby weakening community separation.  The horizon-specific probabilities are
\begin{displaymath}
    \begin{aligned}
    \text{Type~1:}\quad
    &(p_{\mathrm{diag},S},p_{\mathrm{diag},M},p_{\mathrm{diag},L})=(0.95,0.93,0.91),\\
    &(p_{\mathrm{off},S},p_{\mathrm{off},M},p_{\mathrm{off},L})=(0.02,0.03,0.04),\\[1mm]
    \text{Type~2:}\quad
    &(p_{\mathrm{diag},S},p_{\mathrm{diag},M},p_{\mathrm{diag},L})=(0.92,0.90,0.88),\\
    &(p_{\mathrm{upper},S},p_{\mathrm{upper},M},p_{\mathrm{upper},L})=(0.05,0.06,0.07),\\
    &(p_{\mathrm{lower},S},p_{\mathrm{lower},M},p_{\mathrm{lower},L})=(0.08,0.09,0.10).
    \end{aligned}
\end{displaymath}

\begin{figure}[t]
	\centering
	\includegraphics[width=.92\linewidth]{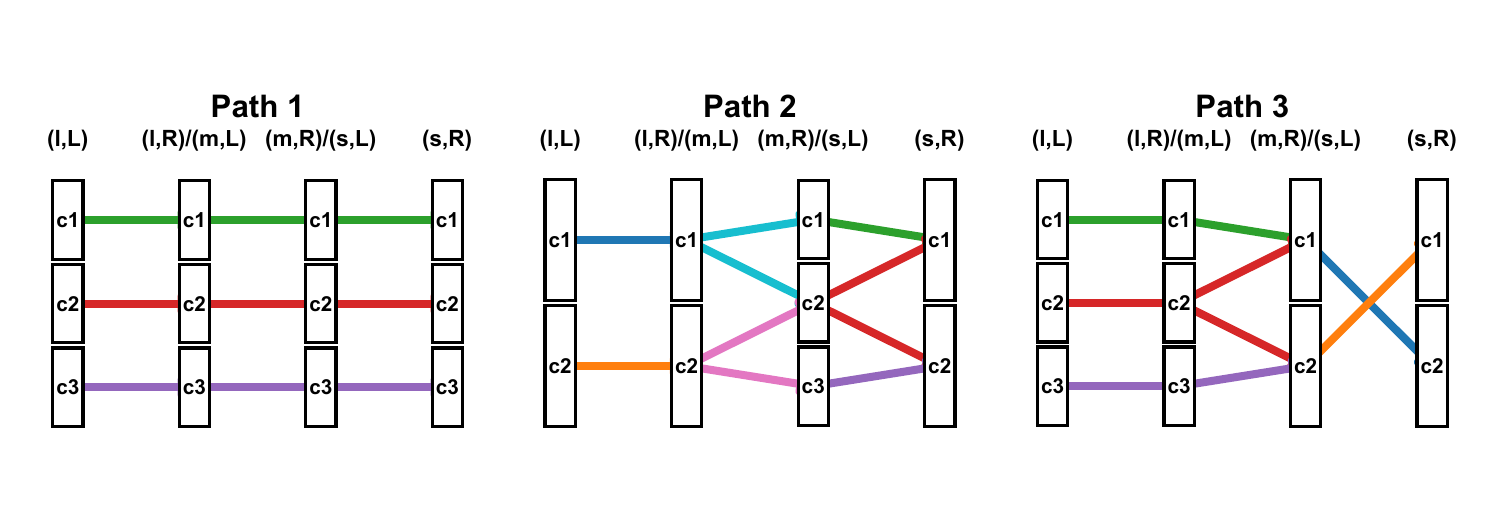}
	\vspace{-7mm}
	\caption{Community paths used in the generalized ScBM--VHAR simulation. Path 1 is static, Path 2 follows $2\to2\to3\to2$, and Path 3 follows $3\to3\to2\to2$. The state labels report horizon and singular-space orientation: Long-L, Long-R/Medium-L, Medium-R/Short-L, and Short-R.}
	\label{fig:vhar_sparsefixed_paths_v3}
\end{figure}

Conditional on $A_{(h)}$, we normalize by sender-block size and define $[\widetilde{A}_{(h)}]_{ij}=[A_{(h)}]_{ij}/n_{h,k}^{y}$, where $n_{h,k}^{y}$ is the size of sending block $k$ at horizon $h$. This prevents larger sending blocks from generating stronger coefficients purely because they contain more nodes, and therefore keeps the block effect interpretable as an average per-sender effect. We then set
\begin{displaymath}
    \widetilde{\Phi}_{(h)}=c_h\widetilde A_{(h)}^\prime,\quad h\in\{s,m,\ell\},
    \qquad
    (c_s,c_m,c_{\ell})=(0.34,0.28\sqrt{b_M},0.24\sqrt{b_L}),
\end{displaymath}
so that $\widetilde\Phi_{(h)}^\prime$ has the intended left/right block organization. We then apply a common rescaling factor $a>0$ satisfying $\sigma_{\max}(a(\widetilde\Phi_{(s)}+\widetilde\Phi_{(m)}+\widetilde\Phi_{(\ell)}))=0.90$. The companion stability and innovation scaling are recomputed after this transposition, and the innovation scale is chosen so that the mean stationary marginal variance equals $0.5$.

Figure~\ref{fig:vhar_sparsefixed_paths_v3} displays the three latent paths in the natural long-to-short order
\begin{displaymath}
    (\ell,L)\;\longrightarrow\;(\ell,R)/(m,L)\;\longrightarrow\;(m,R)/(s,L)\;\longrightarrow\;(s,R).
\end{displaymath}
These four states correspond respectively to $\Xi^y_{(\ell)}$, $\Xi^z_{(\ell)}=\Xi^y_{(m)}$, $\Xi^z_{(m)}=\Xi^y_{(s)}$, and $\Xi^z_{(s)}$ in \eqref{eq:community_path_VHAR}. Path~1 is static. Path~2 follows the split--merge pattern $2\to2\to3\to2$, whereas Path~3 follows the coarsening pattern $3\to3\to2\to2$. We consider $q\in\{18,24,36\}$ and $T\in\{500,1000,2000,3000\}$. Each setup is run for 200 replications with burn-in length 300, $(b_M,b_L)=(3,10)$, and Gaussian innovations with diagonal variance of 0.5. We use both OLS and lasso estimation, but report only the PisCES-smoothed lasso results in the main text. The lasso penalty multiplier is selected by block cross-validation over $c_\lambda\in\{0.10,0.15,0.20,\ldots,1.00\}$, and the PisCES smoothing parameter $\alpha_N$ is selected by the uniform projector cross-validation procedure in Section~\ref{sse:algorithm_cv}.

\begin{table}[t!]
\centering
\scriptsize
\setlength{\tabcolsep}{5pt}
\begin{adjustbox}{max width=\textwidth}
\begin{tabular}{lllcccccccc}
\hline
\multirow{2}{*}{Dim} & \multirow{2}{*}{Path} & \multirow{2}{*}{Type} & \multicolumn{4}{c}{Spectral Norm} & \multicolumn{4}{c}{Accuracy} \\
\cline{4-11}
& & & $T=500$ & $T=1000$ & $T=2000$ & $T=3000$ & $T=500$ & $T=1000$ & $T=2000$ & $T=3000$ \\
\hline
\multirow{6}{*}{$q=18$} & path1 & type1 & 0.347 & 0.278 & 0.227 & 0.198 & 0.742 & 0.906 & 0.983 & 0.997 \\
 & path1 & type2 & 0.343 & 0.276 & 0.227 & 0.200 & 0.643 & 0.810 & 0.943 & 0.980 \\
 & path2 & type1 & 0.347 & 0.282 & 0.232 & 0.205 & 0.664 & 0.747 & 0.850 & 0.911 \\
 & path2 & type2 & 0.345 & 0.279 & 0.230 & 0.203 & 0.611 & 0.709 & 0.831 & 0.886 \\
 & path3 & type1 & 0.351 & 0.286 & 0.236 & 0.209 & 0.573 & 0.643 & 0.784 & 0.873 \\
 & path3 & type2 & 0.350 & 0.286 & 0.237 & 0.210 & 0.562 & 0.605 & 0.702 & 0.793 \\
\hline
\multirow{6}{*}{$q=24$} & path1 & type1 & 0.365 & 0.288 & 0.237 & 0.208 & 0.650 & 0.832 & 0.969 & 0.993 \\
 & path1 & type2 & 0.357 & 0.286 & 0.233 & 0.209 & 0.582 & 0.755 & 0.920 & 0.975 \\
 & path2 & type1 & 0.357 & 0.288 & 0.239 & 0.215 & 0.606 & 0.680 & 0.792 & 0.875 \\
 & path2 & type2 & 0.358 & 0.291 & 0.240 & 0.212 & 0.574 & 0.655 & 0.757 & 0.824 \\
 & path3 & type1 & 0.364 & 0.298 & 0.251 & 0.224 & 0.534 & 0.584 & 0.700 & 0.811 \\
 & path3 & type2 & 0.360 & 0.294 & 0.243 & 0.219 & 0.525 & 0.566 & 0.637 & 0.745 \\
\hline
\multirow{6}{*}{$q=36$} & path1 & type1 & 0.369 & 0.296 & 0.242 & 0.213 & 0.561 & 0.776 & 0.950 & 0.985 \\
 & path1 & type2 & 0.358 & 0.289 & 0.235 & 0.211 & 0.488 & 0.628 & 0.852 & 0.945 \\
 & path2 & type1 & 0.383 & 0.307 & 0.248 & 0.227 & 0.528 & 0.574 & 0.702 & 0.759 \\
 & path2 & type2 & 0.359 & 0.291 & 0.245 & 0.223 & 0.545 & 0.582 & 0.672 & 0.741 \\
 & path3 & type1 & 0.365 & 0.306 & 0.257 & 0.237 & 0.506 & 0.523 & 0.583 & 0.672 \\
 & path3 & type2 & 0.361 & 0.304 & 0.254 & 0.231 & 0.500 & 0.521 & 0.562 & 0.616 \\
\hline
\end{tabular}
\end{adjustbox}
\caption{PisCES-smoothed lasso results for the generalized ScBM--VHAR models under common-rank spectral co-clustering and uniform projector cross-validation.}
\label{tab:vhar_pisces_lasso}
\end{table}

Table~\ref{tab:vhar_pisces_lasso} displays the same sample-size improvement as the PVAR experiment but also shows that horizon-linked recovery is more demanding. Larger $q$, Type~2, and non-static paths weaken recovery, but the linked horizon structure becomes sharply more recoverable as $T$ grows. For example, with $(q=24,\mathrm{Path~2},\mathrm{Type~1})$, PisCES accuracy/ARI increase from $0.606/0.081$ at $T=500$ to $0.875/0.685$ at $T=3000$. Figure~\ref{fig:vhar_ari_comparison} in the Supplementary Material also shows the importance of sparse first-stage estimation: for $(q=36,\mathrm{Path~1},\mathrm{Type~1},T=3000)$, lasso attains ARI $0.966$, compared with $0.171$ for OLS. Table~\ref{tab:vhar_pisces_ols} in the Supplementary Material reports the corresponding spectral norm and accuracy comparison. Taken together, the two simulation sections support a common conclusion: first-stage estimation accuracy improves with sample size, while community recovery additionally depends on block separation and latent-path complexity.

\subsection{Robustness to model misspecification}
\label{subsec:misspecification}

We further examine two deliberate departures from the baseline ScBM specification. The common ranks $r_m=\min(K_{y_m},K_{z_m})$ are retained for each setup, and we select the PisCES smoothing parameter using the uniform projector cross-validation procedure described in Section~\ref{sse:algorithm_cv}. We replicate each setup 100 times. Note that, because the two robustness experiments are generated from separate fixed coefficient systems, their baseline results should be interpreted within
each experiment. First, we under- and over-specify the community path while keeping the data-generating transition matrices unchanged. Second, we relocate a fraction $\eta\in\{0.1,0.2,0.3\}$ of the nonzero off-diagonal coefficients in the transition matrices while preserving their magnitudes and target stability scaling. Accordingly, the $\eta=0$ rows in Table~\ref{tab:misspec_eta} and the correctly specified rows in Table~\ref{tab:misspec_K} are treated as experiment-specific baselines. Full accuracy and ARI results are reported in Tables~\ref{tab:misspec_K} and~\ref{tab:misspec_eta} in the Supplementary Material.

First, Table~\ref{tab:misspec_K} shows that the advantage of using the correct community path becomes clearer at the larger sample sizes. For PVAR at $T=2000$, the correctly specified fit yields accuracy $0.898$ and ARI $0.690$, compared with $(0.756,0.501)$ under the coarse fit and $(0.699,0.499)$ under the fine fit. For VHAR at $T=3000$, the correct fit yields accuracy $0.997$ and ARI $0.990$, whereas the coarse and fine fits yield $(0.665,0.518)$ and $(0.855,0.832)$, respectively. At the smaller sample sizes, the differences are less pronounced; in particular, the fine PVAR fit at $T=500$ has a slightly higher ARI than the correct fit despite having lower accuracy. Overall, the results suggest that the importance of correctly specifying the community dimension becomes more evident as sampling variability decreases.

Second, Table~\ref{tab:misspec_eta} shows that community recovery generally deteriorates as the block structure is increasingly disrupted. For PVAR at $T=2000$, accuracy decreases from $0.898$ at $\eta=0$ to $0.849$, $0.695$, and $0.709$, while ARI decreases from $0.690$ to $0.579$, $0.277$, and $0.300$. For VHAR at $T=3000$, accuracy falls from $0.808$ to $0.734$, $0.700$, and $0.620$, and ARI decreases from $0.510$ to $0.338$, $0.248$, and $0.097$. In contrast, the spectral-norm errors remain nearly unchanged across contamination levels. Thus, departures from the block structure primarily weaken recovery of the latent communities rather than estimation of the transition matrices themselves.

\section{Data applications}
\label{se:appl}

We illustrate the proposed framework through two applications that capture different forms of temporal variation in directional community structure. The first examines seasonal reorganization among U.S. nonfarm payroll sectors, while the second studies changes in dependence across short-, medium-, and long-run horizons in global realized stock-market volatilities. 

In both cases, the goal is to recover the latent directional network structure and its evolving sending and receiving communities from the multivariate time series. This setting differs from the NAR models \cite[e.g.,][]{zhu2017network,yin2023general}, where the network is observed or prespecified and inference is conducted conditional on that network. Rather than assuming an external adjacency matrix, we infer the directed network and its community memberships directly from the data.

\subsection{Quarterly employees on nonfarm payrolls by industry sectors}
\label{sse:appl_pvar}

The nonfarm payroll employment series records paid U.S. workers in selected industries, excluding farmworkers, private household employees, and military personnel.  We use monthly Bureau of Labor Statistics series obtained from FRED, aggregate them to quarterly frequency from January 1990 to March 2020, and apply logs and first differences.  The resulting panel has $q=22$ sectors and $T=120$ observations; the series definitions appear in Table~\ref{tab:employment_sectors} in the Supplementary Material.

A sending community groups sectors whose lagged employment changes have similar predictive effects on the rest of the panel, whereas a receiving community groups sectors with similar responses to lagged sectoral variation.  Figure~\ref{fig:timeplot_PVAR} in the Supplementary Material reports representative time-series diagnostics, which support a low-order periodic specification and show substantial cross-sector heterogeneity.

\begin{figure}[t!]
\centering
\includegraphics[width=1.0\textwidth]{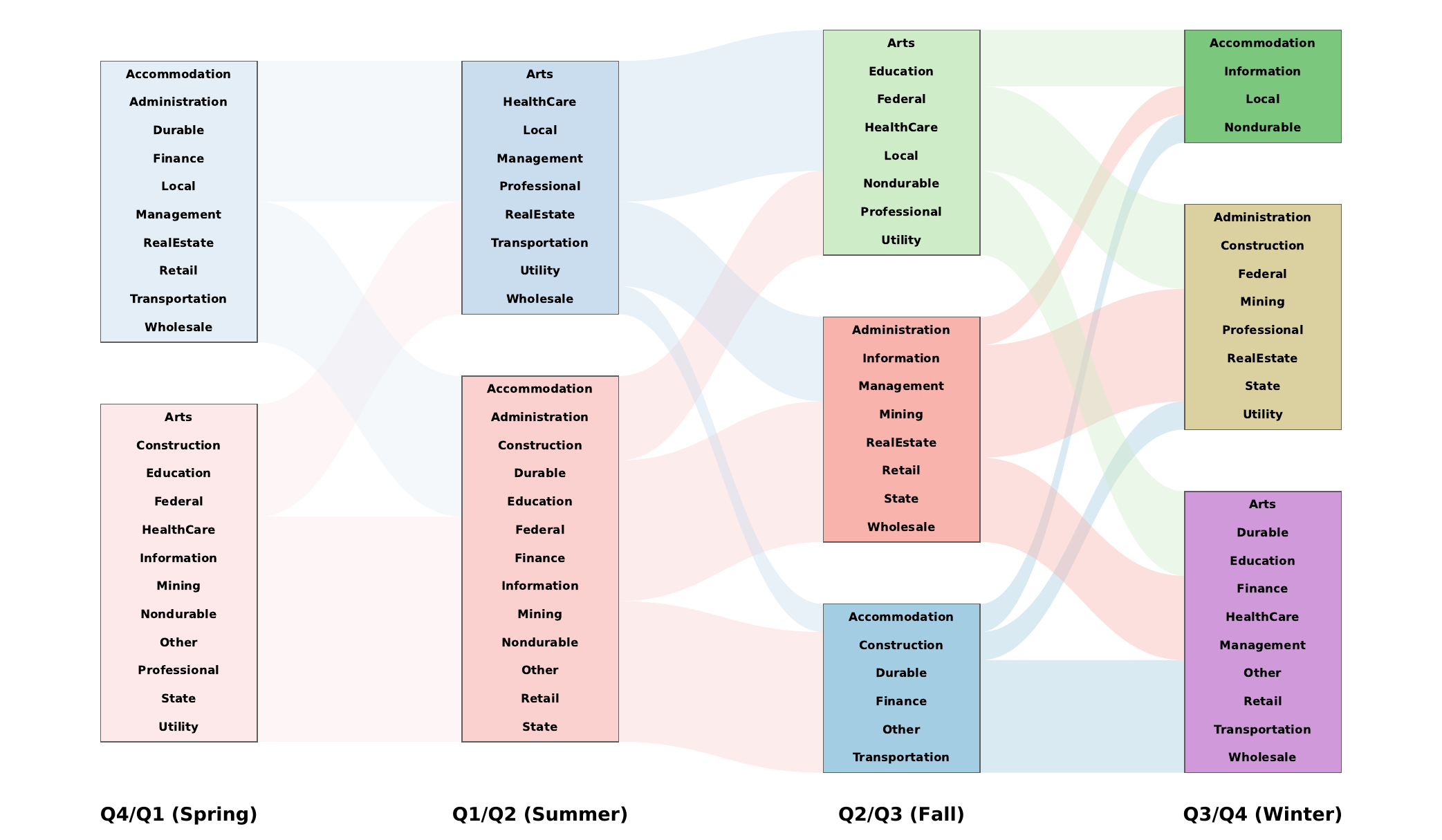}
\caption{Aligned quarterly sending-community states for 22 U.S. nonfarm-payroll sectors under the lasso-based common-rank ScBM--PVAR model.  Under the cyclic continuity restriction, the receiving partition in one quarter is the sending partition in the next; the Q3/Q4 (Winter)-to-Q4/Q1 (Spring) closure is omitted for visual clarity.}
\label{fig:pvar_payroll_sankey}
\end{figure}

We estimate a four-season PVAR with common lag order one. The seasonal scree diagnostics are reported in Figure~\ref{fig:rank_PVAR} in the Supplementary Material. We use them as descriptive spectral dimension diagnostics, in the same broad spirit as eigenvalue-based dimension diagnostics in factor models \citep{bai2007determining}, rather than as a formal order-selection theorem. They motivate the effective common stage ranks $(r_1,r_2,r_3,r_4)=(2,2,3,2)$. Conditional on these ranks and the path-continuity restriction, the baseline directional allocation is
\begin{displaymath}
    (K_{y_1},K_{z_1},K_{y_2},K_{z_2},K_{y_3},K_{z_3},K_{y_4},K_{z_4})
    =(2,2,\,2,3,\,3,3,\,3,2).
\end{displaymath}
Thus the sending sequence is $2\to2\to3\to3$, whereas the receiving sequence is $2\to3\to3\to2$. The uniform projector criterion selects $\widehat\alpha_N=0.1293$, the upper candidate $0.99\alpha_{\max}$. Cross-validation selects only the smoothing intensity, not the community counts.  This path is interpreted conditionally on the continuity restriction in \eqref{eq:season_cyclic}: the restriction aligns adjacent directional roles but does not test whether unmatched quarter-to-quarter changes are absent.

Figure~\ref{fig:pvar_payroll_sankey} shows a coarse two-group sending organization in the first half of the year, followed by substantial differentiation and recomposition. In Q2/Q3 (Fall), one group combines Construction, Durable, Transportation, Finance, Accommodation, and Other; a second combines Mining, Wholesale, Retail, Information, RealEstate, Management, Administration, and State; and a third combines Nondurable, Utility, Professional, Education, HealthCare, Arts, Federal, and Local. Q3/Q4 (Winter) retains three groups but reorganizes them into broader mixed functional blocks: Nondurable, Information, Accommodation, and Local form one compact group; Mining, Construction, Utility, RealEstate, Professional, Administration, Federal, and State form a second; and the remaining durable-goods, trade, finance, management, education, health, arts, and other-service sectors form a third. The final receiving partition then compresses this organization back to the two-group Q4/Q1 (Spring) state.

Several label-invariant relations persist throughout the cycle. Mining and State, Durable and Finance, Wholesale and Management, Utility and Professional, and HealthCare and Arts follow identical sending paths.  Retail, Information, Local, and Administration change peer sets more frequently and act as seasonal interfaces. These statements describe similarities in estimated directional coefficient profiles, not causal classifications. Figure~\ref{fig:pvar_payroll_sankey_alt} in the Supplementary Material reports an alternative admissible community-count/rank configuration. It preserves the qualitative evidence of seasonal reorganization, although the detailed mid-year partition is sensitive to that configuration.

\subsection{Realized volatilities of stock indices across different financial markets}
\label{sse:appl_vhar}

Realized volatility (RV) is an ex-post measure of return variation, typically constructed from the sum of squared intraday returns. We compute RV from 5-minute returns and follow the cleaning and aggregation procedure in Section~4.3 of \citet{corsi2009simple}; see also \citet{andersen2003modeling}. Our panel contains $q=29$ stock-index RV series for major equity markets. Since the original Oxford-Man Institute data are no longer publicly available, we use the processed dataset of \citet{baek2021sparse}. The sample spans January~3, 2010 to December~31, 2019, giving $T=2618$ observations and ending before the COVID-19 pandemic. The stock indices, MSCI classifications \citep{msci2025}, and regional information are listed in Table~\ref{tab:stock_indices} in the Supplementary Material.

At a given horizon, a sending community contains markets with similar outgoing lagged volatility effects, while a receiving community contains markets with similar responses to lagged volatility elsewhere in the system. The linked path is read from monthly to weekly to daily dependence, so it describes how directional market roles change as the horizon shortens.

Figure~\ref{fig:timeplot_VHAR} in the Supplementary Material reports time plots and sample ACF/PACF functions for four representative indices: FTSE~100, Nikkei~225, KOSPI Composite, and IPC Mexico. All four series exhibit persistent autocorrelation, consistent with long-memory-type volatility dynamics, but their volatility spikes differ markedly across markets. For instance, FTSE~100, KOSPI, and IPC Mexico show sharp increases around the third quarter of 2011, whereas Nikkei~225 peaks earlier and more gradually. KOSPI also appears closer to FTSE~100 than to Nikkei~225, which suggests that RV comovement is not driven solely by geographic proximity and motivates a network-based analysis of latent community structure.

We estimate the generalized ScBM--VHAR model under the standard specification $(b_M,b_L)=(5,22)$, so that the short-, medium-, and long-horizon components correspond to daily, weekly, and monthly RV aggregates. Scree plots of the estimated transition matrices, reported in Figure~\ref{fig:rank_VHAR} in the Supplementary Material, support common retained ranks $(3,3,3)$ and motivate a potentially richer middle-horizon partition. The baseline uses
\begin{displaymath}
    \big(K_{y_{(\ell)}},K_{z_{(\ell)}},K_{y_{(m)}},K_{z_{(m)}},K_{y_{(s)}},K_{z_{(s)}}\big)
=(3,3,3,3,3,3).
\end{displaymath}
For this application, we use a parsimonious three-state restricted submodel in which the two directional roles at the longest horizon share a common partition, $\Xi^y_{(\ell)}=\Xi^z_{(\ell)}$. Combined with the generic continuity restriction \eqref{eq:vhar-path-constraints}, this gives
\begin{displaymath}
    \text{Long}=(\ell,L)/(\ell,R)/(m,L),\qquad
    \text{Medium}=(m,R)/(s,L),\qquad
    \text{Short}=(s,R).
\end{displaymath}
This additional equality is an application-specific restriction used to represent a single persistent monthly organization; it is not implied by the generic four-state VHAR model. The uniform projector criterion selects $\widehat\alpha_N=0.1293$, the upper candidate $0.99\alpha_{\max}$, for both the baseline and the $3\to4\to3$ sensitivity specification. The CV loss decreases over the admissible grid, indicating that the criterion favors substantial cross-horizon smoothing within the theoretically allowed range.

\begin{figure}[t!]
	\centering
	\includegraphics[width=1\textwidth,height=0.4\textheight,keepaspectratio]{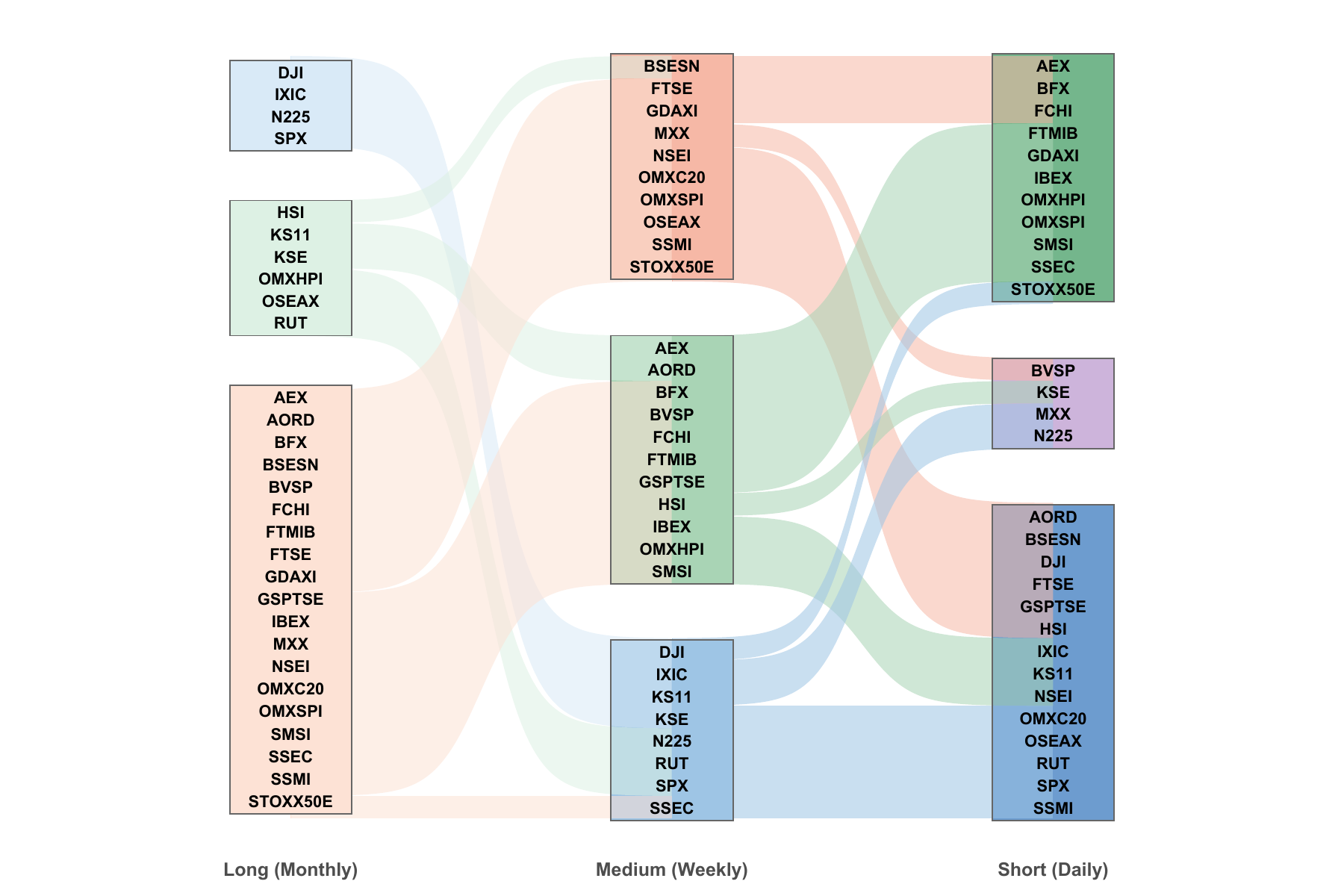}
	\vspace{-3mm}
	\caption{Long-to-short community flows for realized volatilities of 29 stock indices under the lasso-based ScBM--VHAR model and the baseline $3\to3\to3$ aggregate specification.}
	\label{fig:sanky_VHAR}
\end{figure}

Figure~\ref{fig:sanky_VHAR} reveals substantial horizon dependence rather than a single static regional partition. At the long horizon, DJI, IXIC, SPX, and N225 form a compact persistent block. A second group contains HSI, KS11, KSE, OMXHPI, OSEAX, and RUT, while the remaining markets form a broad Europe-heavy and international group. The long state should therefore be interpreted as persistent volatility-network organization rather than as a purely geographic partition.

At the medium horizon, geographic segmentation weakens markedly. DJI, IXIC, SPX, N225, KS11, KSE, RUT, and SSEC share one community, mixing major U.S. and Asian markets. This cross-regional grouping is more naturally interpreted as a common medium-horizon transmission role than as regional proximity. Korea illustrates the shift: KS11 follows the path $2\to1\to1$, so its long-horizon organization differs from the U.S. large-cap block but its medium- and short-horizon roles align with a broad global core.

At the short horizon, regional differentiation partially re-emerges. Continental European markets including AEX, BFX, FCHI, FTMIB, GDAXI, IBEX, and STOXX50E are concentrated in one short-run community, whereas several U.S. and Asian markets occupy different groups. Trading-hour overlap and market microstructure are plausible economic mechanisms for this short-horizon separation, although the clustering results themselves do not identify those mechanisms.

Figure~\ref{fig:vhar_sankey_343} in the Supplementary Material provides a useful specification check. Increasing only the middle state count from three to four leaves the long- and short-horizon partitions unchanged while refining the medium-horizon partition. Thus the extra heterogeneity is concentrated at the intermediate horizon rather than reflecting instability of the overall community structure. The empirical results therefore support a layered interpretation: persistent organization at long horizons, stronger cross-regional mixing and differentiation at the medium horizon, and partial regional re-segmentation at the short horizon.

\section{Conclusion}
\label{se:conclusion}

This paper develops a framework for recovering latent community paths in VAR-type models. ScBMs characterize latent directional organization, spectral co-clustering distinguishes sending and receiving groups, and singular-vector smoothing links community structures across seasons or dependence horizons. The framework is developed for ScBM--PVAR and generalized ScBM--VHAR models, with theoretical results connecting transition-matrix estimation error to singular-subspace perturbation and community misclassification.

The simulation studies illustrate how recovery depends on sample size, signal separation, path complexity, and first-stage regularization. Recovery generally improves with sample lengths, while higher dimension, stronger cross-block interaction, and more complex split--merge or refinement paths make community recovery more challenging. The misspecification experiments further show that incorrect specification of the community path can impose a recovery ceiling, while unstructured contamination gradually weakens the underlying block signal. In the applications, U.S. nonfarm payrolls exhibit substantial seasonal reorganization, whereas global realized volatilities display persistent long-horizon structure together with more pronounced reorganization at shorter horizons.

Several limitations should be acknowledged. First, the number of communities at each stage is selected from spectral diagnostics and then treated as fixed, so uncertainty in community number selection is not propagated to the estimated paths. Second, the ScBM only approximates the transition structure, and recovery may deteriorate when block separation is weak or first-stage estimation error is large. Third, the PVAR path-continuity assumption requires exact matching between adjacent stages and therefore does not accommodate abrupt unmatched changes in community structure. Future research may consider data-driven selection of community numbers, more flexible matching across stages, change-point extensions, approximate block structures, and uncertainty quantification for estimated community paths.

\section*{Acknowledgement}
CB was supported by the National Research Foundation of Korea grant funded by the Korean government (MSIT) (RS-2025-00519717). The authors thank Drs. Francis X. Diebold, Majid Al-Sadoon, and Vladas Pipiras for their comments at the 2025 NBER–NSF Time Series Conference, which substantially improved the quality of this paper. The authors are also grateful to Drs. Gu{\dh}mundur Stef{'a}n Gu{\dh}mundsson and Christian Brownlees for insightful discussions regarding the models.

\section*{Conflicts of Interest}
The authors declare no conflicts of interest.

\section*{Data Availability Statement}
The R code for the simulation study and data applications is available at \url{https://github.com/crbaek/dynamic-network-clustering}.

{\small
\renewcommand{\baselinestretch}{.8}
\setlength{\bibsep}{4pt}

}

\clearpage
\baselineskip17pt

\appendix

\setcounter{section}{0}
\setcounter{page}{1}
\setcounter{figure}{0}
\setcounter{table}{0}
\setcounter{equation}{0}
\renewcommand{\thesection}{S.\arabic{section}}
\renewcommand{\thepage}{\arabic{page}}
\renewcommand{\thetable}{S.\arabic{table}}
\renewcommand{\thefigure}{S.\arabic{figure}}
\renewcommand{\theequation}{S.\arabic{equation}}

\begin{center}
\textbf{\LARGE {Supplementary Material for \\ ``Latent community paths in VAR-type models via dynamic directed spectral co-clustering}''}
\end{center}    

\begin{center}
Younghoon Kim$^{1}$ and Changryong Baek$^{2*}$
\end{center}

\begin{center}
$^1$Department of Mathematics, University of Alabama \\
$^2$Department of Statistics, Sungkyunkwan University
\end{center}

\def\thefootnote{$*$}\footnotetext[1]{Corresponding author. Department of Statistics, Sungkyunkwan University, 25-2, Sungkyunkwan-ro, Jongno-gu, Seoul, Korea 03063, crbaek@skku.edu}

The proofs collected in this Supplementary Material serve a central theoretical purpose: establishing singular-vector perturbation bounds and non-asymptotic misclassification bounds for the proposed ScBM--PVAR and generalized ScBM--VHAR frameworks. Our theoretical development proceeds in three main steps. First, we derive estimation error bounds for the regularized lasso estimators and combine them with random-graph perturbation arguments to bound the operator norm of the block matrices used in directed spectral co-clustering. Second, we adapt the PisCES perturbation framework to ordered chains of left and right singular-vector projectors, enabling us to track structural evolution across seasons or horizons. This links the first-stage estimation error to singular-vector perturbations and, in turn, to the final clustering error. Finally, we provide a simple sufficient setting that justifies Assumption~\ref{assum:cluster_regular}.

\section{Proofs of theoretical properties} 
\label{ap:appendix}

\subsection{Proofs for Lemmas in Section \ref{sse:stability}}

\noindent{\textbf{Proof of Lemma \ref{lem:stability}}.} Consider the ScBM--PVAR model in the stacked VAR$(p^*)$ form \eqref{eq:PVAR_period_cycle_simple}. Define the companion state vector
\begin{displaymath}
    S_n
    = (Y_n^{*'},Y_{n-1}^{*'},\ldots,Y_{n-p^*+1}^{*'})'
    \in \mathbb{R}^{qsp^*},
    \qquad
    \eta_n = (\varepsilon_n^{*'},0',\ldots,0')'.
\end{displaymath}
Then
\begin{displaymath}
    S_n = F S_{n-1} + \eta_n,
    \qquad
    F = \begin{pmatrix}
        \Psi_1^* & \Psi_2^* & \ldots & \Psi_{p^*-1}^* & \Psi_{p^*}^* \\
        I_{qs} & 0 & \ldots & 0 & 0 \\
        0 & I_{qs} & \ldots & 0 & 0 \\
        \vdots & \vdots & \ddots & \vdots & \vdots \\
        0 & 0 & \ldots & I_{qs} & 0
    \end{pmatrix}.
\end{displaymath}
Let $\lambda$ be any eigenvalue of $F$. Standard companion-matrix algebra \cite[e.g., Chapter~2 in][]{lutkepohl2005new_supp} shows that
\begin{displaymath}
    \det\!\Bigl(
    \lambda^{p^*}I_{qs}
    - \lambda^{p^*-1}\Psi_1^*
    - \cdots
    - \lambda \Psi_{p^*-1}^*
    - \Psi_{p^*}^*
    \Bigr) = 0.
\end{displaymath}
If $|\lambda|\ge 1$, then with $z=\lambda^{-1}$ we have $|z|\le 1$ and therefore
\begin{displaymath}
    0
    =
    \det\!\Bigl(
    I_{qs}
    - \Psi_1^* z
    - \cdots
    - \Psi_{p^*}^* z^{p^*}
    \Bigr),
\end{displaymath}
which contradicts Assumption~\ref{assum:PVAR_stability_simple}(i). Hence every eigenvalue of $F$ lies strictly inside the unit disk, so $\rho(F)<1$. Consequently,
\begin{displaymath}
    S_n = \sum_{\ell=0}^{\infty} F^{\ell}\eta_{n-\ell}
\end{displaymath}
converges absolutely and defines the unique causal stationary solution. 
Since $Y_n^*$ is the first $qs$-dimensional block of $S_n$, the stacked ScBM--PVAR process admits a unique causal stationary solution. This argument is applied to each finite-dimensional member of the model sequence; the uniform margin in Assumption~\ref{assum:PVAR_stability_simple}(ii) is used to keep the subsequent constants uniform as $q$ changes. \qed

\medskip
\noindent{\textbf{Proof of Lemma \ref{cor:stability}}.} Let $S_t = (Y_t',Y_{t-1}',\ldots,Y_{t-b_L+1}')' \in \mathbb{R}^{qb_L}$ and define the companion matrix
\begin{displaymath}
	F=
	\begin{pmatrix}
		\Phi_1 & \Phi_2 & \cdots & \Phi_{b_L-1} & \Phi_{b_L} \\
		I_q & 0 & \cdots & 0 & 0 \\
		0 & I_q & \cdots & 0 & 0 \\
		\vdots & \vdots & \ddots & \vdots & \vdots \\
		0 & 0 & \cdots & I_q & 0
	\end{pmatrix},
\end{displaymath}
where $\Phi_1,\ldots,\Phi_{b_L}$ are the constrained VAR($b_L$) coefficients in \eqref{eq:VHAR_coefficients}. Then the generalized ScBM--VHAR model can be written as
\begin{displaymath}
	S_t = FS_{t-1} + \eta_t,
	\qquad
	\eta_t = (\varepsilon_t',0',\ldots,0')'.
\end{displaymath}
By the standard companion-form identity \cite[e.g., Chapter 2 in][]{lutkepohl2005new_supp},
\begin{displaymath}
	\det(I_{qb_L}-Fz)
	=
	\det\!\left(I_q-\sum_{h=1}^{b_L}\Phi_h z^h\right),
	\qquad z\in\mathbb{C}.
\end{displaymath}
Hence Assumption~\ref{assum:VHAR_stability}(i) implies $\rho(F)<1$. Therefore the linear process
\begin{displaymath}
	S_t = \sum_{\ell=0}^{\infty} F^\ell \eta_{t-\ell}
\end{displaymath}
converges absolutely and defines a unique causal stationary solution. Since $Y_t$ is the first $q$-dimensional block of $S_t$, the generalized ScBM--VHAR model admits a unique causal stationary solution. The argument is member-wise in $N$, while Assumption~\ref{assum:VHAR_stability}(ii) supplies the uniform margin used in the high-dimensional bounds. \qed

\medskip
\noindent{\textbf{Proof of Proposition \ref{prop:primitive_mixing}}.}
We first show that the uniform norm margins in the stability assumptions imply a uniform power bound for the companion matrices. Write $d=p^*$ and $G_h=\Psi_h^*$ for PVAR, or $d=b_L$ and $G_h=\Phi_h$ for VHAR. Let $\delta$ denote the corresponding uniform margin, so that $\sum_{h=1}^d\|G_h\|\le1-\delta$. When $d=1$, the result follows directly from $\|G_1\|\le1-\delta$. For $d\ge2$, choose
\begin{displaymath}
    r\in\bigl((1-\delta)^{1/(d-1)},1\bigr)
\end{displaymath}
and, for a companion-state vector
$x=(x_0^\prime,\ldots,x_{d-1}^\prime)^\prime$, define the weighted block norm
\begin{displaymath}
    \|x\|_r:=\max_{0\le j\le d-1}r^j\|x_j\|_2.
\end{displaymath}
The first block of $F_Nx$ satisfies
\begin{displaymath}
    \left\|\sum_{h=1}^dG_hx_{h-1}\right\|_2
    \le
    \sum_{h=1}^d\|G_h\|r^{-(h-1)}\|x\|_r
    \le
    (1-\delta)r^{-(d-1)}\|x\|_r,
\end{displaymath}
whereas every shifted block is bounded by $r\|x\|_r$. Hence
\begin{displaymath}
    \|F_Nx\|_r\le\rho_F\|x\|_r,
    \qquad
    \rho_F:=\max\{r,(1-\delta)r^{-(d-1)}\}<1.
\end{displaymath}
Because $d$ is fixed, $\|\cdot\|_r$ and the Euclidean norm are uniformly equivalent:
\begin{displaymath}
    \|x\|_r\le\|x\|_2
    \le\sqrt d\,r^{-(d-1)}\|x\|_r.
\end{displaymath}
It follows that, for a constant $C_F<\infty$ independent of $N$,
\begin{displaymath}
    \|F_N^\ell\|\le C_F\rho_F^\ell,
    \qquad \ell\ge0.
\end{displaymath}

For each $N$, the stacked PVAR and VHAR companion representations are finite-order stationary vector ARMA processes. Under Gaussian innovations, the innovation distribution is absolutely continuous. A stable vector ARMA process is geometrically completely regular, and hence geometrically strongly mixing \citep{mokkadem1988mixing_supp}. The power bound above, together with the uniform covariance bounds in parts (i) and (ii) of Proposition~\ref{prop:primitive_mixing}, makes the geometric constants uniform along the model sequence. Since strong-mixing coefficients do not increase under measurable transformations, the observed stacked PVAR process and the VHAR companion process inherit
\begin{displaymath}
    \alpha(\ell)\le C\rho^\ell\le\exp(-c\ell)
\end{displaymath}
for constants $C<\infty$, $\rho\in(0,1)$, and $c>0$ independent of $N$, after adjusting the bound for finitely many small lags.

Moreover, every unit projection of the covariance-normalized Gaussian state is standard normal. Hence, for a constant $c_2>0$ independent of $N$,
\begin{displaymath}
    \mathbb{P}\!\left(
    \left|v^\prime\Sigma_S^{-1/2}S_t\right|>x
    \right)
    \le2\exp(-x^2/2)
    \le\exp\!\left\{1-(x/c_2)^2\right\}.
\end{displaymath}
The innovation conditions in Assumption~\ref{assum:pvar_simple}(i)--(ii) follow directly from the Gaussian white-noise specification and the assumed covariance eigenvalue bounds. Thus, Assumption~\ref{assum:pvar_simple}(i)--(iv), or its VHAR counterpart, holds with constants independent of $N$, with $\gamma_1=1$ and $\gamma_2=2$. \qed

\subsection{Proofs for Lemmas in Section \ref{sse:consistency_estimator}}

\noindent{\textbf{Proof of Lemma \ref{lem:estimation}}.}
Let $N_0=N-p^*$. Write
\begin{displaymath}
    \mathbb Y_P=\mathbb X_PB_P+\mathbb E_P,
    \qquad
    B_P=(\Psi_1^{*\prime},\ldots,\Psi_{p^*}^{*\prime})^\prime,
\end{displaymath}
with $\mathbb Y_P$ and $\mathbb X_P$ as defined in Section~\ref{sse:algorithm_estimation}. Define
\begin{displaymath}
    \widehat\Sigma_{X,P}:=N_0^{-1}\mathbb X_P^\prime\mathbb X_P,
    \qquad
    \widehat\Sigma_{YX,P}:=N_0^{-1}\mathbb Y_P^\prime\mathbb X_P,
\end{displaymath}
and let
\begin{displaymath}
    \Sigma_{X,P}:=\mathbb E(W_n^*W_n^{*\prime}),
    \qquad
    \Sigma_{YX,P}:=\mathbb E(Y_n^*W_n^{*\prime}).
\end{displaymath}
The OLS estimator is
\begin{displaymath}
    \widehat B_P=\widehat\Sigma_{X,P}^{-1}\widehat\Sigma_{YX,P}^\prime.
\end{displaymath}
Let
\begin{displaymath}
    S_n=(Y_n^{*\prime},Y_{n-1}^{*\prime},\ldots,
    Y_{n-p^*+1}^{*\prime})^\prime\in\mathbb R^{qsp^*}
\end{displaymath}
be the companion state vector of the stable stacked VAR$(p^*)$. The regressor--response vector $(Y_n^{*\prime},W_n^{*\prime})^\prime$ is a deterministic coordinate-selection and stacking transformation of $S_n$. Since $p^*$ is fixed, Assumptions~\ref{assum:PVAR_stability_simple} and \ref{assum:pvar_simple} imply the mixing, tail, and covariance conditions of Assumption~3(i)--(v) in
\citet{gudhmundsson2021detecting_supp}, with cross-sectional dimension $qs$ and sample size $N_0$. Hence, their Lemma~A.2 gives 
\begin{equation}\label{eq:lema2-gud}
    \|\widehat\Sigma_{X,P}-\Sigma_{X,P}\|
    =\mathcal O_{\mathbb P}\!\left(\sqrt{\frac{qs}{N}}\,\|\Sigma_{X,P}\|\right),
    \qquad
    \|\widehat\Sigma_{YX,P}-\Sigma_{YX,P}\|
    =\mathcal O_{\mathbb P}\!\left(\sqrt{\frac{qs}{N}}\right),
\end{equation}
where $N_0\asymp N$ because $p^*$ is fixed. In particular, $\lambda_{\min}(\widehat\Sigma_{X,P})$ is bounded away from zero with high probability. The inverse-perturbation identity yields
\begin{displaymath}
    \|\widehat\Sigma_{X,P}^{-1}-\Sigma_{X,P}^{-1}\|
    \le
    \|\widehat\Sigma_{X,P}^{-1}\|
    \|\widehat\Sigma_{X,P}-\Sigma_{X,P}\|
    \|\Sigma_{X,P}^{-1}\|.
\end{displaymath}
Combining this bound with \eqref{eq:lema2-gud} gives
\begin{displaymath}
    \|\widehat B_P-B_P\|
    =\mathcal O_{\mathbb P}\!\left(
    \sqrt{\frac{qs}{N}}\,
    \|\Sigma_{X,P}\|\|\Sigma_{X,P}^{-1}\|\|B_P\|
    \right).
\end{displaymath}
The uniform stability and covariance conditions, together with \citet[Lemma~OA.13]{gudhmundsson2021detecting_supp}, bound $\|\Sigma_{X,P}\|$, $\|\Sigma_{X,P}^{-1}\|$, and $\|B_P\|$ uniformly. Since $s$ and $p^*$ are fixed,
\begin{displaymath}
    \|\widehat B_P-B_P\|
    =\mathcal O_{\mathbb P}\!\left(\sqrt{\frac{q}{N}}\right).
\end{displaymath}
Finally, the season-stacked co-clustering matrix is obtained from $B_P$ by a locally Lipschitz recovery map $L_P$ on the neighborhood where the seasonal zero-lag block is uniformly invertible. Its Lipschitz modulus $\operatorname{Lip}(L_P)$ is $\mathcal O(1)$. Therefore
\begin{displaymath}
    \|\widehat\Phi^{\mathrm{ols}}-\Phi\|
    \le \operatorname{Lip}(L_P)\,\|\widehat B_P-B_P\|
    =\mathcal O_{\mathbb P}\!\left(\sqrt{\frac{q}{N}}\right).
\end{displaymath}
\qed

\medskip
\noindent{\textbf{Proof of Lemma \ref{cor:estimation_vhar}}.}
Recall that $Z=X_eB+E$. Let
\begin{displaymath}
    x_t=(Y_{t-1}^\prime,\ldots,Y_{t-b_L}^\prime)^\prime\in\mathbb R^{qb_L},
    \qquad
    x_t^e=R_{b_M,b_L}^\prime x_t\in\mathbb R^{3q}.
\end{displaymath}
The rows of $X_e$ are $(x_{b_L+1}^e,\ldots,x_T^e)^\prime$. Because $b_M$ and $b_L$ are fixed, $R_{b_M,b_L}$ is deterministic with bounded operator norm. Thus $(Y_t^\prime,(x_t^e)^\prime)^\prime$ is a fixed linear transformation of the stable VAR$(b_L)$ companion state. Assumptions~\ref{assum:VHAR_stability} and \ref{assum:vhar_simple} therefore imply the conditions required by Lemma~A.2 of \citet{gudhmundsson2021detecting_supp}.

Define
\begin{displaymath}
    \widehat\Sigma_{X_e}:=N^{-1}X_e^\prime X_e,
    \qquad
    \widehat\Sigma_{X_eZ}:=N^{-1}X_e^\prime Z,
\end{displaymath}
and
\begin{displaymath}
    \Sigma_{X_e}:=\mathbb E[x_t^e(x_t^e)^\prime],
    \qquad
    \Sigma_{X_eZ}:=\mathbb E[x_t^eY_t^\prime].
\end{displaymath}
Then
\begin{displaymath}
    \|\widehat\Sigma_{X_e}-\Sigma_{X_e}\|
    =\mathcal O_{\mathbb P}\!\left(\sqrt{\frac{q}{N}}\,\|\Sigma_{X_e}\|\right),
    \qquad
    \|\widehat\Sigma_{X_eZ}-\Sigma_{X_eZ}\|
    =\mathcal O_{\mathbb P}\!\left(\sqrt{\frac{q}{N}}\right).
\end{displaymath}
Consequently, the restricted OLS estimator $\widehat B=\widehat\Sigma_{X_e}^{-1}\widehat\Sigma_{X_eZ}$ exists with high probability, and the same inverse-perturbation argument gives
\begin{displaymath}
    \|\widehat B-B\|
    =\mathcal O_{\mathbb P}\!\left(
    \sqrt{\frac{q}{N}}\,
    \|\Sigma_{X_e}\|\|\Sigma_{X_e}^{-1}\|\|B\|
    \right).
\end{displaymath}
The covariance factors and $\|B\|$ are uniformly bounded under the fixed-horizon stability conditions. Hence
\begin{displaymath}
    \|\widehat B-B\|=\mathcal O_{\mathbb P}\!\left(\sqrt{\frac{q}{N}}\right).
\end{displaymath}
Since the stacked horizon-specific matrix contains the fixed number $s=3$ of blocks,
\begin{displaymath}
    \|\widehat\Phi^{\mathrm{ols}}-\Phi\|
    =\mathcal O_{\mathbb P}\!\left(\sqrt{\frac{sq}{N}}\right).
\end{displaymath}
\qed

\medskip
\noindent{\textbf{Proof of Lemma \ref{lem:estimation_lasso}}.} \citet[Proposition~4.1]{basu2015regularized_supp} states that any solution of the penalized VAR problem obeys the $\ell_2$-error bound once the sample Gram matrix satisfies the RE condition in equation~(4.7) and the score satisfies the DB in equation~(4.8) therein. Under Assumption~\ref{assum:pvar_lasso}, these conditions hold for the stacked ScBM--PVAR regression, so
\begin{displaymath}
	\|\hat\alpha_P-\alpha_P^*\|_2
	\le C\sqrt{\ell_{P,N}}\,\lambda_{N,P}
\end{displaymath}
with probability tending to one, for a constant $C>0$ depending only on the RE constants. If one wishes to verify Assumption~\ref{assum:pvar_lasso} probabilistically under a Gaussian stable stacked VAR, Proposition~4.2 and Proposition~4.3 of \citet{basu2015regularized_supp} apply directly with process dimension $p=qs$ and lag order $d=p^*$.

Let $L_P$ denote the locally Lipschitz recovery map that extracts the seasonal blocks from $(\Psi_1^*,\ldots,\Psi_{p^*}^*)$ and maps them to the co-clustering matrix $\Phi=(\Phi_1',\ldots,\Phi_s')'$. The map includes the uniformly invertible zero-lag transformation and has Lipschitz modulus $\operatorname{Lip}(L_P)=\mathcal O(1)$ for fixed $s$ and $p^*$. Therefore,
\begin{displaymath}
	\|\hat\Phi^{\mathrm{lasso}}-\Phi\|
	= \|L_P(\hat\alpha_P)-L_P(\alpha_P^*)\|
	\le \operatorname{Lip}(L_P)\,\|\hat\alpha_P-\alpha_P^*\|_2
	= \mathcal{O}_{\mathbb P}\!\left(\sqrt{\ell_{P,N}}\,\lambda_{N,P}\right),
\end{displaymath}
which is the stated rate with $\lambda_{N,P}\asymp\sqrt{\log(qs)/N}$. \qed

\medskip
\noindent{\textbf{Proof of Lemma \ref{cor:estimation_vhar_lasso}}.} 
The restricted generalized ScBM--VHAR estimator in \eqref{eq:vhar_lasso} is an ordinary lasso problem with design matrix $(I_q\otimes X_e)$. \citet[Proposition~4.1]{basu2015regularized_supp} therefore yields the $\ell_2$-error bound once the RE condition and the corresponding DB hold for the restricted design; these are the direct analogues of equation (4.7) and equation (4.8) in \citet{basu2015regularized_supp}.

For the fixed-horizon generalized VHAR design, Proposition~1 of \citet{baek2021sparse_supp} verifies exactly these ingredients for the conventional $(5,22)$ specification: their equation (2.10) defines the transformed Gram matrix and score vector, and their equation (2.11) gives the RE and DB by reducing generalized VHAR to a stable high-order VAR and then invoking \citet[Propositions~4.2--4.3]{basu2015regularized_supp}. Replacing the specific aggregation matrix in their equation (2.12) by our deterministic $R_{b_M,b_L}$ changes only constants because $b_M$ and $b_L$ are fixed. Hence, the same argument gives
\begin{displaymath}
	\|\widehat{\beta}_V-\beta_V^*\|_2
	\le
	C\sqrt{\ell_{V,N}}\lambda_{N,V}
\end{displaymath}
for some constant $C$ with probability tending to one.

If $\widehat{B}$ denotes the matrix obtained by reshaping $\widehat{\beta}_V$, then the stacked horizon-specific matrix used in the co-clustering step is exactly $B$, so
\begin{displaymath}
	\|\widehat{\Phi}^{\,\mathrm{lasso}}-\Phi\|
	=
	\|\widehat{B}-B\|
	\le
	\|\widehat{\beta}_V-\beta_V^*\|_2
	= \mathcal{O}_{\mathbb{P}}\!\left(\sqrt{\ell_{V,N}}\lambda_{N,V}\right).
\end{displaymath}
From $s=3$ and $\lambda_{N,V}\asymp \sqrt{\log(sq)/N}$, we have the stated rate.  \qed

\subsection{Proof for Theorem in Section \ref{sse:consistency_random}}

\noindent{\textbf{Proof of Theorem \ref{lem:population}}.}
We first treat one generic weighted-adjacency block. Let $A$ be its realized nonnegative weighted adjacency matrix and let $\mathcal A=\mathbb E[A]$. Define
\begin{displaymath}
    O=\operatorname{diag}(A\mathbf 1_q),
    \qquad
    P=\operatorname{diag}(A^\prime\mathbf 1_q),
    \qquad
    \tau_A=q^{-1}\mathbf 1_q^\prime A\mathbf 1_q,
\end{displaymath}
and
\begin{displaymath}
    \mathcal O=\operatorname{diag}(\mathcal A\mathbf 1_q),
    \qquad
    \mathcal P=\operatorname{diag}(\mathcal A^\prime\mathbf 1_q),
    \qquad
    \bar\tau=q^{-1}\mathbf 1_q^\prime\mathcal A\mathbf 1_q.
\end{displaymath}
The realized- and population-normalized adjacency matrices are
\begin{displaymath}
    L_A=(O+\tau_AI_q)^{-1/2}A(P+\tau_AI_q)^{-1/2},
    \qquad
    \mathcal L=(\mathcal O+\bar\tau I_q)^{-1/2}
    \mathcal A(\mathcal P+\bar\tau I_q)^{-1/2}.
\end{displaymath}
If $\phi_A$ is the scalar multiplier for this block, its transition matrices are $\Phi_A=\phi_A L_A^\prime$ and $\widetilde\Phi_A=\phi_A\mathcal L^\prime$. Assumption~\ref{assum_degree_simple} ensures $|\phi_A|\le C_\phi$.

Set
\begin{displaymath}
    C_A=(\mathcal O+\bar\tau I_q)^{-1/2}
    A(\mathcal P+\bar\tau I_q)^{-1/2}.
\end{displaymath}
Then
\begin{equation}\label{eq:laplacian_bound}
    \|\mathcal L-L_A\|
    \le \|\mathcal L-C_A\|+\|C_A-L_A\|.
\end{equation}
Let
\begin{displaymath}
    c_*:=\min\!\left\{
    \min_i([\mathcal O]_{ii}+\bar\tau),
    \min_j([\mathcal P]_{jj}+\bar\tau)
    \right\}.
\end{displaymath}
By Assumption~\ref{assum_degree_simple}, $c_*=\Omega(sqB_{sq})$ and $c_*\gtrsim\log(sq)$. Let $w_{\max}$ be the maximum of the finitely many block-wise weight bounds in Section~2.

\smallskip
\noindent\emph{Step 1 (random adjacency under population normalization).}
Using the symmetric dilation, write $C_A^S-\mathcal L^S$ as a sum of independent, centered, self-adjoint matrices indexed by the dyads. Each summand has norm at most $2w_{\max}/c_*$. Moreover, since $0\le A_{ij}\le w_{\max}$,
\begin{displaymath}
    \operatorname{Var}(A_{ij})\le \mathbb E[A_{ij}^2]\le w_{\max}\mathbb E[A_{ij}]=w_{\max}[\mathcal A]_{ij}.
\end{displaymath}
The row--column degree pairing in $\mathcal L$ gives the matrix-variance bound
\begin{displaymath}
    v^2\le \frac{C_w}{c_*},
\end{displaymath}
where $C_w<\infty$ depends only on the common weight bound. Matrix Bernstein's inequality \citep[Theorem~6.1.1]{vershynin2018high_supp} therefore yields
\begin{displaymath}
    \|C_A-\mathcal L\|
    =\mathcal O_{\mathbb P}\!\left(
    \sqrt{\frac{\log(sq)}{c_*}}+
    \frac{\log(sq)}{c_*}
    \right)
    =\mathcal O_{\mathbb P}\!\left(
    \sqrt{\frac{\log(sq)}{sqB_{sq}}}
    \right).
\end{displaymath}
For later use, the correctly paired population normalization also satisfies $\|\mathcal L\|\le1$. Indeed, for unit vectors $x,y$, the inequality $2|ab|\le a^2+b^2$ gives
\begin{displaymath}
    |y^\prime\mathcal Lx|
    \le \frac12\sum_i
    \frac{[\mathcal O]_{ii}}{[\mathcal O]_{ii}+\bar\tau}y_i^2
    +\frac12\sum_j
    \frac{[\mathcal P]_{jj}}{[\mathcal P]_{jj}+\bar\tau}x_j^2
    \le1.
\end{displaymath}
Hence $\|C_A\|=\mathcal O_{\mathbb P}(1)$.

\smallskip
\noindent\emph{Step 2 (realized degrees and realized regularizer).}
Scalar Bernstein inequalities and a union bound give the relative degree bounds
\begin{displaymath}
    \max_i
    \frac{|[O]_{ii}-[\mathcal O]_{ii}|}{[\mathcal O]_{ii}+\bar\tau}
    +
    \max_j
    \frac{|[P]_{jj}-[\mathcal P]_{jj}|}{[\mathcal P]_{jj}+\bar\tau}
    =
    \mathcal O_{\mathbb P}\!\left(
    \sqrt{\frac{\log(sq)}{c_*}}+
    \frac{\log(sq)}{c_*}
    \right).
\end{displaymath}
The regularizer satisfies
\begin{displaymath}
    \tau_A-\bar\tau
    =q^{-1}\sum_{i,j}(A_{ij}-[\mathcal A]_{ij}).
\end{displaymath}
A scalar Bernstein bound, together with $\sum_{i,j}[\mathcal A]_{ij}=q\bar\tau$ and the bounded weights, gives
\begin{displaymath}
    |\tau_A-\bar\tau|
    =
    \mathcal O_{\mathbb P}\!\left(
    \sqrt{\frac{\bar\tau\log(sq)}{q}}+
    \frac{\log(sq)}{q}
    \right).
\end{displaymath}
Since $c_*\ge\bar\tau$ by definition and $\log(sq)/c_*\to0$, it follows that $|\tau_A-\bar\tau|/c_*=\mathcal O_{\mathbb P} (\sqrt{\log(sq)/c_*})$. Consequently,
\begin{displaymath}
    \max_i\left|
    \frac{[O]_{ii}+\tau_A}{[\mathcal O]_{ii}+\bar\tau}-1
    \right|
    +
    \max_j\left|
    \frac{[P]_{jj}+\tau_A}{[\mathcal P]_{jj}+\bar\tau}-1
    \right|
    =
    \mathcal O_{\mathbb P}\!\left(
    \sqrt{\frac{\log(sq)}{c_*}}
    \right).
\end{displaymath}
Define the diagonal matrices
\begin{displaymath}
    R_O=(O+\tau_AI_q)^{-1/2}(\mathcal O+\bar\tau I_q)^{1/2},
    \qquad
    R_P=(P+\tau_AI_q)^{-1/2}(\mathcal P+\bar\tau I_q)^{1/2}.
\end{displaymath}
Then $L_A=R_OC_AR_P$ and
\begin{displaymath}
    \|R_O-I_q\|+\|R_P-I_q\|
    =
    \mathcal O_{\mathbb P}\!\left(
    \sqrt{\frac{\log(sq)}{c_*}}
    \right).
\end{displaymath}
Using $\|C_A\|=\mathcal O_{\mathbb P}(1)$,
\begin{displaymath}
    \|L_A-C_A\|
    \le
    \|R_O-I_q\|\,\|C_A\|\,\|R_P\|
    +\|C_A\|\,\|R_P-I_q\|
    =
    \mathcal O_{\mathbb P}\!\left(
    \sqrt{\frac{\log(sq)}{sqB_{sq}}}
    \right).
\end{displaymath}

\smallskip
\noindent\emph{Step 3 (transition matrices and finite aggregation).}
Combining the preceding bounds with \eqref{eq:laplacian_bound},
\begin{displaymath}
    \|\widetilde\Phi_A-\Phi_A\|
    =|\phi_A|\,\|\mathcal L-L_A\|
    =
    \mathcal O_{\mathbb P}\!\left(
    \sqrt{\frac{\log(sq)}{sqB_{sq}}}
    \right).
\end{displaymath}
The numbers of seasons, horizons, and lag terms are fixed. Taking a union bound over the finitely many blocks and applying the fixed aggregation maps used to form the season- or horizon-stacked matrices changes only the constant. Therefore
\begin{displaymath}
    \|\widetilde\Phi-\Phi\|
    =
    \mathcal O_{\mathbb P}\!\left(
    \sqrt{\frac{\log(sq)}{sqB_{sq}}}
    \right)
\end{displaymath}
for both ScBM--PVAR and generalized ScBM--VHAR. \qed

\subsection{Proof for Theorems in Section \ref{sse:misclassification} and technical perturbation Lemmas}

Throughout, all singular-vector bounds are understood after suitable orthogonal rotations.

\begin{proposition}[A simple sufficient setting for Assumption~\ref{assum:cluster_regular}]
	\label{prop:ass48-sufficient}
	Fix one seasonal or horizon-specific population matrix entering the co-clustering step, and suppress its season/horizon index. Let
	\begin{displaymath}
    	\Omega_q
    	=
    	\widetilde{\Theta}_{y,q}^{1/2} \Xi_q^y B_q (\Xi_q^z)^\prime \widetilde{\Theta}_{z,q}^{1/2},
	\end{displaymath}
	where $\Xi_q^y \in \{0,1\}^{q \times K_y}$ and $\Xi_q^z \in \{0,1\}^{q \times K_z}$ are the sending and receiving membership matrices, $K_y$ and $K_z$ are fixed positive integers, $B_q \in \mathbb{R}^{K_y \times K_z}$ is a block matrix, and $\widetilde{\Theta}_{y,q}$ and $\widetilde{\Theta}_{z,q}$ are diagonal matrices of effective degree-correction weights induced by the regularized population normalization. Define
	\begin{displaymath}
	   D_{y,q} := (\Xi_q^y)^\prime \widetilde{\Theta}_{y,q} \Xi_q^y,
	   \qquad
	   D_{z,q} := (\Xi_q^z)^\prime \widetilde{\Theta}_{z,q} \Xi_q^z,
	   \qquad
	   M_q := D_{y,q}^{1/2} B_q D_{z,q}^{1/2}.
	\end{displaymath}
	Let $K := \min(K_y,K_z)$. Assume the following.
	\begin{itemize}
		\item[(i)] $K_y$ and $K_z$ are fixed.
		
		\item[(ii)] (Balanced communities) There exist constants $0 < c_1 < C_1 < \infty$ such that for every sending block $k$ and receiving block $\ell$,
		\begin{displaymath}
		      c_1 q \le n_{y,k,q} \le C_1 q,
		      \qquad
		      c_1 q \le n_{z,\ell,q} \le C_1 q,
		\end{displaymath}
		where $n_{y,k,q}$ and $n_{z,\ell,q}$ denote the corresponding block sizes.
		
		\item[(iii)] (Bounded effective degree correction) There exist constants $0 < c_2 < C_2 < \infty$ such that
		\begin{displaymath}
		      c_2 q^{-1} \le [\widetilde{\Theta}_{y,q}]_{ii} \le C_2 q^{-1},
		      \qquad
		      c_2 q^{-1} \le [\widetilde{\Theta}_{z,q}]_{jj} \le C_2 q^{-1},
		\end{displaymath}
		for all $i,j=1,\dots,q$.
		
		\item[(iv)] $D_{y,q} \to D_{y,*}$ and $D_{z,q} \to D_{z,*}$ for some positive diagonal matrices $D_{y,*}$ and $D_{z,*}$.
		
		\item[(v)] $M_q \to M_*$ for some matrix $M_* \in \mathbb{R}^{K_y \times K_z}$ of rank $K$, and the $K$ positive singular values of $M_*$ are simple.
		
		\item[(vi)] (Bounded block separation) The row-normalized rows of $D_{y,*}^{-1/2} R_{y,*}$ are pair-wise distinct, and the row-normalized rows of $D_{z,*}^{-1/2} R_{z,*}$ are pair-wise distinct, where $R_{y,*}$ and $R_{z,*}$ denote the left and right singular-vector matrices of $M_*$.
	\end{itemize}
	Then Assumption~\ref{assum:cluster_regular} holds for $\Omega_q$ for all sufficiently large $q$.
\end{proposition}

\begin{proof}
	Write
	\begin{displaymath}
    	Q_{y,q} := \widetilde{\Theta}_{y,q}^{1/2} \Xi_q^y D_{y,q}^{-1/2},
    	\qquad
    	Q_{z,q} := \widetilde{\Theta}_{z,q}^{1/2} \Xi_q^z D_{z,q}^{-1/2}.
	\end{displaymath}
	Since $D_{y,q} = (\Xi_q^y)^\prime \widetilde{\Theta}_{y,q} \Xi_q^y$ and $D_{z,q} = (\Xi_q^z)^\prime \widetilde{\Theta}_{z,q} \Xi_q^z$, we have
	\begin{displaymath}
    	Q_{y,q}^\prime Q_{y,q} = I_{K_y},
    	\qquad
    	Q_{z,q}^\prime Q_{z,q} = I_{K_z},
	\end{displaymath}
	in turn, $\Omega_q	=	Q_{y,q} M_q Q_{z,q}^\prime$. Let the singular value decomposition of $M_q$ be
	\begin{displaymath}
	   M_q = R_{y,q} \Sigma_q R_{z,q}^\prime,
	\end{displaymath}
	where $R_{y,q} \in \mathbb{R}^{K_y \times K}$ and $R_{z,q} \in \mathbb{R}^{K_z \times K}$ have orthonormal columns, and
	\begin{displaymath}
	   \Sigma_q = \operatorname{diag}\{\sigma_1(M_q),\dots,\sigma_K(M_q)\}.
	\end{displaymath}
	Then the left and right singular-vector matrices of $\Omega_q$ corresponding to its $K$ nonzero singular values can be written as
	\begin{displaymath}
    	X_{L,q} = Q_{y,q} R_{y,q},
    	\qquad
    	X_{R,q} = Q_{z,q} R_{z,q}.
	\end{displaymath}
	
	We first verify Assumption~\ref{assum:cluster_regular}(i). Because $Q_{y,q}$ and $Q_{z,q}$ have orthonormal columns, the nonzero singular values of $\Omega_q$ are exactly those of $M_q$. By Assumption~(v), $M_q \to M_*$ and $M_*$ has rank $K$ with simple positive singular values. Hence
	\begin{displaymath}
	   \sigma_K(M_q) \to \sigma_K(M_*) > 0,
	\end{displaymath}
	and, for all sufficiently large $q$,
	\begin{displaymath}
    	\sigma_K(\Omega_q) = \sigma_K(M_q) \ge \frac{1}{2}\sigma_K(M_*),
    	\qquad
    	\sigma_{K+1}(\Omega_q)=0.
	\end{displaymath}
	Thus the retained singular subspace is separated from the remainder by a positive eigengap bounded away from zero.
	
	We next verify Assumption~\ref{assum:cluster_regular}(ii). Let $i$ belong to sending block $k$. Since the $i$th row of $\Xi_q^y$ equals $e_k^\prime$, the $i$th row of $X_{L,q}$ is
	\begin{displaymath}
    	[X_{L,q}]_{i\cdot}
    	=
    	[\widetilde{\Theta}_{y,q}]_{ii}^{1/2} e_k^\prime D_{y,q}^{-1/2} R_{y,q}.
	\end{displaymath}
	Because $R_{y,q}$ has orthonormal columns and $D_{y,q} \to D_{y,*}$ with positive diagonal limit, the matrix $D_{y,q}^{-1/2} R_{y,q}$ has row norms bounded above and below by positive constants uniformly in $q$. Assumption~(iii) gives
	\begin{displaymath}
	   [\widetilde{\Theta}_{y,q}]_{ii}^{1/2} \asymp q^{-1/2},
	\end{displaymath}
	uniformly in $i$. Hence
	\begin{displaymath}
	   \|[X_{L,q}]_{i\cdot}\| \asymp q^{-1/2},
	\end{displaymath}
	uniformly over all $i$. The same argument applied to $X_{R,q}$ yields
	\begin{displaymath}
	   \|[X_{R,q}]_{j\cdot}\| \asymp q^{-1/2},
	\end{displaymath}
	uniformly over all $j$. Therefore Assumption~\ref{assum:cluster_regular}(ii) holds.
	
	Finally, we verify Assumption~\ref{assum:cluster_regular}(iii). For any $i$ in sending block $k$,
	\begin{displaymath}
    	[X_{L,q}^*]_{i\cdot}
    	=
    	\frac{e_k^\prime D_{y,q}^{-1/2} R_{y,q}}
    	{\|e_k^\prime D_{y,q}^{-1/2} R_{y,q}\|},
	\end{displaymath}
	so all nodes in the same sending block share the same row-normalized population centroid. Since $M_q \to M_*$ and the positive singular values of $M_*$ are simple, the singular-vector matrices $R_{y,q}$ and $R_{z,q}$ converge, up to the usual column-wise sign convention, to $R_{y,*}$ and $R_{z,*}$. Together with $D_{y,q} \to D_{y,*}$ and $D_{z,q} \to D_{z,*}$, this implies that the row-normalized sending and receiving centroids converge to the row-normalized rows of $D_{y,*}^{-1/2} R_{y,*}$ and $D_{z,*}^{-1/2} R_{z,*}$, respectively. By Assumption~(vi), these limiting centroids are pair-wise distinct. Since the numbers of communities are fixed, the minimum pair-wise distance between distinct centroids is therefore bounded away from zero for all sufficiently large $q$. Hence Assumption~\ref{assum:cluster_regular}(i)--(iii) holds for $\Omega_q$.
\end{proof}

\begin{lemma}\label{lem:app_pisces_projector}
Let $J$ be fixed, and let $U_j,\widehat U_j$ be rank-$r_j$ orthogonal projectors. For $0\le\alpha_N<\alpha_{\max}$, let $\widehat{\mathcal G}_{\alpha_N}$ and $\mathcal G_{\alpha_N}$ be the empirical and population linear-chain PisCES maps. Define
\begin{displaymath}
    \mathcal V(U)=\sum_{j=1}^{J-1}\|U_{j+1}-U_j\|_F.
\end{displaymath}
If $\bar U$ and $U^{(\alpha_N)}$ are their fixed points, then
\begin{displaymath}
    \sum_{j=1}^{J}\|\bar U_j-U_j\|_F
    \le
    C\left\{
    \sum_{j=1}^{J}\|\widehat U_j-U_j\|_F
    +\alpha_N\mathcal V(U)
\right\},
\end{displaymath}
where $C$ is uniform over $\alpha_N\in[0,\alpha_{\max})$.
\end{lemma}

\begin{proof}
The contraction argument follows Lemma~S1.4 of \citet{liu2018global_supp}.  With the product metric $d(V,W)=\sum_j\|V_j-W_j\|_F$, Weyl's inequality and the Davis--Kahan theorem \citep{davis1970rotation_supp} give
\begin{displaymath}
    d\{\widehat{\mathcal G}_{\alpha_N}(V),\widehat{\mathcal G}_{\alpha_N}(W)\}
    \le
    \rho_{\mathrm P}(\alpha_N)d(V,W),
    \qquad
    \rho_{\mathrm P}(\alpha_N)=\frac{2\sqrt2\,\alpha_N}{1-2\alpha_N}.
\end{displaymath}
Because $\rho_{\mathrm P}(\alpha_{\max})=1/2$, both empirical and population maps are contractions with a uniform margin on the admissible range. Comparing the two fixed points as in the standard Banach perturbation argument yields
\begin{displaymath}
    d\{\bar U,U^{(\alpha_N)}\}
    \le
    C\sum_{j=1}^{J}\|\widehat U_j-U_j\|_F.
\end{displaymath}

It remains to control the population smoothing bias. Let $n_j\in\{1,2\}$ be the number of neighbors of $j$ and write
\begin{displaymath}
    \Delta_j=\sum_{j'\sim j}(U_{j'}-U_j),
    \qquad
    M_j=(1+n_j\alpha_N)U_j+\alpha_N\Delta_j.
\end{displaymath}
The leading rank-$r_j$ projector of $(1+n_j\alpha_N)U_j$ is exactly $U_j$.  Moreover, $M_j$ is positive semidefinite and $\lambda_{r_j+1}(M_j)\le n_j\alpha_N$, so the relevant spectral separation is at least one.  The residual form of Davis--Kahan \citep{davis1970rotation_supp} therefore gives
\begin{displaymath}
    \|[\mathcal G_{\alpha_N}(U)]_j-U_j\|_F
    \le
    \sqrt2\,\alpha_N\|\Delta_j\|_F.
\end{displaymath}
Each adjacent difference is counted at most twice, whence
\begin{displaymath}
    \sum_{j=1}^{J}\|\Delta_j\|_F\le2\mathcal V(U).
\end{displaymath}
A second fixed-point perturbation step gives
$d\{U^{(\alpha_N)},U\}\le C\alpha_N\mathcal V(U)$, and the result follows by the triangle inequality.
\end{proof}

\begin{lemma}\label{lem:app_rownorm}
Let $X,\bar X\in\mathbb R^{q\times r}$ have orthonormal columns and let $R$ be orthogonal. Define $\mathfrak n(0)=0$ and normalize rows by $\mathfrak n(x)=x/\|x\|$ for $x\neq0$.  If
\begin{displaymath}
    \min_i\|X_{i\cdot}\|\ge c_0q^{-1/2},
\end{displaymath}
then, without any small-perturbation condition,
\begin{displaymath}
    \|\bar X^*-X^*R\|_F
    \le
    \frac{2\sqrt q}{c_0}\|\bar X-XR\|_F.
\end{displaymath}
\end{lemma}

\begin{proof}
For every $b\neq0$ and every $a$, including $a=0$,
\begin{displaymath}
    \|\mathfrak n(a)-\mathfrak n(b)\|
    \le
    \frac{2\|a-b\|}{\|b\|}.
\end{displaymath}
Applying this row by row with $b=(XR)_{i\cdot}$ and using $\|(XR)_{i\cdot}\|=\|X_{i\cdot}\|$ gives
\begin{displaymath}
    \|\bar X^*-X^*R\|_F^2
    \le
    4\sum_{i=1}^{q}\frac{\|(\bar X-XR)_{i\cdot}\|^2}{\|(XR)_{i\cdot}\|^2}
    \le
    \frac{4q}{c_0^2}\|\bar X-XR\|_F^2.
\end{displaymath}
A zero empirical row contains no directional information and is assigned by the deterministic tie-breaking convention used in clustering.
\end{proof}

\begin{lemma}\label{lem:app_misclass}
	Let $X^*\in\mathbb R^{q\times r}$ be a population row-normalized singular-vector matrix whose rows take $K$ distinct centroid values, and let $\bar X^*\in\mathbb R^{q\times r}$ be its empirical counterpart. Suppose the population centroids are separated by at least $\delta_*>0$. If $\bar{\mathcal{N}}$ denotes the set of misclustered nodes obtained by applying $k$-means to the rows of $\bar X^*$, then there exists a constant $C>0$, depending only on $\delta_*$, such that
	\begin{displaymath}
		\frac{|\bar{\mathcal{N}}|}{q}
		\leq \frac{C}{q}\|\bar X^*-X^*R\|_F^2
	\end{displaymath}
	for a suitable orthogonal matrix $R$.
\end{lemma}

\begin{proof}
	This is the standard $k$-means reduction used in spectral clustering. For completeness, we include the short argument; compare Appendix~B of \citet{rohe2011spectral_supp}, Theorem~4.4 of \citet{qin2013regularized_supp}, and Section~F of \citet{rohe2016co_supp}. Let $C^*$ and $\bar C$ denote the population and empirical centroid matrices, respectively. Because $C^*$ is feasible for the empirical $k$-means problem,
	\begin{displaymath}
		\|\bar X^*-\bar C\|_F \leq \|\bar X^*-C^*\|_F.
	\end{displaymath}
	Hence
	\begin{displaymath}
		\|\bar C-C^*\|_F
		\leq \|\bar C-\bar X^*\|_F+\|\bar X^*-C^*\|_F
		\leq 2\|\bar X^*-C^*\|_F
		= 2\|\bar X^*-X^*R\|_F.
	\end{displaymath}
	If node $i$ is misclustered, then the centroid assigned to it must be at least $\delta_*/2$ away from the correct population centroid. Therefore,
	\begin{displaymath}
		|\bar{\mathcal{N}}|\frac{\delta_*^2}{4}
		\leq \sum_{i\in\bar N}\|\bar C_{i\cdot}-C^*_{i\cdot}\|^2
		\leq \|\bar C-C^*\|_F^2
		\leq 4\|\bar X^*-X^*R\|_F^2,
	\end{displaymath}
	which yields the result. 
\end{proof}

\begin{lemma}\label{lem:app_pvar_sv}
Suppose
\begin{displaymath}
    \|\widetilde\Phi_S-\widehat\Phi_S\|
    =
    \mathcal O_{\mathbb P}\!\left(
    \sqrt{\eta_{P,N}}+
    \sqrt{\frac{\log(sq)}{sqB_{sq}}}
    \right),
\end{displaymath}
and let $0\le\alpha_N<\alpha_{\max}$.  Then
\begin{displaymath}
    \sum_{m=1}^{s}\sum_{\bullet\in\{L,R\}}
    \|\bar X_{m,\bullet}^*-X_{m,\bullet}^*R_{m,\bullet}\|_F
    =
    \mathcal O_{\mathbb P}\!\left(
    \sqrt{sq\,\eta_{P,N}}
    +
    \sqrt{\frac{\log(sq)}{B_{sq}}}
    +
    \sqrt q\,\alpha_N\mathcal V_{P,N}
    \right).
\end{displaymath}
\end{lemma}

\begin{proof}
Let $S_m$ denote the symmetric dilation of the stage matrix.  Since the common retained ranks are uniformly bounded, the operator-norm sin-$\Theta$ bound \citep{davis1970rotation_supp,yu2015useful_supp} gives, for each stage,
\begin{displaymath}
    \|\sin\Theta(\widehat E_m,E_m)\|_F
    \le
    C\sqrt{r_{\max}}\,\|\widehat S_m-S_m\|/\delta_0.
\end{displaymath}
The number of stages is fixed, so orthogonal Procrustes alignment \citep{schonemann1966generalized_supp} yields
\begin{displaymath}
    \sum_{m,\bullet}\|\widehat X_{m,\bullet}-X_{m,\bullet}R_{m,\bullet}\|_F
    =
    \mathcal O_{\mathbb P}\!\left(
    \sqrt{\eta_{P,N}}+
    \sqrt{\frac{\log(sq)}{sqB_{sq}}}
    \right).
\end{displaymath}
The associated projector error has the same order.  Lemma~\ref{lem:app_pisces_projector}, applied separately to the two linear chains, therefore gives
\begin{displaymath}
    \sum_{m,\bullet}\|\bar U_{m,\bullet}-U_{m,\bullet}\|_F
    =
    \mathcal O_{\mathbb P}\!\left(
    \sqrt{\eta_{P,N}}+
    \sqrt{\frac{\log(sq)}{sqB_{sq}}}
    +
    \alpha_N\mathcal V_{P,N}
    \right).
\end{displaymath}
For equal-rank projectors $\bar U=\bar X\bar X'$ and $U=XX'$, the principal-angle identities imply
\begin{displaymath}
    \min_R\|\bar X-XR\|_F\le\|\bar U-U\|_F.
\end{displaymath}
Applying Lemma~\ref{lem:app_rownorm} and absorbing the fixed factor $s$ gives the stated bound.  The case $\alpha_N=0$ follows directly.
\end{proof}

\begin{lemma}\label{lem:app_vhar_sv}
Under the analogous operator-norm bound for generalized ScBM--VHAR and $0\le\alpha_N<\alpha_{\max}$,
\begin{displaymath}
    \sum_{h\in\{\ell,m,s\}}\sum_{\bullet\in\{L,R\}}
    \|\bar X_{(h),\bullet}^*-X_{(h),\bullet}^*R_{(h),\bullet}\|_F
    =
    \mathcal O_{\mathbb P}\!\left(
    \sqrt{sq\,\eta_{V,N}}
    +
    \sqrt{\frac{\log(sq)}{B_{sq}}}
    +
    \sqrt q\,\alpha_N\mathcal V_{V,N}
    \right).
\end{displaymath}
\end{lemma}

\begin{proof}
Apply the preceding proof to the fixed long--medium--short chain. The restricted VAR representation changes only the first-stage operator-norm rate; the fixed-rank subspace, projector, Procrustes, and row-normalization arguments are identical.
\end{proof}

\noindent{\textbf{Proof of Theorem \ref{thm:missclassification_pvar}.}}
The first-stage estimation and graph-randomness bounds imply the operator-norm premise of Lemma~\ref{lem:app_pvar_sv}, which gives the stated singular-vector rate.  Lemma~\ref{lem:app_misclass} applied to the $2s$ directional partitions yields
\begin{displaymath}
    \frac1{sq}\sum_{m=1}^{s}\left(|\bar{\mathcal N}_m^y|+|\bar{\mathcal N}_m^z|\right)
    \le
    \frac{C}{sq}\sum_{m,\bullet}
    \|\bar X_{m,\bullet}^*-X_{m,\bullet}^*R_{m,\bullet}\|_F^2.
\end{displaymath}
Substitution of Lemma~\ref{lem:app_pvar_sv} and $(a+b+c)^2\lesssim a^2+b^2+c^2$ gives
\begin{displaymath}
    \mathcal O_{\mathbb P}\!\left(
    \eta_{P,N}+
    \frac{\log(sq)}{sqB_{sq}}+
    \alpha_N^2\mathcal V_{P,N}^2
    \right),
\end{displaymath}
where fixed factors involving $s$ are absorbed into the constant. \qed

\medskip
\noindent{\textbf{Proof of Theorem \ref{thm:missclassification_vhar}.}}
The same argument, using Lemma~\ref{lem:app_vhar_sv} and averaging over the three horizons, gives
\begin{displaymath}
    \frac1{3q}\sum_{h\in\{\ell,m,s\}}
    \left(|\bar{\mathcal N}_{(h)}^y|+|\bar{\mathcal N}_{(h)}^z|\right)
    =
    \mathcal O_{\mathbb P}\!\left(
    \eta_{V,N}+
    \frac{\log(sq)}{sqB_{sq}}+
    \alpha_N^2\mathcal V_{V,N}^2
    \right).
\end{displaymath}
\qed

\clearpage
\section{Additional results of finite sample performances}
\label{ap:additional_sim}

\subsection{ScBM--PVAR: OLS results}
The following tables present the PisCES-smoothed OLS estimates and evaluate their similarity to the lasso estimation results using the ARI.

\begin{table}[ht]
\centering
\scriptsize
\setlength{\tabcolsep}{5pt}
\begin{adjustbox}{max width=\textwidth}
\begin{tabular}{lllcccccccc}
\hline
\multirow{2}{*}{Dim} & \multirow{2}{*}{Path} & \multirow{2}{*}{Type} & \multicolumn{4}{c}{Spectral Norm} & \multicolumn{4}{c}{Accuracy} \\
\cline{4-11}
& & & $T=200$ & $T=500$ & $T=1000$ & $T=2000$ & $T=200$ & $T=500$ & $T=1000$ & $T=2000$ \\
\hline
\multirow{6}{*}{$q=18$} & path1 & type1 & 1.571 & 0.858 & 0.650 & 0.544 & 0.562 & 0.970 & 1.000 & 1.000 \\
 & path1 & type2 & 1.593 & 0.883 & 0.689 & 0.595 & 0.523 & 0.804 & 0.971 & 0.997 \\
 & path2 & type1 & 1.618 & 0.989 & 0.842 & 0.773 & 0.570 & 0.757 & 0.860 & 0.922 \\
 & path2 & type2 & 1.592 & 0.941 & 0.779 & 0.699 & 0.563 & 0.651 & 0.762 & 0.829 \\
 & path3 & type1 & 1.639 & 1.020 & 0.872 & 0.800 & 0.572 & 0.704 & 0.765 & 0.788 \\
 & path3 & type2 & 1.611 & 0.942 & 0.778 & 0.699 & 0.556 & 0.662 & 0.727 & 0.752 \\
\hline
\multirow{6}{*}{$q=36$} & path1 & type1 & 4.947 & 1.419 & 0.992 & 0.820 & 0.382 & 0.480 & 0.780 & 0.963 \\
 & path1 & type2 & 4.944 & 1.418 & 0.988 & 0.812 & 0.387 & 0.433 & 0.558 & 0.773 \\
 & path2 & type1 & 4.912 & 1.455 & 1.083 & 0.946 & 0.503 & 0.532 & 0.660 & 0.861 \\
 & path2 & type2 & 4.808 & 1.419 & 1.043 & 0.898 & 0.499 & 0.517 & 0.567 & 0.656 \\
 & path3 & type1 & 4.928 & 1.439 & 1.058 & 0.913 & 0.489 & 0.535 & 0.645 & 0.730 \\
 & path3 & type2 & 4.839 & 1.420 & 1.041 & 0.896 & 0.486 & 0.513 & 0.586 & 0.666 \\
\hline
\multirow{6}{*}{$q=60$} & path1 & type1 & 8.173 & 2.345 & 1.302 & 0.984 & 0.353 & 0.365 & 0.464 & 0.761 \\
 & path1 & type2 & 8.295 & 2.357 & 1.291 & 0.964 & 0.354 & 0.362 & 0.398 & 0.492 \\
 & path2 & type1 & 7.926 & 2.333 & 1.340 & 1.062 & 0.483 & 0.486 & 0.528 & 0.696 \\
 & path2 & type2 & 7.969 & 2.318 & 1.309 & 1.020 & 0.484 & 0.486 & 0.508 & 0.568 \\
 & path3 & type1 & 8.025 & 2.333 & 1.335 & 1.059 & 0.467 & 0.474 & 0.521 & 0.649 \\
 & path3 & type2 & 8.048 & 2.317 & 1.310 & 1.017 & 0.467 & 0.469 & 0.487 & 0.545 \\
\hline
\end{tabular}
\end{adjustbox}
\caption{Spectral norm and accuracy for the PisCES-smoothed OLS results in the ScBM--PVAR simulation under common-rank spectral co-clustering and uniform projector cross-validation.}
\label{tab:pvar_pisces_ols}
\end{table}

\begin{figure}[h!]
	\centering
	\includegraphics[width=.7\linewidth]{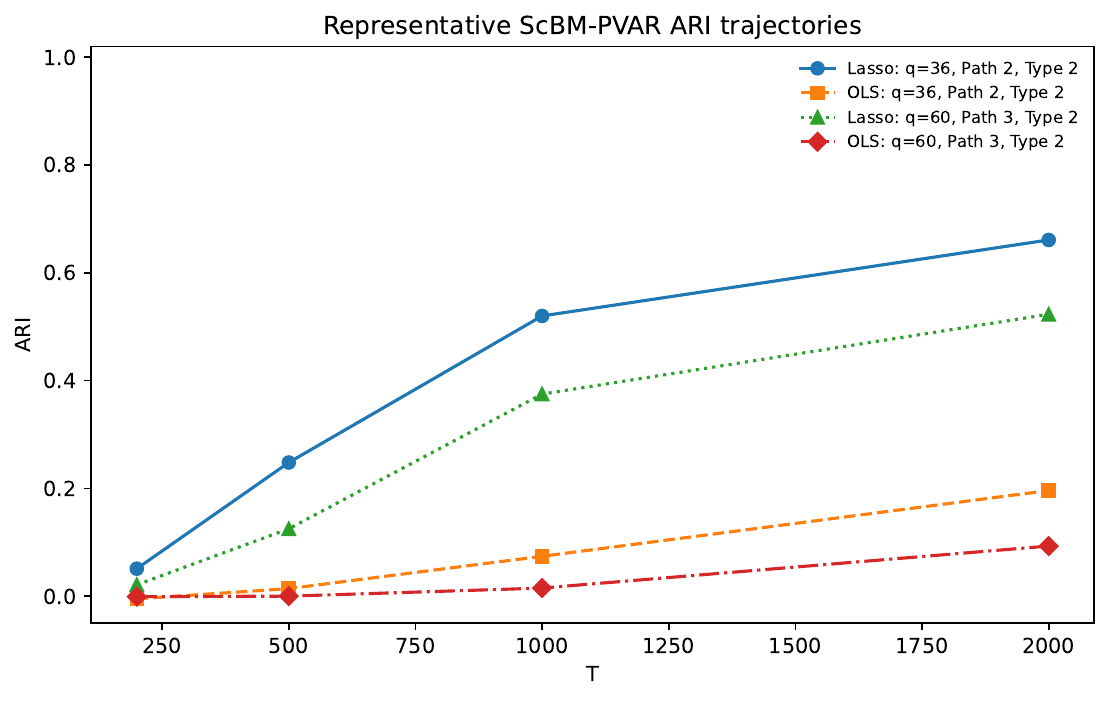}
	\vspace{-5 mm}
	\caption{Representative ARI trajectories comparing lasso and OLS in difficult ScBM--PVAR designs.} 
	\label{fig:pvar_ari_comparison}
\end{figure}

Table~\ref{tab:pvar_pisces_ols} gives the OLS counterpart to the lasso results in the main paper, while Figure~\ref{fig:pvar_ari_comparison} summarizes the corresponding label-invariant ARI trajectories. The gap between lasso and OLS is most pronounced in the higher-dimensional dynamic designs.

\subsection{Generalized ScBM--VHAR: OLS and ARI results}

The following tables report the PisCES-smoothed OLS estimates and the corresponding ARI comparison with lasso results. 

\begin{table}[h]
\centering
\scriptsize
\setlength{\tabcolsep}{5pt}
\begin{adjustbox}{max width=\textwidth}
\begin{tabular}{lllcccccccc}
\hline
\multirow{2}{*}{Dim} & \multirow{2}{*}{Path} & \multirow{2}{*}{Type} & \multicolumn{4}{c}{Spectral Norm} & \multicolumn{4}{c}{Accuracy} \\
\cline{4-11}
& & & $T=500$ & $T=1000$ & $T=2000$ & $T=3000$ & $T=500$ & $T=1000$ & $T=2000$ & $T=3000$ \\
\hline
\multirow{6}{*}{$q=18$} & path1 & type1 & 1.101 & 0.686 & 0.451 & 0.363 & 0.499 & 0.561 & 0.707 & 0.829 \\
 & path1 & type2 & 1.103 & 0.687 & 0.452 & 0.363 & 0.492 & 0.528 & 0.620 & 0.723 \\
 & path2 & type1 & 1.111 & 0.695 & 0.457 & 0.367 & 0.592 & 0.661 & 0.763 & 0.825 \\
 & path2 & type2 & 1.111 & 0.696 & 0.457 & 0.367 & 0.584 & 0.630 & 0.710 & 0.782 \\
 & path3 & type1 & 1.113 & 0.694 & 0.457 & 0.367 & 0.577 & 0.627 & 0.742 & 0.825 \\
 & path3 & type2 & 1.112 & 0.695 & 0.457 & 0.367 & 0.566 & 0.596 & 0.687 & 0.759 \\
\hline
\multirow{6}{*}{$q=24$} & path1 & type1 & 1.429 & 0.838 & 0.543 & 0.431 & 0.468 & 0.496 & 0.586 & 0.705 \\
 & path1 & type2 & 1.427 & 0.838 & 0.544 & 0.432 & 0.464 & 0.488 & 0.557 & 0.648 \\
 & path2 & type1 & 1.433 & 0.845 & 0.548 & 0.435 & 0.571 & 0.606 & 0.702 & 0.765 \\
 & path2 & type2 & 1.434 & 0.845 & 0.548 & 0.435 & 0.563 & 0.589 & 0.673 & 0.731 \\
 & path3 & type1 & 1.429 & 0.843 & 0.547 & 0.434 & 0.543 & 0.580 & 0.669 & 0.764 \\
 & path3 & type2 & 1.431 & 0.845 & 0.548 & 0.435 & 0.535 & 0.561 & 0.625 & 0.705 \\
\hline
\multirow{6}{*}{$q=36$} & path1 & type1 & 2.136 & 1.146 & 0.713 & 0.559 & 0.437 & 0.455 & 0.505 & 0.578 \\
 & path1 & type2 & 2.132 & 1.146 & 0.712 & 0.559 & 0.435 & 0.447 & 0.479 & 0.527 \\
 & path2 & type1 & 2.136 & 1.151 & 0.715 & 0.561 & 0.540 & 0.561 & 0.618 & 0.676 \\
 & path2 & type2 & 2.136 & 1.151 & 0.714 & 0.561 & 0.538 & 0.551 & 0.596 & 0.644 \\
 & path3 & type1 & 2.137 & 1.151 & 0.714 & 0.561 & 0.515 & 0.533 & 0.586 & 0.664 \\
 & path3 & type2 & 2.138 & 1.151 & 0.715 & 0.561 & 0.511 & 0.524 & 0.565 & 0.622 \\
\hline
\end{tabular}
\end{adjustbox}
\caption{Spectral norm and accuracy for the PisCES-smoothed OLS results in the generalized ScBM--VHAR simulation under common-rank spectral co-clustering and uniform projector cross-validation.}
\label{tab:vhar_pisces_ols}
\end{table}

\begin{figure}[h!]
	\centering
	\includegraphics[width=.9\linewidth]{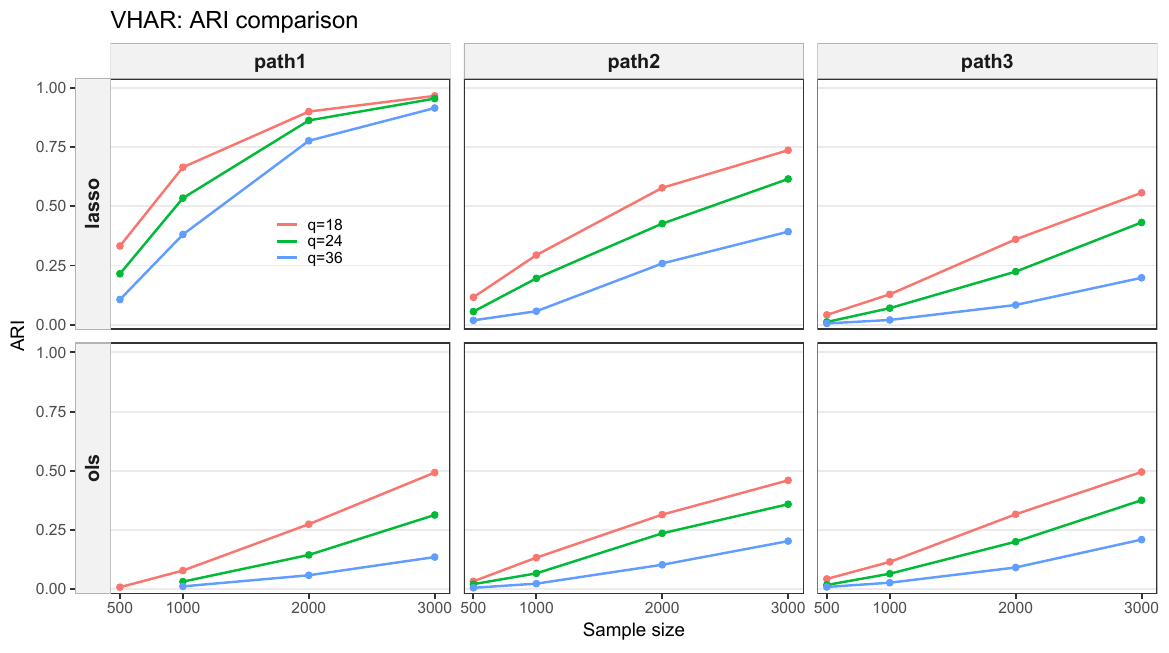}
	\vspace{-5 mm}
	\caption{Representative ARI trajectories comparing lasso and OLS in generalized ScBM--VHAR designs.} 
	\label{fig:vhar_ari_comparison}
\end{figure}

Figure~\ref{fig:vhar_ari_comparison} shows that OLS remains substantially less accurate than lasso in the larger-$q$ VHAR designs. In particular, the $q=36$, Path~1, Type~1 design at $T=3000$ has lasso ARI $0.966$ versus OLS ARI $0.171$. Table~\ref{tab:vhar_pisces_ols} reports the corresponding spectral norm and accuracy results.

\FloatBarrier
\subsection{Selected PisCES smoothing versus no smoothing}
\label{ap:alpha_zero_comparison}

To isolate the effect of the PisCES step from first-stage estimation, Table~\ref{tab:alpha_zero_comparison} compares final lasso recovery at the CV-selected smoothing parameter with the same estimated transition matrices clustered at $\alpha_N=0$. The latter is the unsmoothed spectral co-clustering benchmark. These examples are representative rather than an additional tuning experiment; all numbers are taken from the same $B=200$ production runs reported in the main simulation section.

\begin{table}[ht!]
\centering
\small
\setlength{\tabcolsep}{4.5pt}
\begin{tabular}{llcccccc}
\hline
Model & Design & $T$ & $\widehat\alpha_N$ & Accuracy $(0)$ & Accuracy $(\widehat\alpha_N)$ & ARI $(0)$ & ARI $(\widehat\alpha_N)$ \\
\hline
PVAR & $q=36$, Path 2, Type 1 & 200  & 0.129 & 0.590 & 0.601 & 0.112 & 0.128 \\
PVAR & $q=36$, Path 2, Type 2 & 2000 & 0.129 & 0.838 & 0.883 & 0.567 & 0.661 \\
PVAR & $q=18$, Path 3, Type 1 & 2000 & 0.129 & 0.889 & 0.900 & 0.785 & 0.807 \\
VHAR & $q=24$, Path 2, Type 1 & 500  & 0.000 & 0.606 & 0.606 & 0.081 & 0.081 \\
VHAR & $q=24$, Path 2, Type 1 & 3000 & 0.129 & 0.868 & 0.875 & 0.665 & 0.685 \\
VHAR & $q=36$, Path 1, Type 1 & 3000 & 0.129 & 0.977 & 0.985 & 0.944 & 0.966 \\
\hline
\end{tabular}
\caption{Representative comparison of CV-selected PisCES smoothing with the unsmoothed benchmark $\alpha_N=0$ for lasso-based recovery. Accuracy and ARI are averaged over 200 replications.}
\label{tab:alpha_zero_comparison}
\end{table}

The comparison shows that the CV rule can choose no smoothing in a difficult small-sample design: for the $q=24$ VHAR Path~2 design at $T=500$, $\widehat\alpha_N=0$ and the two recovery columns coincide. In the moderate and larger samples shown here, positive PisCES smoothing improves or matches both accuracy and ARI, with particularly visible gains in the dynamic PVAR designs. These selected rows should not be read as establishing uniform dominance of positive smoothing: $\alpha_N=0$ remains in the candidate set, the CV criterion targets held-out coefficient reconstruction rather than ARI, and the theoretical bound contains a smoothing-bias term proportional to $\alpha_N\mathcal V_{\bullet,N}$. The appropriate interpretation is therefore adaptive regularization rather than a claim that smoothing must improve every design.

\FloatBarrier
\subsection{Robustness to model misspecification}
\label{ap:misspecification}

In each robustness design, a sender-by-receiver block matrix is generated first, and the actual time-series transition matrix is taken as its transpose. Accordingly, the spectral matrix used for singular value decomposition is $\Phi^\prime$, so the left and right singular spaces correspond to $\Xi^y$ and $\Xi^z$, respectively. All designs use 100 replications, the common retained rank $r_m=\min(K_{y_m},K_{z_m})$, uniform projector cross-validation with $\varpi=0.20$ and $M=10$, and an upper grid point of $0.99\alpha_{\max}$. 

For the community path numbers, the PVAR experiment uses $(q=36,\mathrm{Path~2},\mathrm{Type~2})$ design with true community path $2\to3\to3\to2$. We compare the correctly specified fit with a coarse path, $2\to2\to2\to2$, and a fine path, $3\to4\to4\to3$. The VHAR experiment uses the static $(q=24,\mathrm{Path~1},\mathrm{Type~1})$ benchmark with true path $3\to3\to3\to3$, together with coarse and fine fits $2\to2\to2\to2$ and $4\to4\to4\to4$, respectively.

\begin{table}[t]
\centering
\small
\setlength{\tabcolsep}{4.5pt}
\begin{tabular}{lllrcc}
\hline
Model & $T$ & Fitted community path & Scenario & Accuracy & ARI \\
\hline
PVAR & 500  & $2\!\to\!2\!\to\!2\!\to\!2$ & Coarse  & 0.623 & 0.198 \\
PVAR & 500  & $2\!\to\!3\!\to\!3\!\to\!2$ & Correct & 0.661 & 0.223 \\
PVAR & 500  & $3\!\to\!4\!\to\!4\!\to\!3$ & Fine    & 0.572 & 0.239 \\
PVAR & 2000 & $2\!\to\!2\!\to\!2\!\to\!2$ & Coarse  & 0.756 & 0.501 \\
PVAR & 2000 & $2\!\to\!3\!\to\!3\!\to\!2$ & Correct & \textbf{0.898} & \textbf{0.690} \\
PVAR & 2000 & $3\!\to\!4\!\to\!4\!\to\!3$ & Fine    & 0.699 & 0.499 \\
\hline
VHAR & 1000 & $2\!\to\!2\!\to\!2\!\to\!2$ & Coarse  & 0.631 & 0.392 \\
VHAR & 1000 & $3\!\to\!3\!\to\!3\!\to\!3$ & Correct & \textbf{0.871} & \textbf{0.695} \\
VHAR & 1000 & $4\!\to\!4\!\to\!4\!\to\!4$ & Fine    & 0.778 & 0.643 \\
VHAR & 3000 & $2\!\to\!2\!\to\!2\!\to\!2$ & Coarse  & 0.665 & 0.518 \\
VHAR & 3000 & $3\!\to\!3\!\to\!3\!\to\!3$ & Correct & \textbf{0.997} & \textbf{0.990} \\
VHAR & 3000 & $4\!\to\!4\!\to\!4\!\to\!4$ & Fine    & 0.855 & 0.832 \\
\hline
\end{tabular}
\caption{Sensitivity to deliberate community-number misspecification. The PVAR true path is $2\to3\to3\to2$ and the VHAR true path is $3\to3\to3\to3$. Accuracy and ARI are averaged over $B=100$ replications.}
\label{tab:misspec_K}
\end{table}

Table~\ref{tab:misspec_K} reports the corresponding accuracy and ARI values. For PVAR at $T=500$, the differences among the three fits are relatively small. The correct fit yields accuracy $0.661$ and ARI $0.223$, compared with $(0.623,0.198)$ for the coarse fit and $(0.572,0.239)$ for the fine fit. At $T=2000$, the advantage of the correct specification becomes much clearer, with $(0.898,0.690)$ compared with $(0.756,0.501)$ and $(0.699,0.499)$ for the coarse and fine fits, respectively. The VHAR results show a similar pattern. At $T=1000$, the correct fit attains $(0.871,0.695)$, compared with $(0.631,0.392)$ for the coarse fit and $(0.778,0.643)$ for the fine fit. At $T=3000$, the correct fit achieves $(0.997,0.990)$, while the coarse and fine fits yield $(0.665,0.518)$ and $(0.855,0.832)$. Overall, the results show that the correct community specification becomes increasingly important as the sample size grows.

For the sparse-relocation experiment, we replace a fraction $\eta\in \{0.1,0.2,0.3\}$ of the off-diagonal nonzero coefficients while preserving the diagonal, the total number of nonzero coefficients, the multiset of coefficient magnitudes, and the target stability scaling. Increasing $\eta$ therefore disrupts the block organization without changing the overall sparsity level or the distribution of signal magnitudes. The dynamic VHAR benchmark uses the separate long-to-short Path~2, Type~2 design with boundary counts $2\to2\to3\to2$. Note that the $\eta=0$ rows in Table~\ref{tab:misspec_eta} and the correctly specified rows in Table~\ref{tab:misspec_K} should be interpreted as experiment-specific baselines rather than compared directly.

Table~\ref{tab:misspec_eta} shows that increasing contamination mainly weakens community recovery, while the spectral-norm error remains relatively stable. For PVAR, the spectral norm changes little across values of $\eta$, whereas both accuracy and ARI decline overall. At $T=2000$, for example, accuracy decreases from $0.898$ at $\eta=0$ to $0.849$, $0.695$, and $0.709$, while ARI decreases from $0.690$ to $0.579$, $0.277$, and $0.300$. The VHAR results show the same pattern more clearly: at $T=3000$, accuracy declines from $0.808$ to $0.734$, $0.700$, and $0.620$, and ARI from $0.510$ to $0.338$, $0.248$, and $0.097$, while the spectral norm remains near $0.21$. Overall, the results indicate that departures from the ScBM block structure have a much stronger effect on community recovery than on estimation of the transition matrix.

\begin{table}[t!]
\centering
\small
\setlength{\tabcolsep}{4.5pt}
\begin{tabular}{llcccc}
\hline
Model & $T$ & $\eta$ & Spectral norm & Accuracy & ARI \\
\hline
PVAR & 500  & 0.0 & 0.505 & 0.661 & 0.223 \\
PVAR & 500  & 0.1 & 0.502 & 0.652 & 0.195 \\
PVAR & 500  & 0.2 & 0.495 & 0.582 & 0.094 \\
PVAR & 500  & 0.3 & 0.508 & 0.587 & 0.098 \\
PVAR & 2000 & 0.0 & 0.276 & 0.898 & 0.690 \\
PVAR & 2000 & 0.1 & 0.273 & 0.849 & 0.579 \\
PVAR & 2000 & 0.2 & 0.273 & 0.695 & 0.277 \\
PVAR & 2000 & 0.3 & 0.274 & 0.709 & 0.300 \\
\hline
VHAR & 1000 & 0.0 & 0.287 & 0.643 & 0.142 \\
VHAR & 1000 & 0.1 & 0.286 & 0.608 & 0.086 \\
VHAR & 1000 & 0.2 & 0.282 & 0.584 & 0.046 \\
VHAR & 1000 & 0.3 & 0.290 & 0.568 & 0.022 \\
VHAR & 3000 & 0.0 & 0.213 & 0.808 & 0.510 \\
VHAR & 3000 & 0.1 & 0.212 & 0.734 & 0.338 \\
VHAR & 3000 & 0.2 & 0.209 & 0.700 & 0.248 \\
VHAR & 3000 & 0.3 & 0.212 & 0.620 & 0.097 \\
\hline
\end{tabular}
\caption{Sensitivity to sparse departures from the ScBM coefficient pattern. The rewiring fraction $\eta$ preserves the diagonal, sparsity level, coefficient-magnitude multiset, and stability target. Results are averages over 100 replications.}
\label{tab:misspec_eta}
\end{table}

The stability of the spectral norm errors across $\eta$ is expected because the contamination changes the location of nonzero coefficients rather than the first-stage estimation problem. The pronounced decline in accuracy and ARI therefore reflects the progressive loss of block organization. Together, Tables~\ref{tab:misspec_K} and~\ref{tab:misspec_eta} show that the method is reasonably tolerant of mild departures but should not be interpreted as invariant to community-number or block-structure misspecification.

\clearpage
\section{Additional tables and figures for data applications}
\label{ap:additional_data}

\subsection{Data application of ScBM--PVAR model}
\label{ap:additional_data_pvar}

\begin{table}[h]
	\centering
	\resizebox{1\columnwidth}{!}{%
		\begin{tabular}{|l|l|l|l|}
			\hline
			\textbf{Category} & \textbf{Level 2 Subcategory} & \textbf{Level 3 Subcategory} & \textbf{Code} \\ 
			\hline
			\multirow{28}{*}{\textbf{Total Private}} & \textbf{Goods-Producing} &  &  \\ 
			& \textbf{Mining and Logging} &  & Mining \\ 
			& \textbf{Construction} &  & Construction \\ 
			& \textbf{Manufacturing} &  &  \\ 
			&  & - Durable Goods & Durable \\ 
			&  & - Nondurable Goods & Nondurable \\ 
			\cline{2-4}
			& \textbf{Private Service-Provision} &  &  \\ 
			& \textbf{Trade, Transportation, and Utilities} &  &  \\ 
			&  & - Wholesale Trade & Wholesale \\ 
			&  & - Retail Trade & Retail \\ 
			&  & - Transportation and Warehousing & Transportation \\ 
			&  & - Utilities & Utility \\ 
			& \textbf{Information} &  & Information \\ 
			& \textbf{Financial Activities} &  &  \\ 
			&  & - Finance and Insurance & Finance \\ 
			&  & - Real Estate and Rental and Leasing & RealEstate \\ 
			& \textbf{Professional and Business Services} &  &  \\ 
			&  & - Professional, Scientific, and Technical Services & Professional \\ 
			&  & - Management of Companies and Enterprises & Management \\ 
			&  & - Administrative and Support and Waste Management and Remediation Services & Administration \\ 
			& \textbf{Private Education and Health Services} &  &  \\ 
			&  & - Private Educational Services & Education \\ 
			&  & - Health Care and Social Assistance & HealthCare \\ 
			& \textbf{Leisure and Hospitality} &  &  \\ 
			&  & - Arts, Entertainment, and Recreation & Arts \\ 
			&  & - Accommodation and Food Services & Accommodation \\ 
			& \textbf{Other Services} &  & Other \\ 
			\hline
			\multirow{3}{*}{\textbf{Government}} & \textbf{Federal} &  & Federal \\ 
			& \textbf{State Government} &  & State \\ 
			& \textbf{Local Government} &  & Local \\ 
			\hline
	\end{tabular}}
	\caption{Hierarchy of industry sectors in Nonfarm Payroll Employment.}
	\label{tab:employment_sectors}
\end{table}

\begin{figure}[h]
	\centering
	\includegraphics[width=0.9\textwidth,height=0.3\textheight]{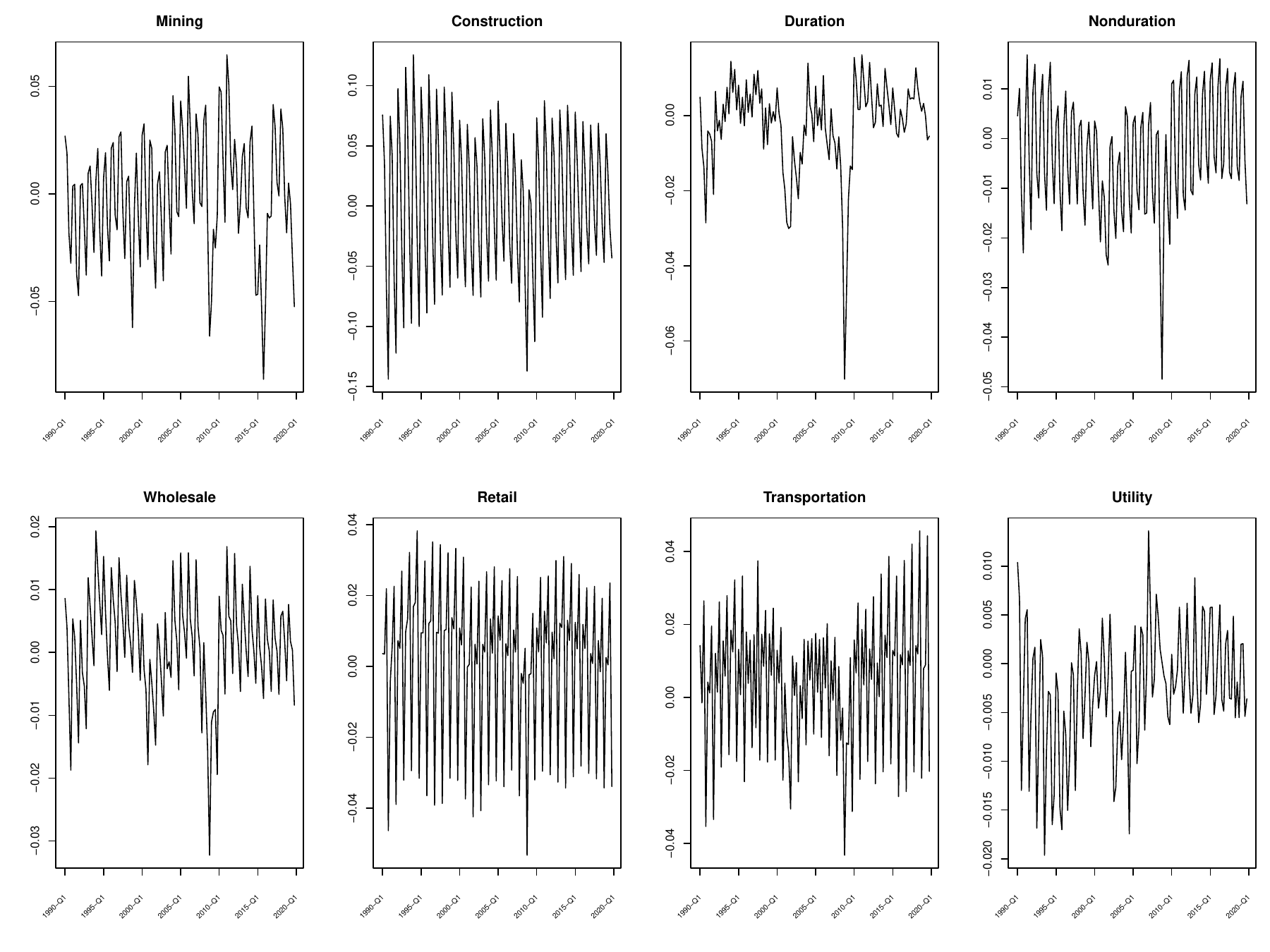}
	\includegraphics[width=0.45\textwidth,height=0.35\textheight]{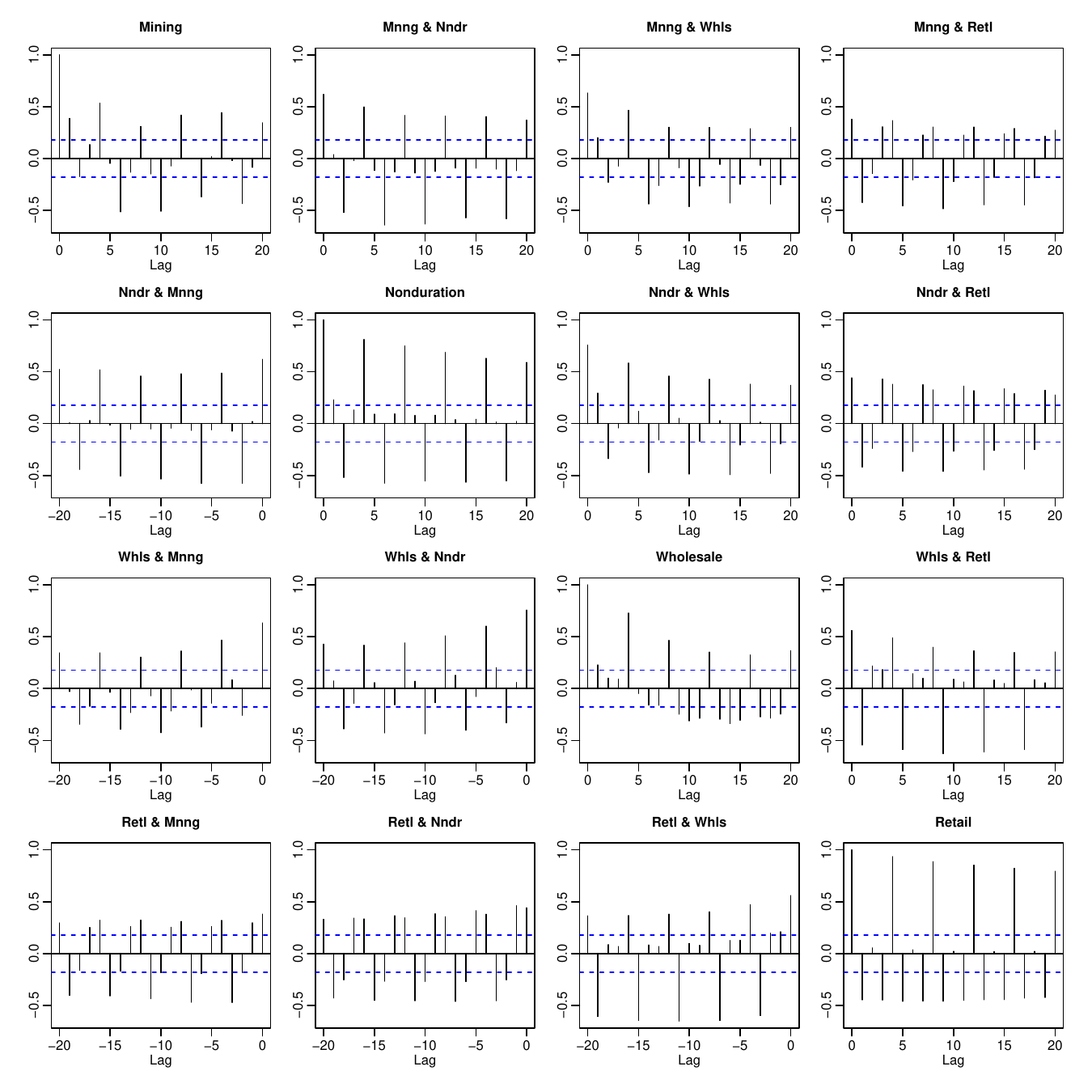}
	\includegraphics[width=0.45\textwidth,height=0.35\textheight]{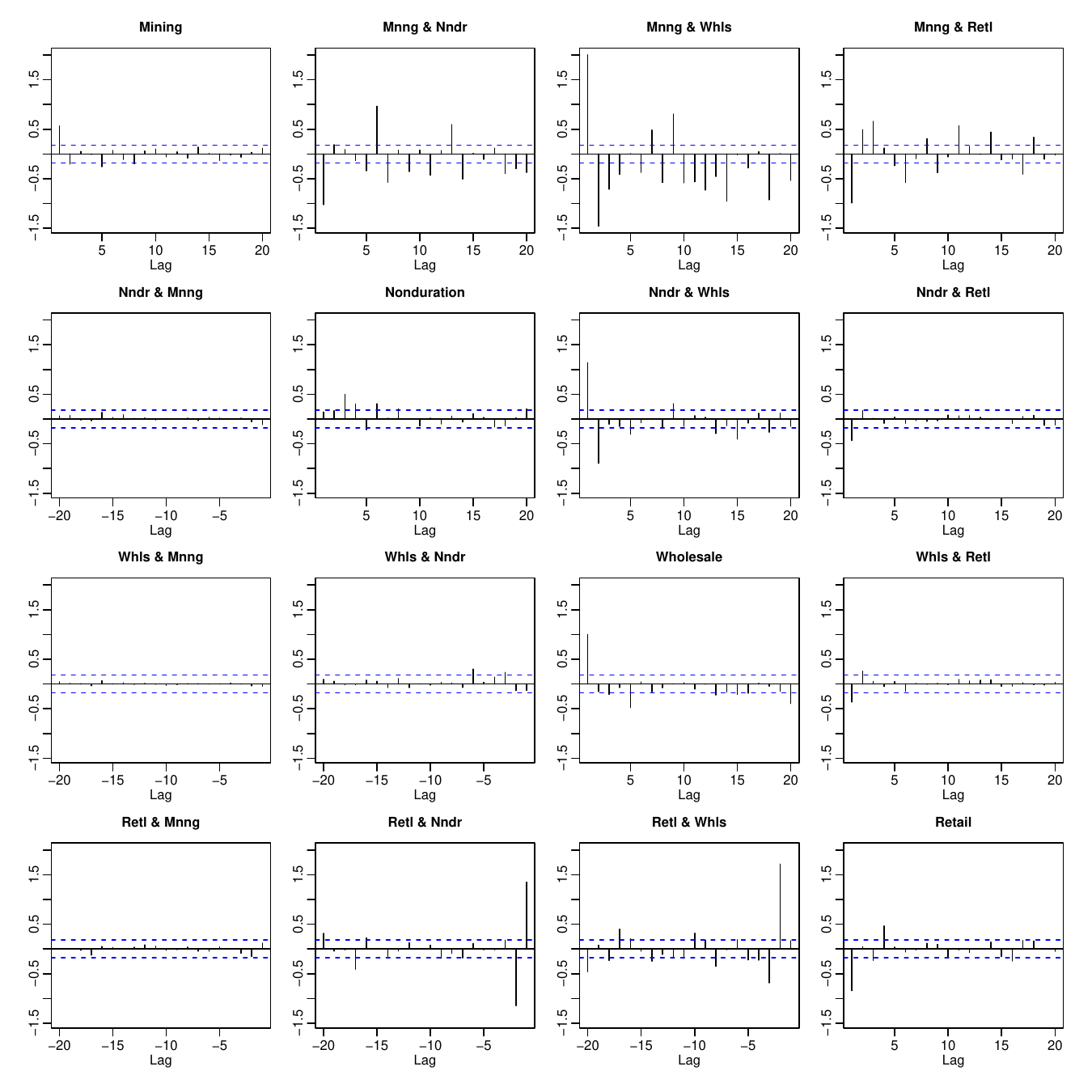}
	\caption{(Top) Time plots of the first eight time series (see Table \ref{tab:employment_sectors} for code). (Left Bottom) Sample ACF plots for selected sectors (Mining, Nondurable, Wholesale, and Retail). (Right Bottom) Sample PACF plots for the same sectors.}
	\label{fig:timeplot_PVAR}
\end{figure}

Figure~\ref{fig:rank_PVAR} reports the seasonal scree diagnostics used to motivate the retained common ranks in the PVAR application. We treat the scree as a descriptive spectral diagnostic, analogous in spirit to eigenvalue-based dimension diagnostics in factor models \citep{bai2007determining_supp}, rather than as a formal consistency result for the community counts.
\begin{figure}[t!]
	\centering
	\includegraphics[width=0.9\textwidth]{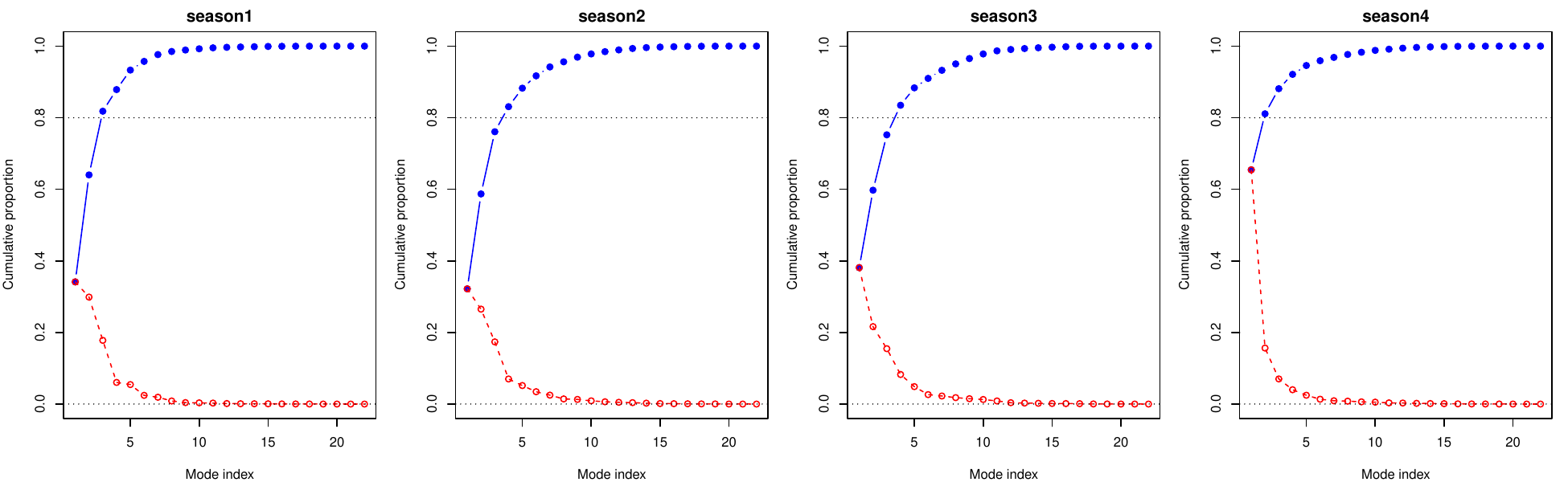}
    \caption{Scree plots of singular values from the estimated transition matrices of the ScBM--PVAR model with four seasons. The blue solid curves with filled circles show the cumulative proportion of variation explained, while the red dotted curves with open circles show the proportion explained by each singular value. The horizontal black dotted lines mark the 80\% threshold.}
	\label{fig:rank_PVAR}
\end{figure}

As a configuration-sensitivity check, Figure~\ref{fig:pvar_payroll_sankey_alt} uses the alternative admissible configuration $(2,3,\,3,3,\,3,2,\,2,2)$, with the stage ranks $(2,3,2,2)$. The existence of substantial seasonal reorganization is robust, whereas the detailed Q1/Q2 (Summer)--Q2/Q3 (Fall) partition is more sensitive to the allocation of directional community counts.

\begin{figure}[t!]
\centering
\includegraphics[width=.96\textwidth]{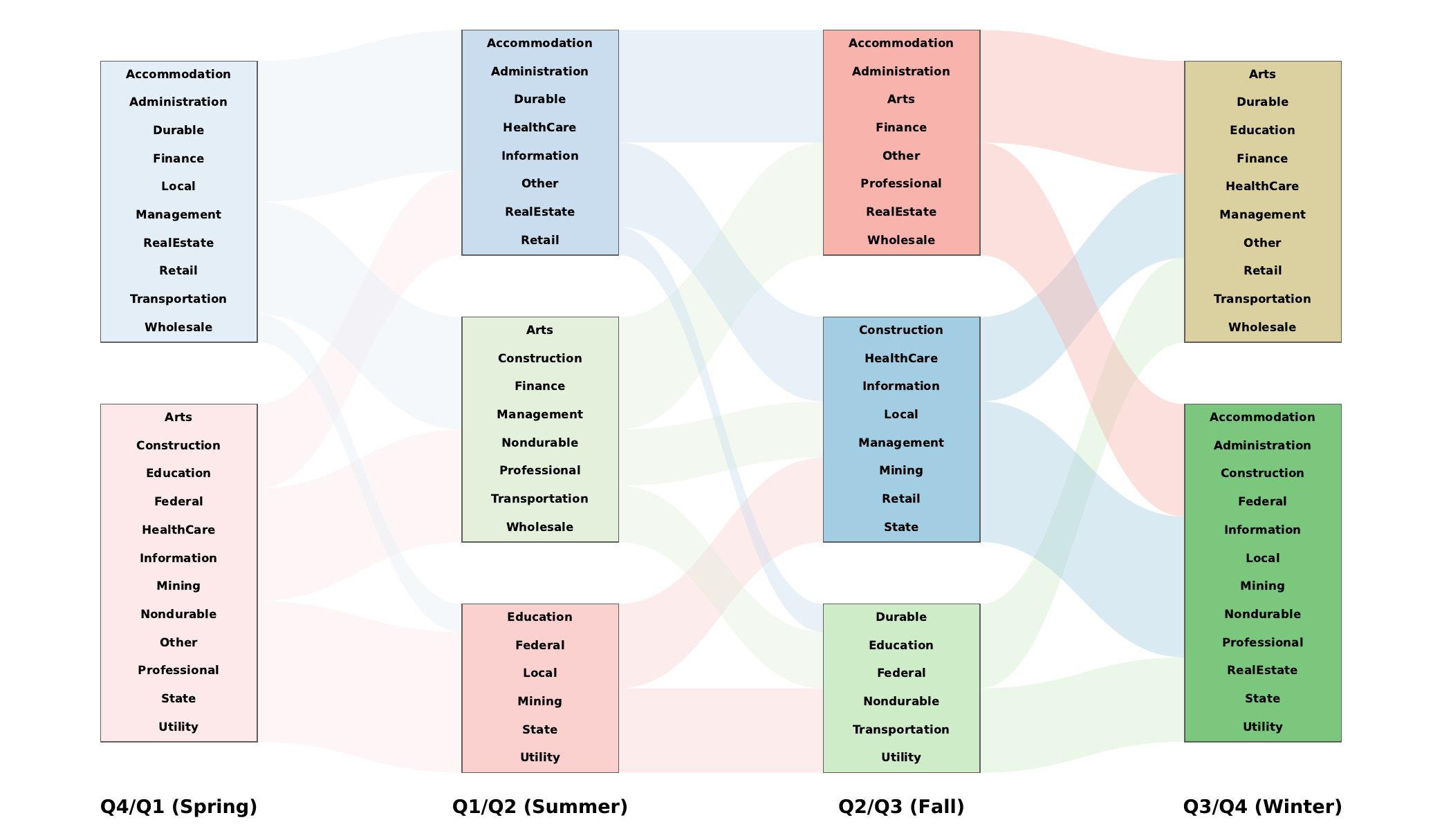}
\caption{PVAR-based payroll community path under the alternative admissible configuration $(2,3,\,3,3,\,3,2,\,2,2)$, with the stage ranks $(2,3,2,2)$.}
\label{fig:pvar_payroll_sankey_alt}
\end{figure}

\clearpage
\subsection{Data application of generalized ScBM--VHAR model}
\label{ap:additional_data_Vhar}

\begin{table}[ht!]
	\centering
	\small
	\resizebox{0.9\columnwidth}{!}{%
		\begin{tabular}{|l|l|l|l|l|}
			\hline
			\textbf{Stock Index} & \textbf{MSCI Classification} & \textbf{Continent} & \textbf{Market} & \textbf{Ticker} \\
			\hline
			AEX Index & Developed & Europe & Netherlands & AEX \\
			All Ordinaries & Developed & Oceania & Australia & AORD \\
			BFX Index & Developed & Europe & Belgium & BFX \\
			BSESN Index & Emerging & Asia & India & BSESN \\
			BVSP Index & Emerging & South America & Brazil & BVSP \\
			DJIA & Developed & North America & United States & DJI \\
			CAC 40 & Developed & Europe & France & FCHI \\
			FTSE MIB & Developed & Europe & Italy & FTMIB \\
			FTSE 100 & Developed & Europe & United Kingdom & FTSE \\
			DAX & Developed & Europe & Germany & GDAXI \\
			S\&P/TSX Composite Index & Developed & North America & Canada & GSPTSE \\
			Hang Seng & Developed & Asia & Hong Kong & HSI \\
			IBEX 35 & Developed & Europe & Spain & IBEX \\
			Nasdaq Composite & Developed & North America & United States & IXIC \\
			KOSPI Composite Index & Developed & Asia & South Korea & KS11 \\
			KSE 100 & Emerging & Asia & Pakistan & KSE \\
			IPC Mexico & Emerging & North America & Mexico & MXX \\
			Nikkei 225 & Developed & Asia & Japan & N225 \\
			NSE Nifty 50 & Emerging & Asia & India & NSEI \\
			OMX C20 & Developed & Europe & Denmark & OMXC20 \\
			OMX HPI & Developed & Europe & Finland & OMXHPI \\
			OMX SPI & Developed & Europe & Sweden & OMXSPI \\
			Oslo All Share & Developed & Europe & Norway & OSEAX \\
			Russell 2000 & Developed & North America & United States & RUT \\
			S\&P 500 & Developed & North America & United States & SPX \\
			Shanghai Composite & Emerging & Asia & China & SSEC \\
			Swiss Market Index & Developed & Europe & Switzerland & SSMI \\
			Madrid General Index  & Developed & Europe & Spain &  SMSI \\
			Euro STOXX 50 & Developed & Europe & Eurozone & STOXX50E \\
			\hline
	\end{tabular}}
	\caption{The list of 29 stock indices ordered by MSCI market classification, continent, and country.}
	\label{tab:stock_indices}
\end{table}

\begin{figure}[h]
	\centering
	\includegraphics[width=0.9\textwidth,height=0.3\textheight]{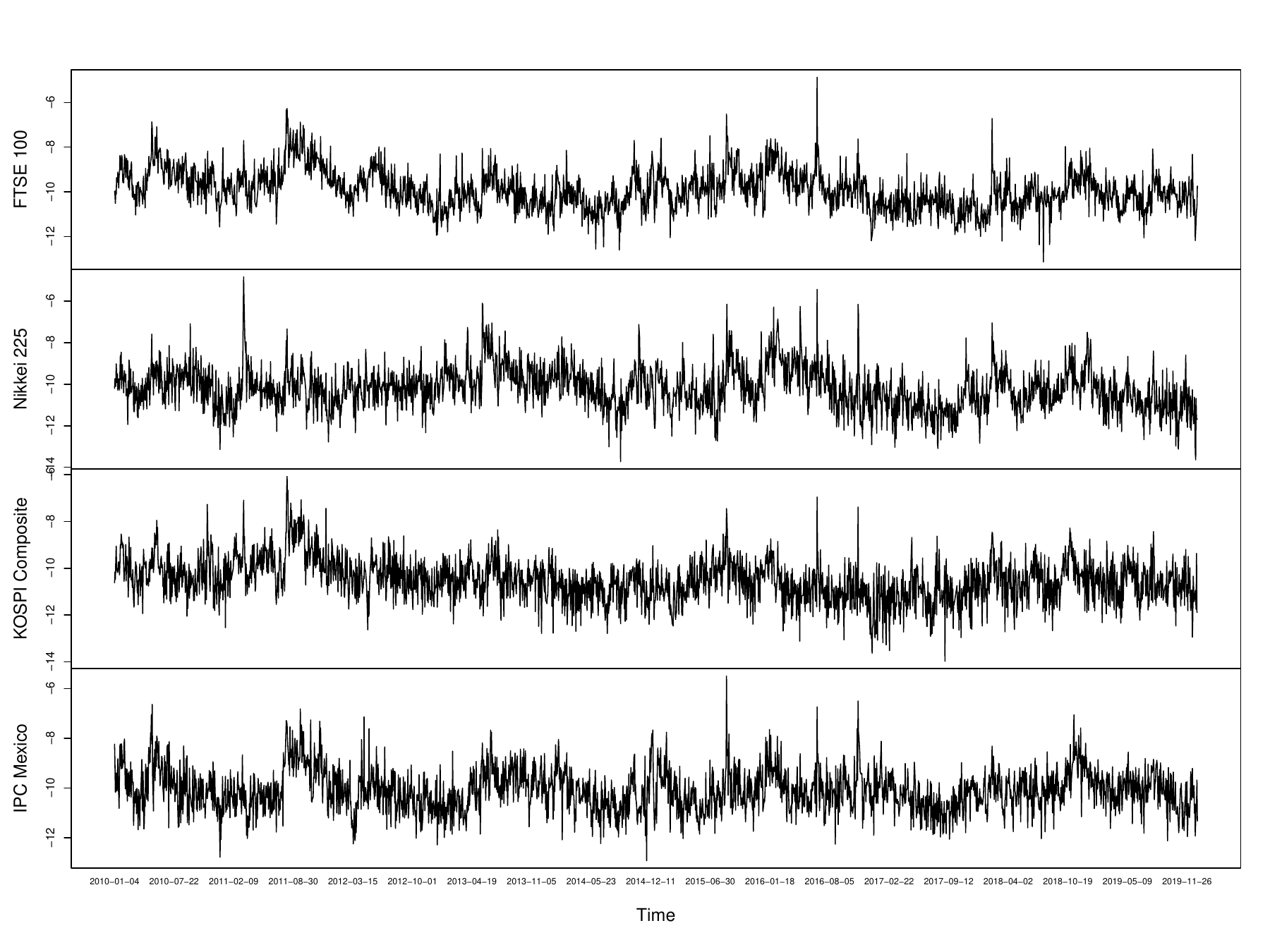}
	\includegraphics[width=0.45\textwidth,height=0.35\textheight]{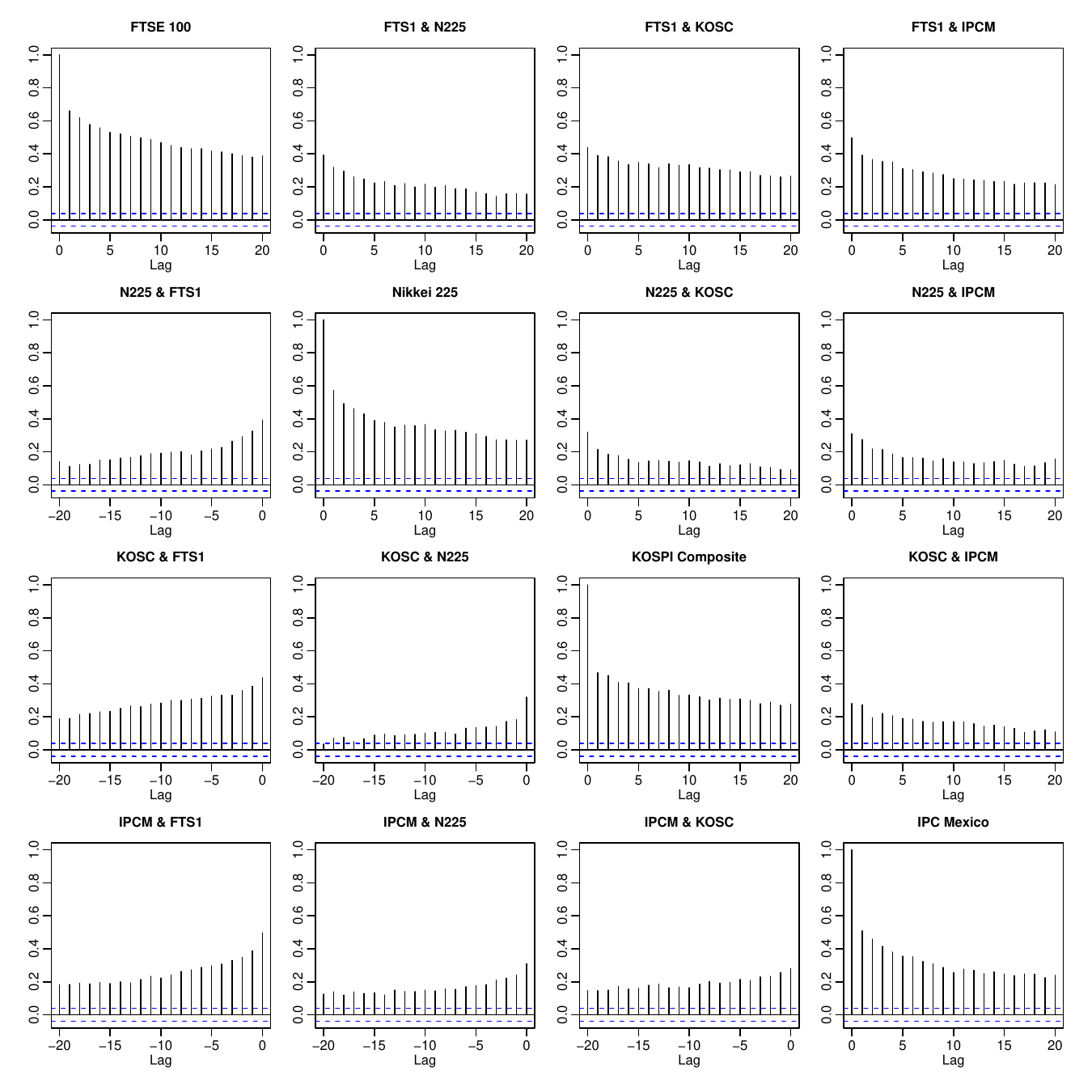}
	\includegraphics[width=0.45\textwidth,height=0.35\textheight]{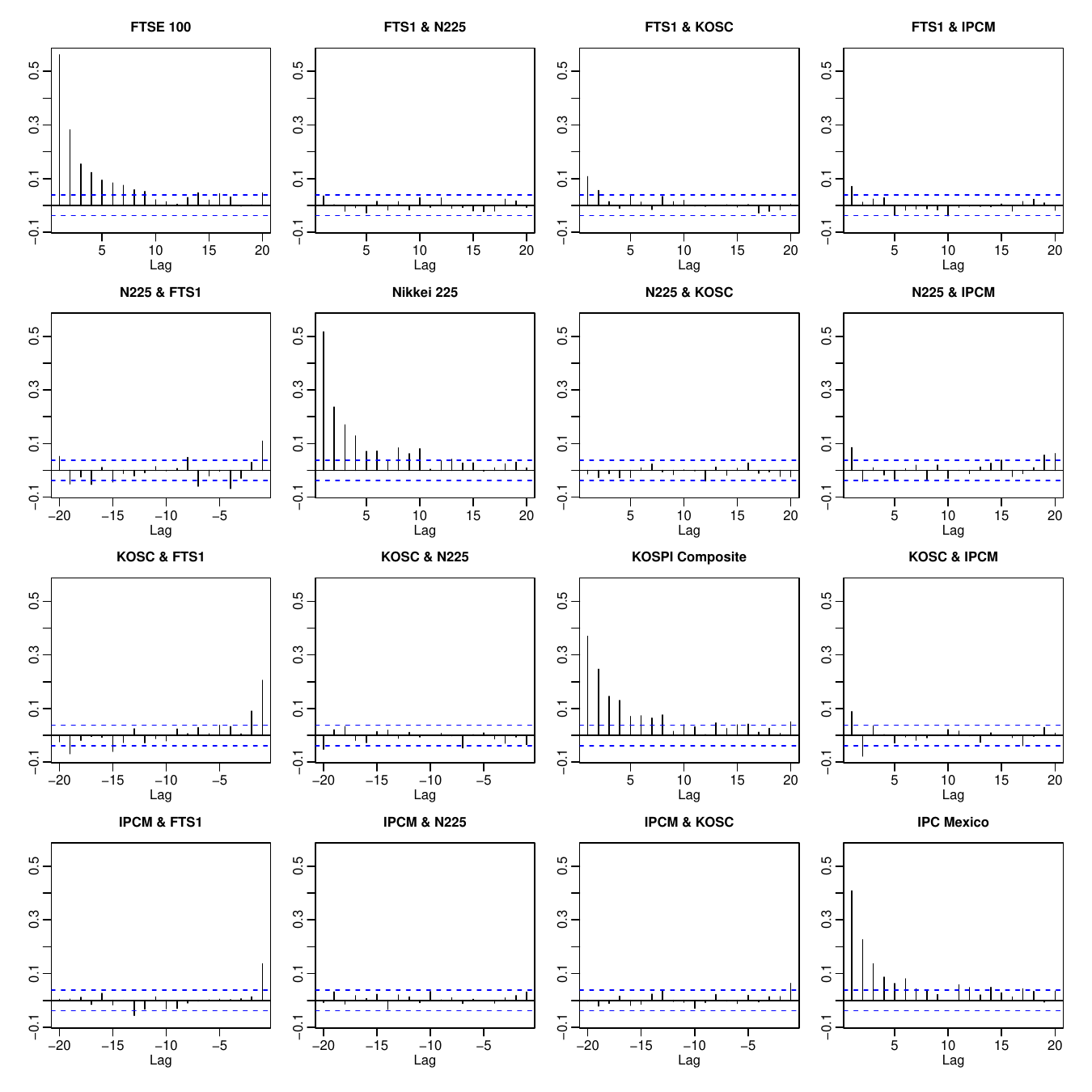}
	\caption{(Top) Time plots of four selected stock indices (FTSE 100, Nikkei 225, KOSPI Composite Index, and IPC Mexico). (Bottom Left) Sample ACF plots of the series. (Bottom Right) Sample PACF plots of the series.}
	\label{fig:timeplot_VHAR}
\end{figure}

\begin{figure}[t!]
	\centering
	\includegraphics[width=1\textwidth]{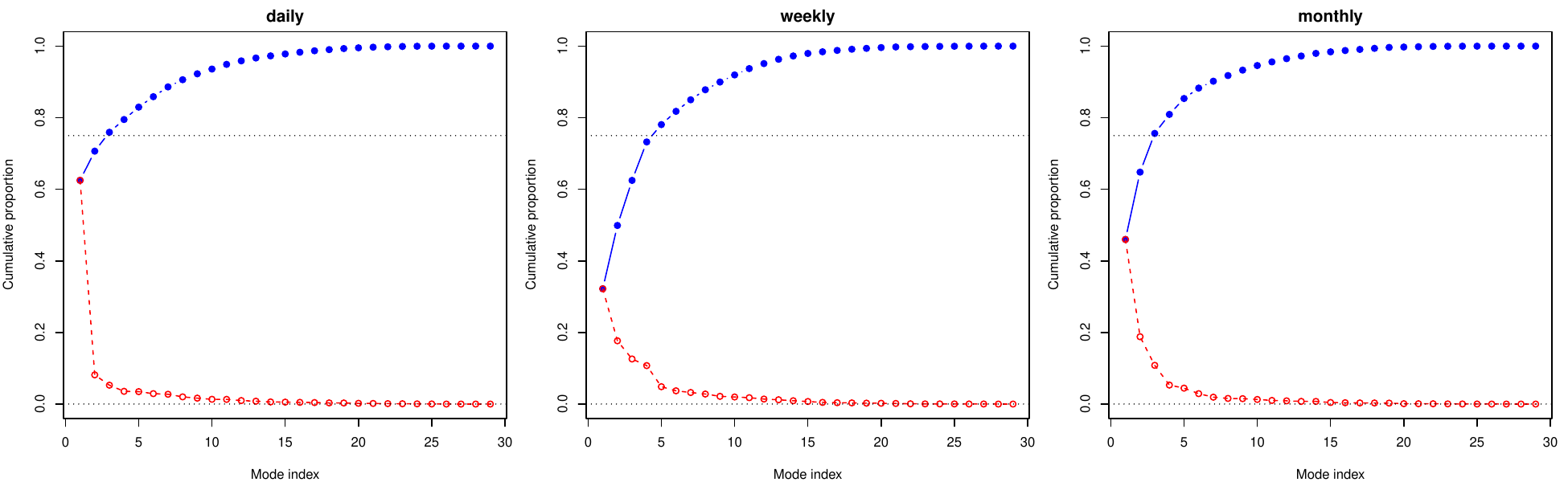}
    \caption{Scree plots of singular values from the estimated transition matrices of the generalized ScBM--VHAR model with $(b_M,b_L)=(5,22)$. The blue solid curves with filled circles show the cumulative proportion of variation explained, while the red dotted curves with open circles show the proportion explained by each singular value. The horizontal black dotted lines mark the 75\% threshold.}
	\label{fig:rank_VHAR}
	\vspace{-1mm}
\end{figure}

Figure~\ref{fig:vhar_sankey_343} allows four middle-horizon communities while retaining three at the long and short horizons.  Relative to the $3\to3\to3$ baseline in Section \ref{sse:appl_vhar}, the long- and short-horizon partitions are unchanged and the middle horizon is refined, supporting the substantive long-to-short interpretation.

\begin{figure}[t!]
\centering
\includegraphics[width=.98\textwidth]{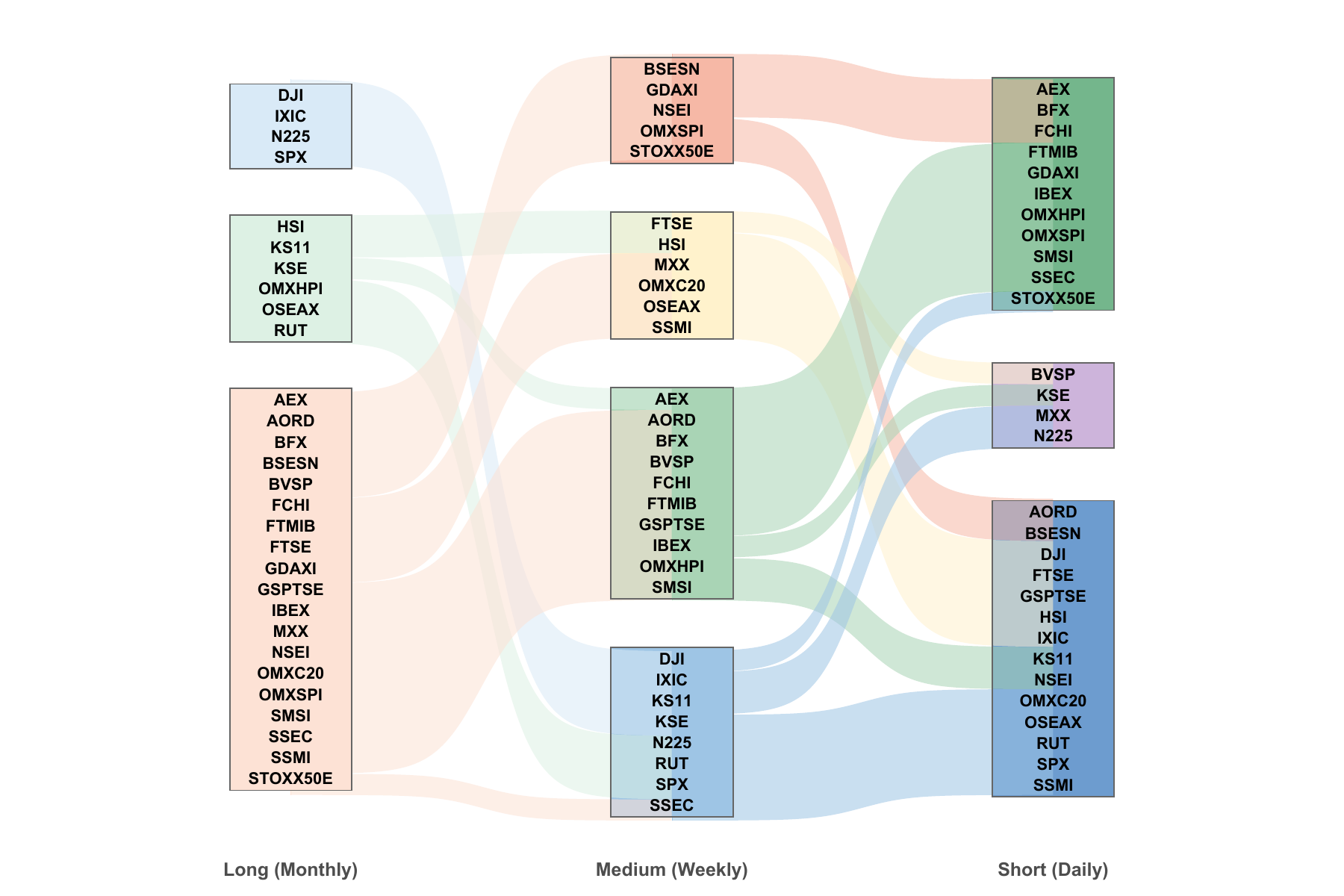}
\caption{VHAR-based long-to-short volatility communities under the alternative $3\to4\to3$ aggregate configuration. The long- and short-horizon partitions are unchanged relative to the $3\to3\to3$ baseline, while the middle horizon is refined into four groups.}
\label{fig:vhar_sankey_343}
\end{figure}

\end{document}